\pdfoutput=1
\documentclass[final,1p,times,nopreprintline]{elsarticle}

\usepackage[shortlabels]{enumitem}
\setlist{topsep=0.25\baselineskip,partopsep=0pt,itemsep=1pt,parsep=0pt}
\usepackage{arydshln} % dashed line in array (used in appendix and section 4)
\usepackage{tikz}
\usetikzlibrary{arrows,positioning}
\tikzset{node/.style={draw, rectangle, rounded corner}}
\tikzset{
    dimdegdet/.style={
        draw, rectangle, rounded corners
    },
}
\usepackage{amsmath,amsfonts,amssymb,amsthm,thmtools}
\declaretheorem[style=plain,parent=section]{definition}
\declaretheorem[sibling=definition]{theorem}
\declaretheorem[sibling=definition]{corollary}
\declaretheorem[sibling=definition]{proposition}
\declaretheorem[sibling=definition]{lemma}
\declaretheorem[sibling=definition]{remark}
\usepackage{algorithm}
\usepackage[noend]{algpseudocode}      %%% this is for the algorithms
\algnewcommand{\algorithmicassumption}{\textbf{Requirement:}}
\algnewcommand{\Assume}{\item[\algorithmicassumption]}
\algrenewcomment[1]{\hfill \(\triangleright\) \emph{\small #1}} % to italicize text in comments
\algnewcommand{\CommentLine}[1]{\(\triangleright\;\;\) \emph{\small #1} \(\;\;\triangleleft\)} % for big comments, on their own line
\algnewcommand{\InlineIf}[2]{% single line if-then
  \algorithmicif\ #1\ \algorithmicthen\ #2}
\algnewcommand{\InlineIfElse}[3]{% single line if-then-else
  \algorithmicif\ #1\ \algorithmicthen\ #2\ \algorithmicelse\ #3}
\usepackage[linkcolor=black,colorlinks=true,citecolor=blue,urlcolor=blue,hypertexnames=false]{hyperref} 
\usepackage[capitalise]{cleveref}
\newcommand{\storeArg}{} % aux, not to be used in document!!

\newcommand{\bigO}[1]{\mathchoice{O\left(#1\right)}{O(#1)}{O(#1)}{O(#1)}} % big O for complexity
\newcommand{\expmm}{\omega} % exponent for the cost of matrix multiplication

\newcommand{\timepm}[1]{\mathchoice{\operatorname{\mathsf{M}}\!\left(#1\right)}{\operatorname{\mathsf{M}}(#1)}{\operatorname{\mathsf{M}}(#1)}{\operatorname{\mathsf{M}}(#1)}} % cost of polynomial multiplication
\newcommand{\timegcd}[1]{\mathchoice{\operatorname{\mathsf{M}}'\!\left(#1\right)}{\operatorname{\mathsf{M}}'(#1)}{\operatorname{\mathsf{M}}'(#1)}{\operatorname{\mathsf{M}}'(#1)}} % cost of polynomial multiplication with one additional log factor
\newcommand{\algoname}[1]{{\normalfont\textsc{#1}}} % typesetting for algo names
\newcommand{\assign}{\leftarrow} % assign value to variable
\newcommand{\ZZ}{\mathbb{Z}} % relative integers
\newcommand{\NN}{\mathbb{Z}_{\ge 0}} %  integers
\newcommand{\ZZp}{\mathbb{Z}_{> 0}} % positive integers
\newcommand{\card}[1]{\# #1}  % cardinality of a set
\newcommand{\var}{x} % default variable for univariate polynomials
\newcommand{\field}{\mathbb{K}} % base field
\newcommand{\polRing}{\field[\var]} % polynomial ring
\newcommand{\rdim}{m} % default row dimension
\newcommand{\cdim}{n} % default column dimension
\newcommand{\rk}{r} % generic symbol for rank
\newcommand{\cork}{k} % generic symbol for corank
\newcommand{\matRing}[1][\rdim]{\renewcommand\storeArg{#1}\matRingAux} % scalar matrix space, 2 opt args
\newcommand{\pmatRing}[1][\rdim]{\renewcommand\storeArg{#1}\pmatRingAux} % polynomial matrix space, 2 opt args
\newcommand{\matRingAux}[1][\storeArg]{\field^{\storeArg \times #1}} % not to be used in document
\newcommand{\pmatRingAux}[1][\storeArg]{\polRing^{\storeArg \times #1}} % not to be used in document
\newcommand{\row}[1]{\mathbf{\MakeLowercase{#1}}} % for a row of a matrix
\newcommand{\mat}[1]{\mathbf{\MakeUppercase{#1}}} % for a matrix
\newcommand{\matt}[1]{\mathbf{\hat{\MakeUppercase{#1}}}} % for a matrix, bis
\newcommand{\matrow}[2]{{#1}_{#2,*}} % for the row #2 of the matrix #1
\newcommand{\matrows}[2]{{#1}_{#2,*}} % for the submatrix of #1 of its rows in #2
\newcommand{\matcol}[2]{{#1}_{*,#2}} % for the column #2 of the matrix #1
\newcommand{\matcols}[2]{{#1}_{*,#2}} % for the submatrix of #1 of its columns in #2
\newcommand{\submat}[3]{{#1}_{#2,#3}} % submatrix of #1 of its rows in #2 and columns in #3
\newcommand{\trsp}[1]{#1^\mathsf{T}} % transpose
\newcommand{\diag}[1]{\mathrm{diag}(#1)}  % diagonal matrix with diagonal entries #1
\newcommand{\xDiag}[1]{\mat{\var}^{#1\,}} % shift matrix, diagonal of powers of X with exponents given by #1
\newcommand{\idMat}[1]{\mat{I}_{#1}} % identity matrix of size #1 x #1
\newcommand{\revmat}[1]{\mat{J}_{#1}} % anti-diagonal matrix with entries=1, of size #1 x #1
\newcommand{\matz}{\mat{0}}  % zero matrix
\newcommand{\anyMat}{\boldsymbol{\ast}}  % indicate block in a matrix : bold asterisk
\newcommand{\sddots}{\raisebox{3pt}{$\scalebox{.75}{$\ddots$}$}} % to scale ddots, e.g. in smallmatrix
\newcommand{\svdots}{\raisebox{3pt}{$\scalebox{.75}{$\vdots$}$}} % to scale vdots, e.g. in smallmatrix
\newcommand{\rank}[1]{\mathrm{rank}(#1)}
\newcommand{\tuple}[1]{\boldsymbol{#1}}  % tuples (mainly used for integers I think)
\newcommand{\sumTuple}[1]{|#1|} % sum of entries in a tuple
\newcommand{\rdeg}[2][]{\mathrm{rdeg}_{{#1}}(#2)} % shifted row degree
\newcommand{\cdeg}[2][]{\mathrm{cdeg}_{{#1}}(#2)} % shifted column degree
\newcommand{\lmat}[2][]{\mathrm{lm}_{#1}(#2)} % leading matrix of polynomial matrix, default shifts = 0
\newcommand{\lmatBig}[2][]{\mathrm{lm}_{#1}\left(#2\right)} % leading matrix of polynomial matrix, default shifts = 0
\newcommand{\shiftz}{\mathbf{0}} % notation for the zero shift ~ [0,..,0]
\newcommand{\shift}{s} % shifts letter, main
\newcommand{\shifts}{\tuple{s}} % shifts vector, main
\newcommand{\shiftt}{\tuple{t}} % shifts vector, secondary
\newcommand{\pivDeg}{\delta} % entry of the minimal degree or pivot degree
\newcommand{\pivDegs}{\boldsymbol{\pivDeg}} % minimal degree
\newcommand{\pivInd}{\pi} % entry of the pivot index
\newcommand{\pivInds}{\boldsymbol{\pi}} % pivot index
\newcommand{\degdet}{D} % degree of determinant
\newcommand{\subTuple}[2]{{#1}_{#2}} % subtuple
\newcommand{\pvec}{\row{p}} % some polynomial vector
\newcommand{\pmat}{\mat{A}} % some polynomial matrix
\newcommand{\sys}{\mat{F}} % ``system'' matrix for approximant
\newcommand{\app}{\row{p}} % one relation (approximant)
\newcommand{\appbas}{\mat{M}} % basis of approximants
\newcommand{\order}{\gamma} % individual order for approximation
\newcommand{\orders}{\boldsymbol{\gamma}} % order for approximation (list of individual ones)
\newcommand{\cdegs}{\tuple{d}} % vector of column degrees
\newcommand{\modApp}[2]{\operatorname{\mathcal{A}}_{#1}(#2)} % approximant module
\newcommand{\modKer}[1]{\operatorname{\mathcal{K}}(#1)} % kernel module
\newcommand{\degExp}{ \delta }  % linearization degree parameter
\newcommand{\expandMat}{\mat{E}}  % expansion-compression matrix
\newcommand{\quoExp}{\alpha}  % number of columns in the expanded block
\newcommand{\remExp}{\beta}  % residual degree after splitting in blocks of size \degExp
\newcommand{\expand}[1]{\overline{#1}}  % format for expanded stuff
\newcommand{\kerExpMat}{\mat{S}} % left kernel of ExpandMat
\newcommand{\rowbas}{\mat{B}} % row basis for a polynomial matrix
\newcommand{\kerbas}{\mat{K}} % general kernel basis for a polynomial matrix
\newcommand{\hypsl}{\mathcal{H}_{\mathrm{sl}}} % superlinearity assumption
\newcommand{\hypsm}{\mathcal{H}_{\mathrm{sm}}} % submultiplicativity assumption
\newcommand{\hypeps}{\mathcal{H}_{\mathrm{\expmm}}} % ??? assumption
\begin{document}

\begin{frontmatter}

\title{Deterministic computation of the characteristic polynomial \\ in the time of matrix multiplication}

\author{Vincent Neiger}
\address{Univ.~Limoges, CNRS, XLIM, UMR 7252, F-87000 Limoges, France}

\author{Cl\'ement Pernet}
\address{Universit\'e Grenoble Alpes, Laboratoire Jean Kuntzmann, CNRS, UMR 5224 \\
700 avenue centrale, IMAG - CS 40700, 38058 Grenoble cedex 9, France}

\begin{abstract}
  This paper describes an algorithm which computes the characteristic
  polynomial of a matrix over a field within the same asymptotic complexity, up
  to constant factors, as the multiplication of two square matrices.
  Previously, this was only achieved by resorting to
  genericity assumptions or randomization techniques, while the best known
  complexity bound with a general deterministic algorithm was obtained by
  Keller-Gehrig in 1985 and involves logarithmic factors. Our algorithm
  computes more generally the determinant of a univariate polynomial matrix in
  reduced form, and relies on new subroutines for transforming shifted reduced
  matrices into shifted weak Popov matrices, and shifted weak Popov matrices
  into shifted Popov matrices.
\end{abstract}


\end{frontmatter}

\section{Introduction}
\label{sec:intro}

The last five decades witnessed a constant effort towards computational
reductions of linear algebra problems to matrix multiplication. It has been
showed that most classical problems are not harder than multiplying two square
matrices, such as matrix inversion, LU decomposition, nullspace basis
computation, linear system solving, rank and determinant computations, etc.
\cite{BunHop74} \cite{IbaMorHui82}
\cite[Chap.\,16]{BurgisserClausenShokrollahi2010}. In this context, one major
challenge stands out: designing a similar reduction to matrix multiplication
for the computation of characteristic polynomials and related objects such as
minimal polynomials and Frobenius forms. For the characteristic polynomial,
significant progress was achieved by Keller-Gehrig \cite{KelGeh85}, and more
recently by Pernet and Storjohann \cite{PerSto07} who solved the problem if one
allows randomization. This paper closes the problem by providing a
deterministic algorithm with the same asymptotic complexity as matrix
multiplication.

The characteristic polynomial of a square matrix over a field \(\field\), say
\(\mat{M} \in \matRing[\rdim]\), is defined as \(\det(x\idMat{\rdim} -
\mat{M})\). Specific algorithms exist for sparse or structured matrices; here
we consider the classical, dense case. In this paper the complexity of an algorithm
is measured as an upper bound on its arithmetic cost, that is, the number of
basic field operations it uses to compute the output.

\begin{theorem}
  \label{thm:charpoly}
  Let \(\field\) be a field. Using a subroutine which multiplies two matrices
  in \(\matRing[\rdim]\) in \(\bigO{\rdim^\expmm}\) field operations for some
  \(\expmm>2\), the characteristic polynomial of a matrix in
  \(\matRing[\rdim]\) can be computed deterministically in
  \(\bigO{\rdim^\expmm}\) field operations.
\end{theorem}

\paragraph{Outline}

The rest of this introduction gives more details about our framework for
complexity bounds (\cref{sec:intro:cost_bounds}), summarizes previous work
(\cref{sec:intro:previous}), describes our contribution on polynomial matrix
determinant computation (\cref{sec:intro:determinant}), gives an overview of
our approach and of new tools that we designed to avoid logarithmic factors
(\cref{sec:intro:approach,sec:intro:new_tools}), and finally lists a few
perspectives (\cref{sec:intro:perspectives}). \cref{sec:prelim} introduces the
notation, main definitions, and basic properties used in this paper. Then
\cref{sec:determinant} presents the main algorithm of this paper along with a
detailed complexity analysis. This algorithm uses two new technical tools
described in \cref{sec:reduced_to_weak_popov,sec:weak_popov_to_popov}: the
transformation of reduced forms into weak Popov forms and of weak Popov forms
into Popov forms, in the case of shifted forms.

\subsection{Framework for complexity bounds}
\label{sec:intro:cost_bounds}

In this paper, \(\field\) is any field and we seek upper bounds on the
complexity of algorithms which operate on objects such as matrices and
polynomials over \(\field\). We consider the arithmetic cost of these
algorithms, i.e.~the number of basic operations in \(\field\) that are used to
compute the output from some input of a given size. The basic operations are
addition, subtraction, multiplication, and inversion in the field, as well as
testing whether a given field element is zero.

As already highlighted in \cref{thm:charpoly}, in this paper we fix any \(2 <
\expmm \le 3\) as well as any algorithm which multiplies matrices in
\(\matRing[\rdim]\) using \(\bigO{\rdim^\expmm}\) operations in \(\field\):
this algorithm is assumed to be the one used as a black box for all matrix
multiplications arising in the algorithms we design. The current best known
cost bounds ensure that any \(\expmm > 2.373\) is feasible \cite{LeGall14}. In
practice, one often considers a cubic algorithm with \(\expmm=3\) or Strassen's
algorithm with \(\expmm = \log_2(7)\) \cite{Strassen1969}. Our results hold
with the only assumption that \(2 < \expmm \le 3\).
%% \(\expmm \le 3\) is probably not required, but who cares?

In the computer algebra literature, this setting is classical and often
implicit;
we still emphasize it because here, and more generally when one studies the
logarithmic factors in the cost bound of some algorithm, this clarification of
how the underlying matrix multiplications are performed is of the utmost
importance. Indeed, if one were allowed to use any matrix multiplication
subroutine, then the question of logarithmic factors becomes void: for any
exponent \(\expmm\) known to be feasible at the time of writing, it is known
that \(\expmm-\varepsilon\) is feasible as well for a sufficiently small
\(\varepsilon>0\); then one might rather rely on this faster subroutine, and
apply Keller-Gehrig's algorithm to obtain the characteristic polynomial in
\(\bigO{\rdim^{\expmm-\varepsilon} \log(\rdim)}\) operations in \(\field\),
which is in \(\bigO{\rdim^\expmm}\).

Similarly, we consider a nondecreasing function \(d\mapsto \timepm{d}\) and an
algorithm which multiplies two polynomials in \(\polRing\) of degree at most
\(d\) using at most \(\timepm{d}\) operations in \(\field\); our algorithms
rely on this subroutine for polynomial multiplication. Here \(d\) is any
nonnegative real number; it will often be a fraction \(\degdet/\rdim\) of
positive integers; we assume that \(\timepm{d} = 1\) for \(0 \le d < 1\), so
that \(\timepm{d} \ge 1\) for all \(d\ge 0\). To help derive complexity upper
bounds, we also consider the following assumptions \(\hypsl\), \(\hypsm\),
and \(\hypeps\).
%% Note: submultiplicativity and M(d) >= 1 implies nondecreasing, but since we
%% may sometimes not assume submultiplicativity, we state nondecreasing-ness
%% explicitly.

\begin{description}
  \item[\hspace{1em}\rm \(\hypsl\):]
    \(2\timepm{d} \le \timepm{2d}\) for all \(d\ge 1\)
    (superlinearity).
  \item[\hspace{1em}\rm \(\hypsm\):]
   \(\timepm{d_1 d_2} \le \timepm{d_1} \timepm{d_2}\) for all \(d_1,d_2 \ge 0\)
   (submultiplicativity).
  \item[\hspace{1em}\rm \(\hypeps\):]
    \(\timepm{d} \in \bigO{d^{\expmm-1-\epsilon}}\) for some \(\epsilon>0\).
\end{description}
The first assumption is customary, see e.g.~\cite[Sec.\,8.3]{vzGathen13}; note
that it implies \(\timepm{d} \ge d\) for all \(d\ge 1\). The second and last
assumptions are commonly made in complexity analyses for divide and conquer
algorithms on polynomial matrices \cite{Storjohann03,GuSaStVa12}: we refer to
\cite[Sec.\,2]{Storjohann03} for further comments on these assumptions. They
are satisfied by the cost bounds of polynomial multiplication
algorithms such as the quasi-linear algorithm of Cantor and Kaltofen
\cite{CanKal91} and, for suitable fields \(\field\), the quasi-linear algorithm
of Harvey and van der Hoeven and Lecerf \cite{HarveyVDHoevenLecerf2017}, and
most of Toom-Cook subquadratic algorithms \cite{Toom1963,Cook1966}. For the
latter only \(\hypeps\) might not be satisfied, depending on \(\expmm\) and on
the number of points used. Note that with the current estimates having \(\omega
> 2.373\), an order 5 Toom-Cook multiplication (requiring a field with at least
9 points) has exponent \(\log(9)/\log(5) \approx 1.365 < \expmm-1\); thus for
such exponents \(\expmm\) all Toom-Cook algorithms of order 5 or more
satisfy all the above assumptions.

Following \cite{Storjohann03,GuSaStVa12}, we also define a function \(d \mapsto
\timegcd{d}\) related to the cost of divide and conquer methods such as the
half-gcd algorithm: \(\timegcd{d} = \sum_{0 \le i \le \lceil \log_2(d) \rceil}
2^i \timepm{2^{-i}d}\) for \(d \ge 1\), and \(\timegcd{d} = 1\) for \(0\le d
\le 1\). By definition one has \(\timegcd{d} \ge \timepm{d} \ge 1\) for all
\(d\ge 0\), and the identity \(\timegcd{2d} = 2 \timegcd{d} + \timepm{2d}\) for
\(d\ge 1\) ensures that \(\timegcd{d}\) is superlinear: \(2\timegcd{d} \le
\timegcd{2d}\) for all \(d\ge 1\). Assuming \(\hypsl\) yields the asymptotic
bound \(\timegcd{d} \in \bigO{\timepm{d}  \log(d)}\) where the \(\log(d)\) factor
only occurs if a quasi-linear polynomial multiplication is used;
in particular, \(\hypsl\) and \(\hypeps\) imply \(\timegcd{d} \in
\bigO{d^{\expmm-1-\varepsilon}}\) for some \(\varepsilon>0\).  Furthermore, if
one assumes \(\hypsl\) and \(\hypsm\), then \(\timegcd{\cdot}\) is
submultiplicative as well: \(\timegcd{d_1 d_2} \le \timegcd{d_1}
\timegcd{d_2}\) for all \(d_1,d_2 \ge 0\).

In what follows we assume that two polynomial matrices in \(\pmatRing[\rdim]\)
of degree at most \(d\ge 0\) can be multiplied in \(\bigO{\rdim^\expmm
\timepm{d}}\) operations in \(\field\). This is a very mild assumption: it
holds as soon as \(\timepm{d}\) corresponds to one of the above-mentioned
polynomial multiplication algorithms, and it also holds if the chosen matrix
multiplication algorithm defining \(\expmm\) supports matrices over a
commutative ring using only the operations \(\{+,-,\times\}\) (so that one can
use it to multiply \(\rdim\times\rdim\) matrices over
\(\polRing/(\var^{2d+1})\)). Note still that this bound \(\bigO{\rdim^\expmm
\timepm{d}}\) is slightly worse than the best known ones
\cite{CanKal91,HarveyVDHoevenLecerf2017}; for example, Cantor and Kaltofen's
algorithm performs polynomial matrix multiplication in \(\bigO{\rdim^\expmm d
\log(d) + \rdim^2 d \log(d) \log(\log(d))}\) field operations, which is finer
than the bound \(\bigO{\rdim^\expmm \timepm{d}}\) with \(\timepm{d} = \Theta(d
\log(d) \log(\log(d)))\) in that case. This simplification is frequent in the
polynomial matrix literature, and it is made here for the sake of presentation,
to improve the clarity of our main complexity results and of the analyses that
lead to them.

\subsection{Previous work}
\label{sec:intro:previous}

Previous algorithms based on linear algebra over \(\field\) for computing the
characteristic polynomial of \(\mat{M} \in \matRing[\rdim]\) mainly fall in three
types of methods.
\begin{description}
  \item[\hspace{1em}\it Traces of powers:]
    combining the traces of the first \(n\) powers of the input matrix using
    the Newton identities reveals the coefficients of the characteristic
    polynomial. Known as the Faddeev-LeVerrier algorithm, it was introduced in
    \cite{Leverrier1840}, refined and rediscovered in
    \cite{Souriau48,FaSo49,Frame49}, and used in \cite{Csanky75} to prove that
    the problem is in the \(\mathcal{NC}^2\) parallel complexity class.
  \item[\hspace{1em}\it Determinant expansion formula:]
    introduced in \cite{Samuelson42} and improved in \cite{Berkowitz84}, this
    approach does not involve division, and is therefore well suited for
    computing over integral domains. Later developments in this field include
    \cite{AbMa01,KaVi05}, the latter reaching the best known cost bound of \(\bigO{\rdim^{2.6973}\log(\rdim)^c}\)
    ring operations using a deterministic algorithm, for some constant~\(c>0\).
  \item[\hspace{1em}\it Krylov methods:]
    based on sequences of iterates of vectors under the application of the
    matrix: \((\row{v}, \mat{M}\row{v}, \mat{M}^2 \row{v}, \ldots)\). These methods
    rely on the fact that the first linear dependency between these iterates
    defines a polynomial which divides the characteristic polynomial. Some
    algorithms construct the Krylov basis explicitly \cite{KelGeh85,
      Giesbrecht95, DPW05}, while others can be interpreted as an implicit
      Krylov iteration with structured vectors \cite{Danilevskij1937,
      PerSto07}.
\end{description}
Methods based on traces of powers use \(\bigO{\rdim^4}\) or
\(\bigO{\rdim^{\omega+1}}\) field operations, and are mostly competitive for
their parallel complexity. Methods based on determinant expansions use
\(\bigO{\rdim^4}\) or \(\bigO{\rdim^{\omega+1}}\) field operations and are
relevant for division-free algorithms. Lastly, the Krylov methods run in
\(\bigO{\rdim^3}\) \cite{Danilevskij1937, DPW05} or \(\bigO{\rdim^{\expmm}\log
\rdim}\) \cite{KelGeh85} field operations with deterministic algorithms, or in
\(\bigO{\rdim^\expmm}\) field operations with the Las Vegas randomized
algorithm in \cite{PerSto07}.

Note that the characteristic polynomial of \(\mat{M}\) cannot be computed
faster than the determinant of \(\mat{M}\), since the latter is the constant
coefficient of the former. Furthermore, under the model of computation trees,
the determinant of \(\rdim\times\rdim\) matrices cannot be computed faster than
the product of two \(\rdim\times\rdim\) matrices
\cite[Sec.\,16.4]{BurgisserClausenShokrollahi2010}, a consequence of Baur and
Strassen's theorem \cite{BaurStrassen83}.

Another type of characteristic polynomial algorithms is based on operations on
matrices over \(\polRing\), called polynomial matrices in what follows. Indeed
the characteristic polynomial may be obtained by calling a determinant
algorithm on the characteristic matrix \(x\idMat{\rdim} - \mat{M}\), which is
in \(\pmatRing[\rdim]\). Existing algorithms, which accept any matrix in
\(\pmatRing[\rdim]\) of degree \(d\) as input, include
\begin{itemize}
  \item the evaluation-interpolation method, which costs
    \(\bigO{\rdim^{\expmm+1} d + \rdim^3 \timegcd{d}}\) field operations,
    requires that the field \(\field\) is large enough, and mainly relies on
    the computation of about \(\rdim d\) determinants of matrices in
    \(\matRing[\rdim]\);
  \item the algorithm of Mulders and Storjohann \cite{MulSto03} based on weak
    Popov form computation, which uses \(\bigO{\rdim^{3} d^2}\) field
    operations;
  \item retrieving the determinant as the product of the diagonal entries of
    the Smith form, itself computed by a Las Vegas randomized algorithm in
    \(\bigO{\rdim^{\expmm} \timegcd{d} \log(\rdim)^2}\) field operations
    \cite[Prop.\,41]{Storjohann03}, assuming \(\timegcd{d} \in
    \bigO{d^{\expmm-1}}\);
  \item the algorithm based on unimodular triangularization in
    \cite{LaNeZh17}, which is deterministic and uses \(\bigO{\rdim^\expmm d
    \log(d)^a \log(\rdim)^b}\) field operations for some constants \(a,b \in
    \ZZp\).
\end{itemize}
In the last two items the cost bound is, up to logarithmic factors, the same as
the cost of multiplying matrices \(\pmatRing[\rdim]\) of degree \(d\) by
relying on both fast linear algebra over \(\field\) and fast arithmetic in
\(\polRing\), as showed in \cite{CanKal91}. The last two of these cost bounds
do involve factors logarithmic in \(\rdim\), whereas the first two have an
exponent on \(\rdim\) which exceeds \(\expmm\).

In summary, the fastest characteristic polynomial algorithms either are
randomized or have a cost a logarithmic factor away from the lower bound. This
paper, with \cref{thm:charpoly}, bridges this gap by proposing the first
deterministic algorithm with cost \(\bigO{\rdim^\expmm}\).

\subsection{A more general result: determinant of reduced polynomial matrices}
\label{sec:intro:determinant}

Our algorithm falls within the category of polynomial matrix determinant
computation. Yet unlike the above-listed approaches ours is tailored to a
specific family of polynomial matrices, which contains the characteristic
matrix \(x\idMat{\rdim} - \mat{M}\): the family of \emph{row reduced matrices}
\cite{Wolovich74,Kailath80}. Restricting to such matrices provides us with good
control of the degrees in computations; as a typical example, it is easy to
predict the degree of a vector-matrix product \(\row{v} (x\idMat{\rdim} -
\mat{M})\) by observing the degrees in \(\row{v}\), without actually computing
the product. As we explain below, this degree control allows us to avoid
searches of degree profiles, which would add logarithmic terms to the cost.
Although the characteristic matrix has other properties besides row reducedness
(it has degree \(1\), and is in Popov form \cite{Popov72} hence column
reduced), we do not exploit them.

When appropriate, the average row degree \(\degdet/\rdim\), where \(\degdet\)
is the sum of the degrees of the rows of the matrix, is chosen as a measure of
the input degree which refines the matrix degree \(d\) used above.  This gives
cost bounds more sensitive to the input degrees and also, most importantly,
leverages the fact that even if the algorithm starts from a matrix with uniform
degrees such as \(x\idMat{\rdim} - \mat{M}\), it may end up handling matrices
with unbalanced row degrees in the process.

\begin{theorem}
  \label{thm:polmat_det}
  Assuming \(\hypsl\), \(\hypsm\), and \(\hypeps\) (hence in particular
  \(\expmm>2\)), there is an algorithm which takes as input a row reduced
  matrix \(\pmat \in \pmatRing[\rdim]\) and computes its determinant using
  \[
    \bigO{\rdim^\expmm \timegcd{\degdet/\rdim}}
    \subseteq \bigO{\rdim^\expmm \timegcd{\deg(\pmat)}}
  \]
  operations in \(\field\), where \(\degdet = \deg(\det(\pmat))\) is equal to
  the sum of the degrees of the rows of \(\pmat\).
\end{theorem}
\noindent The fact that \(\deg(\det(\pmat))\) is the sum of row degrees is a
consequence of row reducedness \cite{Kailath80}, and the cost bound inclusion
follows from \(\deg(\det(\pmat)) \le \rdim\deg(\pmat)\). Taking \(\pmat = x
\idMat{\rdim} - \mat{M}\) for \(\mat{M} \in \matRing[\rdim]\),
\cref{thm:charpoly} is a direct corollary of \cref{thm:polmat_det}. The only
assumption needed in \cref{thm:charpoly} is \(\expmm>2\), since it implies the
existence of a polynomial multiplication algorithm such that \(\hypsl\),
\(\hypsm\), and \(\hypeps\) hold, such as Cantor and Kaltofen's algorithm
\cite{CanKal91}.

Previous polynomial matrix determinant algorithms with costs of the order of
\(\rdim^\expmm \deg(\pmat)\), up to logarithmic factors, have been listed
above: a randomized one from \cite{Storjohann03}, and a deterministic one from
\cite{LaNeZh17}. To our knowledge, this paper gives the first description of an
algorithm achieving such a cost involving no factor logarithmic in \(\rdim\).
Our approach partially follows the algorithm of \cite{LaNeZh17}, but also
substantially differs from it in a way that allows us to benefit from the
reducedness of \(\pmat\). The cost bound \(\bigO{\rdim^\expmm
\timegcd{\deg(\pmat)}}\) has been obtained before in
\cite[Sec.\,4.2.2]{GiJeVi03} in the particular case of a ``sufficiently
generic\footnote{Precisely, if the upper triangular, row-wise Hermite normal
  form of \(\pmat\) has diagonal entries \((1,\ldots,1,\lambda\det(\pmat))\),
  for some \(\lambda \in \field\setminus\{0\}\) making \(\lambda\det(\pmat)\)
monic.}'' matrix \(\pmat\). In that case, both the algorithm of \cite{LaNeZh17}
and the one here coincide and become the algorithm of
\cite[Sec.\,4.2.2]{GiJeVi03}; when \(\pmat\) is the characteristic matrix
\(x\idMat{\rdim} - \mat{M}\), this also relates to the fast algorithm
in \cite[Sec.\,6]{KelGeh85} for a generic \(\mat{M}\).

\subsection{Approach, and existing tools}
\label{sec:intro:approach}

For the sake of presentation, suppose \(\rdim\) is a power of \(2\). Writing
\(\pmat = [\begin{smallmatrix} \pmat_1 & \pmat_2 \\ \pmat_3 & \pmat_4
\end{smallmatrix}]\) with the \(\pmat_i\)'s of dimensions
\((\rdim/2) \times (\rdim/2)\), the algorithm of \cite{LaNeZh17} is based on
the block triangularization
\begin{equation*}
  \begin{bmatrix}
    \anyMat & \anyMat \\
    \kerbas_1 & \kerbas_2
  \end{bmatrix}
  \begin{bmatrix}
    \pmat_1 & \pmat_2 \\
    \pmat_3 & \pmat_4
  \end{bmatrix}
  =
  \begin{bmatrix}
    \mat{R} & \anyMat \\
    \matz & \mat{B}
  \end{bmatrix}
\end{equation*}
where the entries ``\(\anyMat\)'' are not computed, \(\mat{B} = \kerbas_1
\pmat_2 + \kerbas_2 \pmat_4\), and \(\mat{R}\) and \([\kerbas_1 \;\;
\kerbas_2]\) are computed from \([\begin{smallmatrix} \pmat_1 \\ \pmat_3
\end{smallmatrix}]\) as a \emph{row basis} and a \emph{kernel basis},
respectively (see \cref{sec:prelim:app_ker} for definitions). Then the leftmost
matrix in the above identity is unimodular~\cite[Lem.\,3.1]{LaNeZh17} and thus, up to a constant factor,
\(\det(\pmat)\) can be computed recursively as \(\det(\mat{R})\det(\mat{B})\).

A first observation is that neither the kernel basis computation nor the matrix
multiplication giving \(\mat{B}\) is an obstacle towards a cost which is free
of \(\log(\rdim)\). (The fastest known method for multiplying matrices with
unbalanced degrees, such as in \(\mat{B} = \kerbas_1 \pmat_2 + \kerbas_2
\pmat_4\), splits the computation into \(\bigO{\log(\rdim)}\) multiplications
of smaller matrices with balanced degrees \cite[Sec.\,3.6]{ZhLaSt12},
suggesting that its cost may involve a \(\log(\rdim)\) factor.) Indeed we show
that, under the above assumptions on \(\timepm{\cdot}\), the cost of these
operations is in \(\bigO{\rdim^\expmm \timepm{\degdet/\rdim}}\) and
\(\bigO{\rdim^\expmm \timegcd{\degdet/\rdim}}\), thus only involving factors
logarithmic in \(\degdet/\rdim\). In previous work, cost bounds either hide
logarithmic factors \cite{ZhLaSt12} or they are derived without assuming
\(\hypsm\) and have the form \(\bigO{\rdim^{\expmm-1} \timepm{\degdet}}\) and
\(\bigO{\rdim^{\expmm-1} \timegcd{\degdet}}\) \cite{JeNeScVi17}, thus resulting
in factors logarithmic in \(\degdet\).  Proving this observation is
straightforward from the analyses in \cite{ZhLaSt12,JeNeScVi17} (see
\cref{sec:prelim:subroutines}). This is a first key towards our main result:
the characteristic matrix has \(\degdet=\rdim\), and \(\bigO{\rdim^\expmm
\timepm{1}}\) is the same as \(\bigO{\rdim^\expmm}\) whereas
\(\bigO{\rdim^{\expmm-1} \timepm{\rdim}}\) involves factors logarithmic in
\(\rdim\).

However, the computation of the row basis \(\mat{R}\) remains an obstacle which
prevents the algorithm of \cite{LaNeZh17} from being a candidate for
\cref{thm:polmat_det}. Indeed, among the row basis algorithms we are aware of,
only one has a cost bound which fits into our target up to logarithmic factors:
the one of \cite{ZhoLab13}. It relies on three kernel bases computations, and
while one of them is similar to the computation of \([\kerbas_1 \;\;
\kerbas_2]\) and is handled via the algorithm of \cite{ZhLaSt12}, the two
others have different constraints on the input and were the subject of a
specific algorithm described in \cite[Sec.\,4]{ZhoLab13}. In this reference,
cost bounds were given without showing logarithmic factors; our own analysis
revealed the presence of a factor logarithmic in \(\rdim\). The algorithm has a
loop over \(\Theta(\log(\rdim))\) iterations, each of them calling
\cite[Algo.\,2]{ZhoLab12} for minimal approximant bases with unbalanced input.
This approximant basis algorithm may spend a logarithmic number of iterations for finding some
degree profile of the output basis, in a way reminiscent of Keller-Gehrig's
algorithm in \cite[Sec.\,5]{KelGeh85} which finds the lengths of Krylov
sequences (the link between the two situations becomes more explicit for
approximant bases at small orders, see \cite[Sec.\,7]{JeNeScVi17}).

Our attempts at accelerating the row basis algorithm of \cite{ZhoLab13} having
not succeeded, the algorithm in this paper follows an approach which is more
direct at first: remove the obstacle. Instead of computing a row basis
\(\mat{R}\) and relying on the identity \(\det(\pmat) =
\det(\mat{R})\det(\mat{B})\) (up to a constant), keep the first block row of
\(\pmat\):
\begin{equation}
  \label{eqn:nonsingular_triangularization}
  \begin{bmatrix}
    \idMat{\rdim/2} & \matz \\
    \kerbas_1 & \kerbas_2
  \end{bmatrix}
  \begin{bmatrix}
    \pmat_1 & \pmat_2 \\
    \pmat_3 & \pmat_4
  \end{bmatrix}
  =
  \begin{bmatrix}
    \pmat_1 & \pmat_2 \\
    \matz & \mat{B}
  \end{bmatrix}
\end{equation}
and rely on the identity \(\det(\pmat) = \det(\pmat_1) \det(\mat{B}) /
\det(\kerbas_2)\). The nonsingularity of \(\pmat_1\) and \(\kerbas_2\) is
easily ensured thanks to the assumption that \(\pmat\) is reduced, as discussed
in \cref{sec:intro:new_tools}.

This leads to an unusual recursion scheme: we are not aware of a similar scheme
being used in the literature on computational linear algebra. The algorithm
uses three recursive calls with \((\rdim/2) \times (\rdim/2)\) matrices whose
determinant has degree at most \(\degdet/2\) for two of them and at most
\(\degdet\) for the third; our complexity analysis in
\cref{sec:determinant:complexity} shows that such a recursion gives the cost in
\cref{thm:polmat_det}. Precisely, if \(\deg(\det(\pmat_1)) \le \degdet/2\) then
degree properties of minimal kernel bases imply that \(\deg(\det(\kerbas_2))
\le \degdet/2\), yielding the two calls in half the degree; otherwise the
algorithm uses inexpensive row and column operations on \(\pmat\) to reduce to
the case \(\deg(\det(\pmat_1)) \le \degdet/2\).

Although this approach removes the obstacle of row basis computation which
arises in \cite{LaNeZh17}, it adds a requirement: all recursive calls must take
input matrices that are reduced. In the next section we discuss how to ensure
the reducedness of \(\pmat_1\) and \(\mat{B}\) thanks to a straightforward generalization
of \cite[Sec.\,3]{SarSto11}, and we describe a new algorithm which handles the
more involved case of \(\kerbas_2\).

\subsection{New tools, and ensuring reduced form in recursive calls}
\label{sec:intro:new_tools}

When outlining the approach of our determinant algorithm via the identity in
\cref{eqn:nonsingular_triangularization}, we implicitly assumed that the
matrices used as input in recursive calls, i.e.~\(\pmat_1\) and \(\kerbas_2\)
and \(\mat{B}\), do satisfy the input requirement of row reducedness: this is
not necessarily the case, even if starting from a reduced matrix \(\pmat\).

Concerning \(\pmat_1\), one may locate such a reduced submatrix of \(\pmat\)
and then permute rows and columns of \(\pmat\) (which only affects the sign of 
\(\det(\pmat)\)) to make this submatrix become the leading
principal submatrix \(\pmat_1\). This is a classical operation on reduced
matrices which suggests using a form slightly stronger than reduced form called
\emph{weak Popov form} \cite{MulSto03} (see \cref{sec:prelim:popov}). Assuming
that \(\pmat\) has this form ensures that its leading principal submatrix
\(\pmat_1\) has it as well.  This assumption is acceptable in terms of
complexity since one can transform a reduced \(\pmat\) into a weak Popov
\(\mat{P}\) by means of fast linear algebra in a cost negligible compared to
our target \cite[Sec.\,3]{SarSto11}; note that \(\pmat\) and \(\mat{P}\) have
the same determinant up to an easily found constant (see
\cref{algo:determinant_reduced}).

Next, the cases of \(\kerbas_2\) and \(\mat{B}\) are strongly linked. First, we
will not discuss \(\kerbas_2\) but the whole kernel basis \([\kerbas_1 \;\;
\kerbas_2]\). The fastest known algorithm for computing such a basis is that of
\cite{ZhLaSt12}, and for best efficiency it outputs a matrix in \emph{shifted}
reduced form, which is a generalization of reducedness involving degree weights
given by a tuple \(\shifts\in\ZZ^\rdim\) called a \emph{shift} (see
\cref{sec:prelim:reduced,sec:prelim:popov} for definitions); the non-shifted
case is for \(\shifts=\shiftz\). As in the determinant algorithm of
\cite{LaNeZh17}, here the shift for \([\kerbas_1 \;\; \kerbas_2]\) is taken as
the list of row degrees of \(\pmat\), denoted by \(\shifts=\rdeg{\pmat}\); for
the characteristic matrix one has \(\shifts=(1,\ldots,1)\) but non-uniform
shifts may arise in recursive calls. We want \([\kerbas_1 \;\; \kerbas_2]\) to
be not only \(\shifts\)-reduced, but in \(\shifts\)-weak Popov form: a direct
consequence is that \(\mat{B}\) is in weak Popov form, and is thus suitable
input for a recursive call.

To obtain \([\kerbas_1 \;\; \kerbas_2]\) we use the kernel basis algorithm of
\cite{ZhLaSt12} and transform its output into \(\shifts\)-weak Popov form. A
minor issue is that the fastest known algorithm for such transformations was
written in \cite[Sec.\,3]{SarSto11} for non-shifted forms; yet it easily
extends to shifted forms as we show in \cref{sec:reduced_to_weak_popov},
obtaining the next result.

\begin{theorem}
  \label{thm:reduced_to_weak_popov}
  There is an algorithm~\algoname{ReducedToWeakPopov} which takes as input a
  matrix \(\pmat\in\pmatRing[\rdim][\cdim]\) with \(\rdim\leq \cdim\) and a shift
  \(\shifts\in\ZZ^\cdim\) such that \(\pmat\) is in \(\shifts\)-reduced form,
  and returns an \(\shifts\)-weak Popov form of \(\pmat\) using
  \(\bigO{\rdim^{\expmm-2} \cdim \degdet + \rdim^{\expmm-1} \cdim}\) operations
  in \(\field\), where \(\degdet = \sumTuple{\rdeg[\shifts]{\pmat}} - \rdim
  \cdot \min(\shifts)\).
\end{theorem}

\noindent Here, following usual notation recalled in
\cref{sec:prelim:notation}, \(\sumTuple{\rdeg[\shifts]{\pmat}}\) is the sum of
the \(\shifts\)-degrees of the rows of \(\pmat\). This result extends
\cite[Thm.\,13]{SarSto11} since for \(\shifts=\shiftz\) the quantity
\(\degdet\) is the sum of the row degrees of \(\pmat\) and in particular
\(\degdet \le \rdim \deg(\pmat)\), leading to the cost bound
\(\bigO{\rdim^{\expmm-1} \cdim \deg(\pmat)}\).

To summarize, at this stage we have outlined how to ensure, without exceeding
our target cost bound, that \(\pmat_1\) and \(\mat{B}\) are valid input for
recursive calls, i.e.~are in weak Popov form. Having \(\det(\pmat_1)\) and
\(\det(\mat{B})\), it remains to find \(\det(\kerbas_2)\) and then the sought
\(\det(\pmat)\) follows. We noted that, to ensure the form of \(\mat{B}\) but
also for efficiency reasons, the kernel basis \([\kerbas_1 \;\; \kerbas_2]\) is
computed in \(\shifts\)-weak Popov form for the shift \(\shifts=\rdeg{\pmat}\).
This causes the main difficulty related to our modification of the determinant
algorithm of \cite{LaNeZh17}: \(\kerbas_2\) is not valid input for a recursive
call since it is in \(\tuple{v}\)-weak Popov form for some shift \(\tuple{v}\),
a subtuple of \(\shifts\) which is possibly nonzero.

A first idea is to extend our approach to the shifted case, allowing recursive
calls with such a \(\tuple{v}\)-reduced matrix: this is straightforward but
gives an inefficient algorithm. Indeed, along the recursion the shift drifts
away from its initial value and becomes arbitrarily large and unbalanced with
respect to the degrees of the input matrices of recursive calls. For example,
as mentioned above the sum of row degrees of the initial non-shifted
\(\rdim\times\rdim\) matrix \(\pmat\) is \(\degdet=\deg(\det(\pmat))\), whereas
for the \(\tuple{v}\)-shifted \((\rdim/2)\times(\rdim/2)\) matrix \(\kerbas_2\)
we only have the same bound \(\degdet\) instead of one related to
\(\deg(\det(\kerbas_2))\) itself, which is known to be at most \(\degdet/2\) in
our algorithm. This gap, here between \(\degdet\) and \(\degdet/2\), will only
grow as the algorithm goes down the tree of recursive calls, meaning that
degrees in matrices handled recursively are not sufficiently well controlled.

Another idea is to compute a \(\shiftz\)-reduced matrix which has the same
determinant as \(\kerbas_2\). Finding a \(\shiftz\)-reduced form of
\(\kerbas_2\) within our target cost seems to be a difficult problem. The best
known algorithms for general \(\shiftz\)-reduction involve \(\log(\rdim)\)
factors, either explicitly \cite{GuSaStVa12} or implicitly \cite{NeiVu17} (in
the latter approach one starts by using the above-discussed triangularization
procedure of \cite{LaNeZh17} which we are modifying here to avoid
\(\log(\rdim)\) factors). More specific algorithms exploit the form of
\(\kerbas_2\), interpreting the problem as a change of shift from \(\tuple{v}\)
to \(\shiftz\); yet at the time of writing efficient changes of shifts have
only been achieved when the target shift is larger than the origin shift
\cite[Sec.\,5]{JeNeScVi17}, a fact that offers degree control for the
transformation between the two matrices. Another possibility is to compute the
so-called \(\tuple{v}\)-Popov form \(\mat{P}\) of \(\kerbas_2\), since its
transpose \(\trsp{\mat{P}}\) is \(\shiftz\)-reduced by definition (see
\cref{sec:prelim:popov}), and \(\det(\trsp{\mat{P}}) = \det(\mat{P})\) is
\(\det(\kerbas_2)\) up to a constant. However this suffers from the same issue,
as computing \(\mat{P}\) is essentially the same as changing the shift
\(\tuple{v}\) into the nonpositive shift \(-\pivDegs\), where \(\pivDegs\) is
the list of diagonal degrees of \(\kerbas_2\) \cite{SarSto11,JeNeScVi16}.

To circumvent these issues, we use the property that the transpose
\(\trsp{\kerbas_2}\) of a \(\tuple{v}\)-reduced matrix is in
\(-\tuple{d}\)-reduced form where \(\tuple{d} = \rdeg[\tuple{v}]{\kerbas_2}\).
This fact naturally comes up here since \(\det(\kerbas_2) =
\det(\trsp{\kerbas_2})\), but seems otherwise rarely exploited in polynomial
matrix algorithms: in fact we did not find a previous occurrence of it apart
from related degree considerations in \cite[Lem.\,2.2]{ZhoLab13}.

Transposing the above two approaches using \(\trsp{\kerbas_2}\) instead of
\(\kerbas_2\), we observe that computing a \(\shiftz\)-reduced form of
\(\trsp{\kerbas_2}\) is a change of shift from \(-\tuple{d}\) to
\(\tuple{0}\), and computing the \(-\tuple{d}\)-Popov form \(\mat{P}\) of
\(\trsp{\kerbas_2}\) is essentially a change of shift from \(-\tuple{d}\) to
\(-\pivDegs\). In both cases the target shift is larger than the origin shift,
implying that the kernel-based change of shift of \cite[Sec.\,5]{JeNeScVi17}
involves matrices of well-controlled degrees. Still, this is not enough to make
this change of shift efficient as such, the difficulty being now that the
average row degree of \(\trsp{\kerbas_2}\) may not be small: only its average
column degree, which corresponds to the average row degree of \(\kerbas_2\), is
controlled.

Our solution uses the second approach, computing the \(-\tuple{d}\)-Popov form
\(\mat{P}\), because it offers the a priori knowledge that the column degrees
of \(\mat{P}\) are exactly \(\pivDegs\). We exploit this degree knowledge to
carry out partial linearization techniques, originally designed for approximant
bases \cite{Storjohann06,JeaNeiVil20}, which we extend here to kernel bases.
These techniques allow us to reduce our problem to a kernel basis computation
where the matrix entries have uniformly small degrees, implying that it can be
efficiently handled via the minimal approximant basis algorithm
\algoname{PM-Basis} from \cite{GiJeVi03}. The next result summarizes the new
algorithmic tool developed in \cref{sec:weak_popov_to_popov} for finding
\(\mat{P}\).

\begin{theorem}
  \label{thm:weak_popov_to_popov}
  Let \(\shifts\in\ZZ^\rdim\), let \(\pmat\in\pmatRing[\rdim]\) be in
  \(-\shifts\)-weak Popov form, let \(\pivDegs \in \NN^\rdim\) be the
  \(-\shifts\)-pivot degree of \(\pmat\), and assume that \(\shifts \ge
  \pivDegs\). There is an algorithm~\algoname{WeakPopovToPopov} which takes as
  input \((\pmat,\shifts)\) and computes the \(-\shifts\)-Popov form of
  \(\pmat\) by
  \begin{itemize}
    \item performing \algoname{PM-Basis} at order less than
      \(\sumTuple{\shifts}/\rdim  + 4\) on an input matrix of row dimension at
      most \(6\rdim\) and column dimension at most \(3\rdim\),
    \item multiplying the inverse of a matrix in \(\matRing[\rdim]\) by a
      matrix in \(\pmatRing[\rdim]\) of column degree \(\pivDegs\),
    \item and performing \(\bigO{\rdim^2}\) extra operations in \(\field\).
  \end{itemize}
  Thus, computing the \(-\shifts\)-Popov form of \(\pmat\) can be done in
  \(\bigO{\rdim^\expmm \timegcd{\sumTuple{\shifts}/\rdim}}\) operations in
  \(\field\).
\end{theorem}

This theorem is a generalization of \cite[Sec.\,4]{SarSto11} to shifted forms,
for shifts \(-\shifts\) that satisfy the assumption \(\shifts \ge \pivDegs\).
Indeed, if \(\pmat\) is \(\shiftz\)-weak Popov, then one recovers
\cite[Thm.\,20]{SarSto11} by taking
\(\shifts=(\deg(\pmat),\ldots,\deg(\pmat))\) in the above theorem. For
comparison, the naive generalization of \cite[Sec.\,4]{SarSto11} to shifted
forms
runs in \(\bigO{\rdim^\expmm \timegcd{\max(\shifts)}}\), which exceeds our
target complexity as soon as \(\max(\shifts) \gg \sumTuple{\shifts}/\rdim\).
Hence the use of partial linearization techniques, which were not needed in the
non-shifted case featuring \(\max(\shifts) = \sumTuple{\shifts}/\rdim =
\deg(\pmat)\).

As mentioned above, our Algorithm~\algoname{WeakPopovToPopov} is based on the
computation of a kernel basis with a priori knowledge of the degree profile of
the output. This kernel problem is very close to the one handled in
\cite[Sec.\,4]{ZhoLab13}, except that in this reference one only has upper
bounds on the output degrees, implying a certain number---possibly logarithmic
in \(\rdim\)---of calls to \algoname{PM-Basis} to recover the output and its
actual degrees. In the same spirit but in the context of approximant bases, \cite[Sec.\,5]{JeaNeiVil20}
uses partial linearization techniques to reduce an arbitrary input with known
output degrees to essentially one call to \algoname{PM-Basis}, whereas
\cite[Algo.\,2]{ZhoLab12}  assumes weaker output degree information
and makes a potentially logarithmic number of calls to \algoname{PM-Basis}.

\subsection{Perspectives}
\label{sec:intro:perspectives}

We plan to implement our characteristic polynomial algorithm in the LinBox
ecosystem \cite{fflas-ffpack,Linbox}. First prototype experiments suggest that,
for large finite fields, it could be competitive with the existing
fastest-known implementation, based on the randomized algorithm of
\cite{PerSto07}. The native support for small fields of our algorithm should
outperform the algorithm of \cite{PerSto07} which requires expensive field
extensions. Another perspective stems from the remark that our algorithm
resorts to fast polynomial multiplication (see assumption $\hypeps$), while
previous ones did not \cite{KelGeh85,PerSto07}: we woud like to understand
whether the same cost can be achieved by a purely linear algebraic approach.
Finally, perhaps the most challenging problem related to characteristic
polynomial computation is to compute Frobenius forms deterministically in the
time of matrix multiplication, the current best known complexity bound being
\(\bigO{\rdim^\expmm \log(\rdim)\log(\log(\rdim))}\) \cite{Storjohann01}; and
more generally computing Smith forms of polynomial matrices with a cost free of
factors logarithmic in the matrix dimension.

\section{Preliminaries on polynomial matrices}
\label{sec:prelim}

In this section we present the notation as well as basic definitions and
properties that will be used throughout the paper.

\subsection{Notation}
\label{sec:prelim:notation}

Tuples of integers will often be manipulated entry-wise. In particular, for
tuples \(\tuple{s},\tuple{t}\in\ZZ^\cdim\) of the same length \(\cdim\), we
write \(\tuple{s} + \tuple{t}\) for their entry-wise sum, and the inequality
\(\tuple{s} \le \tuple{t}\) means that each entry in \(\tuple{s}\) is less than
or equal to the corresponding entry in \(\tuple{t}\).  The concatenation of
tuples is denoted by \((\tuple{s},\tuple{t})\). We write
\(\sumTuple{\tuple{t}}\) for the sum of the entries of \(\tuple{t}\). The tuple
of zeros is denoted by \(\shiftz = (0,\ldots,0)\); its length is understood
from the context. 

For an \(\rdim \times \cdim\) matrix \(\pmat\) over some ring, we write
\(\pmat_{i,j}\) for its entry at index \((i,j)\). We extend this to
submatrices: given sets \(I \subseteq \{1,\ldots,\rdim\}\) and \(J \subseteq
\{1,\ldots,\cdim\}\) of row and column indices, we write \(\pmat_{I,J}\) for
the submatrix of \(\pmat\) formed by its entries indexed by \(I \times J\).
Besides, \(\matrows{\pmat}{I}\) stands for the submatrix of \(\pmat\) formed by
its rows with index in \(I\), and we use the similar notation
\(\matcols{\pmat}{J}\). The transpose of \(\pmat\) is denoted by
\(\trsp{\pmat}\). The identity matrix of size \(\cdim\) is denoted by
\(\idMat{\cdim}\), while the \(\cdim\times\cdim\) matrix with \(1\) on the
antidiagonal and \(0\) elsewhere is denoted by \(\revmat{\cdim}\). In
particular, when writing \(\shifts\revmat{\cdim}\) for a tuple \(\shifts =
(\shift_1,\ldots,\shift_\cdim) \in \ZZ^\cdim\), we mean the reversed tuple
\(\shifts\revmat{\cdim} = (\shift_\cdim,\ldots,\shift_1)\).

Now consider \(\pmat\) with polynomial entries, i.e.~\(\pmat \in
\pmatRing[\rdim][\cdim]\). The \emph{degree} of \(\pmat\) is denoted by
\(\deg(\pmat)\) and is the largest of the degrees of its entries, or
\(-\infty\) if \(\pmat=\matz\). The \emph{row degree} of \(\pmat\) is the tuple
\(\rdeg{\pmat} \in (\NN \cup \{-\infty\})^{\rdim}\) whose \(i\)th entry is
\(\max_{1\le j \le \cdim}(\deg(\pmat_{i,j}))\). More generally, for a tuple
\(\shifts = (\shift_1,\ldots,\shift_\cdim) \in \ZZ^\cdim\), the
\emph{\(\shifts\)-row degree} of \(\pmat\) is the tuple \(\rdeg[\shifts]{\pmat}
\in (\ZZ \cup \{-\infty\})^{\rdim}\) whose \(i\)th entry is \(\max_{1\leq j
\leq \cdim}(\deg(\pmat_{i,j}) + \shift_j)\). In this context, the tuple
\(\shifts\) is commonly called a \emph{(degree) shift} \cite{BeLaVi99}. The
\emph{(shifted) column degree} of \(\pmat\) is defined similarly.

We write \(\xDiag{\shifts}\) for the \(\cdim \times \cdim\) diagonal matrix
\(\diag{\var^{\shift_1},\ldots,\var^{\shift_\cdim}}\) which is over the ring
\(\field[\var,\var^{-1}]\) of Laurent polynomials over \(\field\). Note that,
hereafter, Laurent polynomials will only arise in proofs and explanations, more
specifically in considerations about shifted degrees: they never arise in
algorithms, which for the sake of clarity only involve polynomials in
\(\polRing\). The usefulness of this matrix \(\xDiag{\shifts}\) will become
clear in the definition of leading matrices in the next subsection.

The next lemma gives a link between shifted row degrees and shifted column
degrees. We will mostly use the following particular case of it: the column
degree of \(\pmat\) is at most \(\cdegs\in\NN^\cdim\) (entry-wise) if and only
if the \(-\cdegs\)-row degree of \(\pmat\) is nonpositive.

\begin{lemma}[{\cite[Lemma~2.2]{ZhoLab13}}]
  \label{lem:rdegcdeg}
  Let \(\pmat\) be a matrix in \(\pmatRing[\rdim][\cdim]\), \(\cdegs\) be a
  tuple in \(\ZZ^\cdim\), and \(\tuple{t}\) be a tuple in \(\ZZ^\rdim\). Then,
  \(\cdeg[\tuple{t}]{\pmat} \leq \cdegs\) if and only if
  \(\rdeg[-\cdegs]{\pmat}\leq -\tuple{t}\).
\end{lemma}

\subsection{Bases of modules, kernel bases and approximant bases}
\label{sec:prelim:app_ker}

We recall that any \(\polRing\)-submodule \(\mathcal{M}\) of \(\pmatRing[1][\cdim]\)
is free, and admits a basis formed by \(\rk\) elements of \(\pmatRing[1][\cdim]\),
where \(\rk \le \cdim\) is called the rank of \(\mathcal{M}\) \cite[see
e.g.][]{DumFoo04}. Such a basis can thus be represented as an
\(\rk\times\cdim\) matrix \(\mat{B}\) over \(\polRing\) whose rows are the
basis elements; this basis matrix \(\mat{B}\) has rank \(\rk\).

For a matrix \(\pmat \in \pmatRing[\rdim][\cdim]\), its \emph{row space} is the
\(\polRing\)-submodule \(\{\row{p} \pmat , \row{p} \in \pmatRing[1][\rdim]\}\)
of \(\pmatRing[1][\cdim]\), that is, the set of all \(\polRing\)-linear
combinations of its rows. If \(\rowbas \in \pmatRing[\rk][\cdim]\) is a basis
of this row space, then \(\rowbas\) is said to be a \emph{row basis} of
\(\pmat\); in particular, \(\rk\) is the rank of \(\rowbas\) and of \(\pmat\).

The \emph{left kernel} of \(\pmat\), denoted by \(\modKer{\pmat}\), is the
\(\polRing\)-module \(\{ \row{p} \in \pmatRing[1][\rdim] \mid \row{p} \pmat =
\matz \}\). A matrix \(\kerbas \in \pmatRing[k][\rdim]\) is a \emph{left kernel
basis of \(\pmat\)} if its rows form a basis of \(\modKer{\pmat}\), in which case
\(k=\rdim-\rk\).  Similarly, a \emph{right} kernel basis of \(\pmat\) is a
matrix \(\kerbas \in \pmatRing[\cdim][(\cdim-\rk)]\) whose \emph{columns} form
a basis of the right kernel of \(\pmat\).

Given positive integers \(\orders=(\order_1,\ldots,\order_\cdim)\in\ZZp^\cdim\)
and a matrix \(\sys \in \pmatRing[\rdim][\cdim]\), the set of
\emph{approximants for \(\sys\) at order \(\orders\)} \cite[see
e.g.][]{BarBul92,BecLab94} is the \(\polRing\)-submodule of
\(\pmatRing[1][\rdim]\) defined as
\[
  \modApp{\orders}{\sys} =
  \{ \app \in \pmatRing[1][\rdim] \mid \app \sys = \matz \bmod \xDiag{\orders} \}.
\]
The identity \(\app \sys = \matz \bmod \xDiag{\orders}\) means that \(\app
\matcol{\sys}{j} = 0 \bmod \var^{\order_j}\) for \(1 \le j \le \cdim\). Since
all \(\rdim\) rows of the matrix \(x^{\max(\orders)} \idMat{\rdim}\) are in
\(\modApp{\orders}{\sys}\), this module has rank \(\rdim\).

\subsection{Leading matrices and reduced forms of polynomial matrices}
\label{sec:prelim:reduced}

We will often compute with polynomial matrices that have a special form, called
the \emph{(shifted) reduced form}. It corresponds to a type of minimality of
the degrees of such matrices, and also provides good control of these degrees
during computations as illustrated by the \emph{predictable degree property}
\cite{Forney75} \cite[Thm.\,6.3-13]{Kailath80} which we recall below. In this
section, we introduce the notion of \emph{row} reducedness; to avoid confusion,
we will not use the similar notion of column reducedness in this paper, and
thus all further mentions of reducedness refer to row reducedness.

For shifted reduced forms, we follow the definitions in
\cite{BeLaVi99,BeLaVi06}. Let \(\pmat \in \pmatRing[\rdim][\cdim]\) and
\(\shifts \in \ZZ^\cdim\), and let \(\shiftt = (t_1,\ldots,t_\rdim) =
\rdeg[\shifts]{\pmat}\). Then, the \(\shifts\)-leading matrix of \(\pmat\) is
the matrix \(\lmat[\shifts]{\pmat} \in \matRing[\rdim][\cdim]\) whose entry
\((i,j)\) is the coefficient of degree \(t_i - \shift_j\) of the entry
\((i,j)\) of \(\pmat\), or \(0\) if \(t_i = -\infty\). Equivalently,
\(\lmat[\shifts]{\pmat}\) is the coefficient of degree zero of
\(\xDiag{-\shiftt} \pmat \xDiag{\shifts}\), whose entries are in
\(\field[\var^{-1}]\). The matrix \(\pmat\) is said to be in
\(\shifts\)-reduced form if its \(\shifts\)-leading matrix has full row rank.
In particular, a matrix in \(\shifts\)-reduced form must have full row rank.

For a matrix $\mat{M} \in \pmatRing[k][\rdim]$, we have
\(\rdeg[\shifts]{\mat{M}\pmat} \le \rdeg[\tuple{t}]{\mat{M}}\) and this is an
equality when no cancellation of leading terms occurs in this
left-multiplication. The predictable degree property states that \(\pmat\) is
\(\shifts\)-reduced if and only if \(\rdeg[\shifts]{\mat{M}\pmat} =
\rdeg[\tuple{t}]{\mat{M}}\) holds for any $\mat{M} \in \pmatRing[k][\rdim]$.
Here is a useful consequence of this characterization.

\begin{lemma}
  \label{lem:product_of_lmats}
  Let \(\shifts\in\ZZ^\cdim\), let \(\pmat \in \pmatRing[\rdim][\cdim]\), and
  let \(\shiftt = \rdeg[\shifts]{\pmat}\). If \(\pmat\) is \(\shifts\)-reduced,
  then the identity \(\lmat[\shifts]{\mat{M}\pmat} = \lmat[\shiftt]{\mat{M}}
  \lmat[\shifts]{\pmat}\) holds for any \(\mat{M} \in \pmatRing[k][\rdim]\).
\end{lemma}
\begin{proof}
  Let \(\tuple{d} = \rdeg[\shifts]{\mat{M}\pmat}\). By definition,
  \(\lmat[\shifts]{\mat{M}\pmat}\) is the coefficient of degree \(0\) of the
  matrix \(\xDiag{-\tuple{d}} \mat{M}\pmat \xDiag{\shifts} = \xDiag{-\tuple{d}}
  \mat{M} \xDiag{\shiftt} \xDiag{-\shiftt} \pmat \xDiag{\shifts}\), whose
  entries are in \(\field[x^{-1}]\). Besides, since \(\rdeg[\shifts]{\pmat}
  = \shiftt\) and since the predictable degree property gives
  \(\rdeg[\shiftt]{\mat{M}} = \tuple{d}\), the matrices \(\xDiag{-\tuple{d}}
  \mat{M} \xDiag{\shiftt}\) and \(\xDiag{-\shiftt} \pmat \xDiag{\shifts}\) are
  over \(\field[x^{-1}]\) and their coefficients of degree \(0\) are
  \(\lmat[\shiftt]{\mat{M}}\) and \(\lmat[\shifts]{\pmat}\), respectively.
\end{proof}

Another characterization of matrices in \(\shifts\)-reduced form is that they
have minimal \(\shifts\)-row degree among all matrices which represent the same
\(\polRing\)-module \cite[Def.\,2.13]{Zhou12}; in this paper, we will use the
following consequence of this minimality.

\begin{lemma}
  \label{lem:minimality_implies_basis}
  Let \(\mathcal{M}\) be a submodule of \(\pmatRing[1][\cdim]\) of rank
  \(\rdim\), let \(\shifts\in\ZZ^\cdim\), and let \(\shiftt \in \ZZ^\rdim\) be the
  \(\shifts\)-row degree of some \(\shifts\)-reduced basis of \(\mathcal{M}\).
  Without loss of generality, assume that \(\shiftt\) is nondecreasing. Let
  \(\mat{B} \in \pmatRing[\rdim][\cdim]\) be a matrix of rank \(\rdim\) whose rows
  are in \(\mathcal{M}\), and let \(\tuple{d} \in \ZZ^\rdim\) be its
  \(\shifts\)-row degree sorted in nondecreasing order. If \(\tuple{d} \le
  \shiftt\), then \(\mat{B}\) is an \(\shifts\)-reduced basis of \(\mathcal{M}\),
  and \(\tuple{d} = \shiftt\).
\end{lemma}
\begin{proof}
  Up to permuting the rows of \(\mat{B}\), we assume that
  \(\rdeg[\shifts]{\mat{B}} = \tuple{d}\) without loss of generality. Let
  \(\pmat \in \pmatRing[\rdim][\cdim]\) be an \(\shifts\)-reduced basis of
  \(\mathcal{M}\) such that \(\rdeg[\shifts]{\pmat} = \shiftt\). Since the rows
  of \(\mat{B}\) are in \(\mathcal{M}\), there exists a matrix \(\mat{U} \in
  \pmatRing[\rdim][\rdim]\) such that \(\mat{B} = \mat{U}\pmat\); and
  \(\mat{U}\) is nonsingular since \(\mat{B}\) and \(\pmat\) have rank
  \(\rdim\). Since \(\pmat\) is \(\shifts\)-reduced, the predictable degree
  property applies, ensuring that
  \[
    \tuple{d} = \rdeg[\shifts]{\mat{B}} = \rdeg[\shifts]{\mat{U}\pmat} =
    \rdeg[\shiftt]{\mat{U}}.
  \]
  This means that \(\deg(\mat{U}_{i,j}) \le d_i - t_j\) for all \(1\le i,
  j\le\rdim\).

  Now, assume by contradiction that \(\tuple{d}=\shiftt\) does not hold. Thus,
  \(d_{k} < t_{k}\) for some \(1\le k\le\rdim\). Then, for \(i\le k\) and
  \(j\ge k\) we have \(d_i \le d_k < t_k \le t_j\), hence \(\deg(\mat{U}_{i,j})
  < 0\).  Thus, the submatrix
  \(\submat{\mat{U}}{\{1,\ldots,k\}}{\{k,\ldots,\rdim\}}\) is zero, which
  implies that \(\mat{U}\) is singular; this is a contradiction, hence
  \(\tuple{d} = \shiftt\).

  Since \(\shiftt\) is nondecreasing, the inequality \(\deg(\mat{U}_{i,j}) \le
  t_i - t_j\) implies that \(\mat{U}\) is a block lower triangular matrix whose
  diagonal blocks have degree \(0\); hence these blocks are invertible matrices
  over \(\field\), and \(\mat{U}\) is unimodular \cite[see][Lemma~6 for
  similar degree considerations, starting from stronger assumptions on
  \(\pmat\) and \(\mat{B}\)]{SarSto11}. Thus, \(\mat{B}\) is a basis of
  \(\mathcal{M}\).

  Furthermore, it is easily observed that
  \(\lmat[\tuple{d}]{\mat{U}}\in\matRing[\rdim][\rdim]\) is block lower
  triangular with the same invertible diagonal blocks as \(\mat{U}\); hence
  \(\lmat[\tuple{d}]{\mat{U}}\) is invertible. On the other hand,
  \cref{lem:product_of_lmats} states that \(\lmat[\shifts]{\mat{B}} =
  \lmat[\tuple{d}]{\mat{U}}\lmat[\shifts]{\pmat}\). Thus
  \(\lmat[\shifts]{\mat{B}}\) has rank \(\rdim =
  \rank{\lmat[\shifts]{\pmat}}\), and \(\mat{B}\) is \(\shifts\)-reduced.
\end{proof}

\subsection{Pivots and weak Popov forms of polynomial matrices}
\label{sec:prelim:popov}

For algorithmic purposes, it is often convenient to work with reduced forms
that satisfy some additional requirements, called \emph{weak Popov
forms}. These are intrinsically related to the notion of \emph{pivot} of a
polynomial matrix.

For a nonzero vector \(\pvec \in \pmatRing[1][\cdim]\) and a shift \(\shifts
\in \ZZ^\cdim\), the \(\shifts\)-pivot of \(\pvec\) is its rightmost entry
\(p_j\) such that \(\deg(p_j) + s_j = \rdeg[\shifts]{\pvec}\)
\cite{BeLaVi99,MulSto03}; it corresponds to the rightmost nonzero entry of
\(\lmat[\shifts]{\pvec}\). The index \(j=\pivInd\) and the degree
\(\deg(p_\pivInd) = \pivDeg\) of this entry are called the \(\shifts\)-pivot
index and \(\shifts\)-pivot degree, respectively. For brevity, in this paper
the pair \((\pivInd,\pivDeg)\) is called the \emph{\(\shifts\)-pivot profile}
of \(\pvec\). By convention, the zero vector in \(\pmatRing[1][\cdim]\) has
\(\shifts\)-pivot index \(0\) and \(\shifts\)-pivot degree \(-\infty\). These
notions are extended to matrices \(\pmat \in \pmatRing[\rdim][\cdim]\) by
forming row-wise lists. For example, the \(\shifts\)-pivot index of \(\pmat\)
is \(\pivInds = (\pivInd_1,\ldots,\pivInd_\rdim) \in \ZZp^{\rdim}\) where
\(\pivInd_i\) is the \(\shifts\)-pivot index of the row \(\matrow{\pmat}{i}\).
The \(\shifts\)-pivot degree \(\pivDegs\) and the \(\shifts\)-pivot profile
\((\pivInd_i,\pivDeg_i)_{1\le i\le \rdim}\) of \(\pmat\) are defined similarly.

Then, \(\pmat\) is said to be in \(\shifts\)-weak Popov form if it has no zero
row and \(\pivInds\) is strictly increasing; and \(\pmat\) is said to be in
\(\shifts\)-unordered weak Popov form if it is in \(\shifts\)-weak Popov form
up to row permutation, i.e.~the entries of \(\pivInds\) are pairwise distinct.
Furthermore, a matrix is in \(\shifts\)-Popov form if it is in \(\shifts\)-weak
Popov form, its \(\shifts\)-pivots are monic, and each of these
\(\shifts\)-pivots has degree strictly larger than the other entries in the
same column. For a given \(\polRing\)-submodule \(\mathcal{M}\) of
\(\pmatRing[1][\cdim]\), there is a unique basis of \(\mathcal{M}\) which is in
\(\shifts\)-Popov form \cite{BeLaVi99}.

For a given matrix \(\mat{B}\), the matrix \(\pmat\) is said to be an
\(\shifts\)-reduced (resp.~\(\shifts\)-weak Popov, \(\shifts\)-Popov) form of
\(\mat{B}\) if \(\pmat\) is a row basis of \(\mat{B}\) and \(\pmat\) is in
\(\shifts\)-reduced (resp.~\(\shifts\)-weak Popov, \(\shifts\)-Popov) form.

Like for \(\shifts\)-reducedness, the property of a matrix \(\pmat \in
\pmatRing[\rdim][\cdim]\) to be in \(\shifts\)-weak Popov form depends only on
its \(\shifts\)-leading matrix \(\lmat[\shifts]{\pmat} \in
\matRing[\rdim][\cdim]\), namely on the fact that it has a staircase shape.
Indeed, \(\pmat\) is in \(\shifts\)-weak (resp.~\(\shifts\)-unordered weak) Popov
form if and only if \(\lmat[\shifts]{\pmat}\) has no zero row and
\(\revmat{\rdim}\lmat[\shifts]{\pmat}\revmat{\cdim}\) is in row echelon form
(resp.~in row echelon form up to row permutation); this was used as a
definition by Beckermann et al.~\cite{BeLaVi99,BeLaVi06}. In particular, for
any constant matrix \(\mat{C} \in \matRing[\rdim][\cdim]\), we have
\(\lmat[\shiftz]{\mat{C}} = \mat{C}\) and therefore \(\mat{C}\) is in
\(\shiftz\)-weak (resp.~\(\shiftz\)-unordered weak) Popov form if and only if it
has no zero row and \(\revmat{\rdim}\mat{C}\revmat{\cdim}\) is in row echelon
form (resp.~in row echelon form up to row permutation). Taking
\(\mat{C}=\lmat[\shifts]{\pmat}\), the next lemma follows.

\begin{lemma}
  \label{lem:weak_popov_lmat}
  Let \(\pmat \in \pmatRing[\rdim][\cdim]\) and let \(\shifts\in\ZZ^\cdim\).
  Then, \(\pmat\) is in \(\shifts\)-weak (resp.~\(\shifts\)-unordered weak) Popov
  form if and only if \(\lmat[\shifts]{\pmat}\) is in \(\shiftz\)-weak
  (resp.~\(\shiftz\)-unordered weak) Popov form.
\end{lemma}

Furthermore, if \(\pmat\) is in \(\shifts\)-weak Popov form and
\((j_1,\ldots,j_\rdim)\) is the list of indices of pivot columns in the row
echelon form \(\revmat{\rdim}\lmat[\shifts]{\pmat}\revmat{\cdim}\) (in other
words, this list is the column rank profile of that matrix), then the
\(\shifts\)-pivot index of \(\pmat\) is equal to
\((\cdim+1-j_\rdim,\ldots,\cdim+1-j_1)\). This leads to the following lemma
which states that the \(\shifts\)-pivot profile is an invariant of
left-unimodularly equivalent \(\shifts\)-weak Popov forms
\cite{Kailath80,BeLaVi99,BeLaVi06}, generalizing the fact that for matrices
over \(\field\) the set of indices of pivot columns is an invariant of
left-equivalent row echelon forms.

\begin{lemma}
  \label{lem:popov_pivot}
  Let \(\shifts\in\ZZ^\cdim\) and let \(\pmat \in \pmatRing[\rdim][\cdim]\) be
  in \(\shifts\)-unordered weak Popov form with \(\shifts\)-pivot profile
  \((\pivInd_i,\pivDeg_i)_{1\le i\le \rdim}\). Then, the \(\shifts\)-pivot
  profile of the \(\shifts\)-Popov form of \(\pmat\) is
  \((\pivInd_{\sigma(i)},\pivDeg_{\sigma(i)})_{1\le i\le \rdim}\), where
  \(\sigma : \{ 1,\ldots,\rdim \} \to \{ 1,\ldots,\rdim \}\) is the permutation
  such that \((\pivInd_{\sigma(i)})_{1\le i\le\rdim}\) is strictly increasing.
\end{lemma}
\begin{proof}
  Without loss of generality we assume that \(\pmat\) is in \(\shifts\)-weak
  Popov form, implying also \(\sigma(i) = i\) for \(1 \le i \le \rdim\).
  Let \(\mat{P} \in \pmatRing[\rdim][\cdim]\) be the \(\shifts\)-Popov form of
  \(\pmat\): we want to prove that \(\pmat\) and \(\mat{P}\) have the same
  \(\shifts\)-pivot index and the same \(\shifts\)-pivot degree. Let
  \(\mat{U}\) be the unimodular matrix such that \(\mat{P} = \mat{U}\pmat\);
  then \cref{lem:product_of_lmats} yields \(\lmat[\shifts]{\mat{P}} =
  \lmat[\shiftt]{\mat{U}}\lmat[\shifts]{\pmat}\), where \(\shiftt =
  \rdeg[\shifts]{\pmat}\). Since both \(\lmat[\shifts]{\mat{P}}\) and
  \(\lmat[\shifts]{\pmat}\) have full row rank, \(\lmat[\shiftt]{\mat{U}} \in
  \matRing[\rdim]\) is invertible. Then
  \[
    \revmat{\rdim}\lmat[\shifts]{\mat{P}}\revmat{\cdim}
    =
    \revmat{\rdim}\lmat[\shiftt]{\mat{U}}\revmat{\rdim}\revmat{\rdim}\lmat[\shifts]{\pmat}\revmat{\cdim}
  \]
  holds, and thus the row echelon forms
  \(\revmat{\rdim}\lmat[\shifts]{\mat{P}}\revmat{\cdim}\) and
  \(\revmat{\rdim}\lmat[\shifts]{\pmat}\revmat{\cdim}\) have the same pivot
  columns since \(\revmat{\rdim}\lmat[\shiftt]{\mat{U}}\revmat{\rdim} \in
  \matRing[\rdim]\) is invertible. It follows from the discussion preceding
  this lemma that \(\mat{P}\) has the same \(\shifts\)-pivot index as
  \(\pmat\).
  
  As a consequence, \(\mat{P}\) has the same \(\shifts\)-pivot degree as
  \(\pmat\) if and only if \(\rdeg[\shifts]{\mat{P}} = \rdeg[\shifts]{\pmat}\).
  Suppose by contradiction that there exists an index \(i\) such that
  \(\rdeg[\shifts]{\matrow{\mat{P}}{i}} < \rdeg[\shifts]{\matrow{\pmat}{i}}\).
  Then, build the matrix \(\mat{B} \in \pmatRing[\rdim][\cdim]\) which is equal
  to \(\pmat\) except for its \(i\)th row which is replaced by
  \(\matrow{\mat{P}}{i}\). By construction, \(\mat{B}\) has rank \(\rdim\)
  (since it is in \(\shifts\)-weak Popov form) and its rows are in the
  row space of \(\pmat\). Writing \(\tuple{d}\) for the tuple
  \(\rdeg[\shifts]{\mat{B}}\) sorted in nondecreasing order, and \(\tuple{u}\)
  for the tuple \(\shiftt\) sorted in nondecreasing order, we have \(\tuple{d}
  \le \tuple{u}\) and \(\tuple{d} \neq \tuple{u}\), which contradicts
  \cref{lem:minimality_implies_basis}. Hence there is no such index \(i\), and
  since this proof by contradiction is symmetric in \(\pmat\) and \(\mat{P}\),
  there is no index \(i\) such that \(\rdeg[\shifts]{\matrow{\pmat}{i}} <
  \rdeg[\shifts]{\matrow{\mat{P}}{i}}\) either. Thus \(\rdeg[\shifts]{\pmat} =
  \rdeg[\shifts]{\mat{P}}\).
\end{proof}

We will also use the following folklore fact, which is a corollary of
\cref{lem:product_of_lmats}, and has often been used in algorithms for
approximant bases or kernel bases in order to preserve the reducedness of
matrices during the computation.

\begin{lemma}
  \label{lem:product_of_reduced}
  Let \(\pmat \in \pmatRing[\rdim][\cdim]\) and \(\mat{B} \in
  \pmatRing[k][\rdim]\), and let \(\shifts\in\ZZ^\cdim\) and \(\shiftt =
  \rdeg[\shifts]{\pmat} \in \ZZ^\rdim\). Then,
  \begin{itemize}
    \item if \(\pmat\) is \(\shifts\)-reduced and \(\mat{B}\) is
      \(\shiftt\)-reduced, then \(\mat{B}\pmat\) is \(\shifts\)-reduced;
    \item if \(\pmat\) is in \(\shifts\)-weak Popov form and \(\mat{B}\) is
      in \(\shiftt\)-weak Popov form, then \(\mat{B}\pmat\) is in
      \(\shifts\)-weak Popov form.
  \end{itemize}
\end{lemma}
\begin{proof}
  Since \(\pmat\) is \(\shifts\)-reduced, \cref{lem:product_of_lmats} states
  that \(\lmat[\shifts]{\mat{B}\pmat} = \mat{M}\mat{L}\) where \(\mat{L} =
  \lmat[\shifts]{\pmat} \in \matRing[\rdim][\cdim]\) and \(\mat{M} =
  \lmat[\shiftt]{\mat{B}} \in \matRing[k][\rdim]\). The first item then follows
  from the fact that if \(\mat{M}\) has rank \(k\) and \(\mat{L}\) has rank
  \(\rdim\), then \(\mat{M}\mat{L}\) has rank \(k\). Similarly, the second item
  reduces to prove that, assuming \(\mat{M}\) and \(\mat{L}\) are in row
  echelon form with full row rank, then \(\mat{M}\mat{L}\) is also in row
  echelon form. Let \((a_1,\dots,a_k)\) (resp.~\((b_1,\dots,b_\rdim)\)) be the
  pivot indices of  \(\mat{M}\) (resp.~\(\mat{L}\)). Then the \(i\)th row of
  \(\mat{M}\mat{L}\) is a nonzero multiple of row \(a_i\) of \(\mat{L}\)
  combined with multiples of rows of \(\mat{L}\) of index greater than \(a_i\).
  Consequently, the pivot indices of the rows of \(\mat{M}\mat{L}\) are
  \(b_{a_1} < \cdots < b_{a_k}\), which proves that \(\mat{M}\mat{L}\) is in
  row echelon form.
\end{proof}

Finally, under assumptions that generalize the situation encountered in our
determinant algorithm below, we show that the pivot entries of a kernel basis
\([\kerbas_1 \;\; \kerbas_2]\) are located in its rightmost columns, that is,
in \(\kerbas_2\).

\begin{lemma}
  \label{lem:kernel_pivots}
  Let \(\shiftt\in\ZZ^\cdim\), let \(\mat{F} \in \pmatRing[\cdim]\) be in
  \(\shiftt\)-weak Popov form, and let \(\tuple{u} =
  \rdeg[\shiftt]{\mat{F}}\). Let \(\mat{G} \in \pmatRing[\rdim][\cdim]\) and
  \(\tuple{v} \in \ZZ^\rdim\) be such that \(\tuple{v} \ge
  \rdeg[\shiftt]{\mat{G}}\), and let \(\kerbas = [\kerbas_1 \;\; \kerbas_2] \in
  \pmatRing[\rdim][(\rdim+\cdim)]\) be a \((\tuple{u},\tuple{v})\)-weak
  Popov basis of \(\modKer{[\begin{smallmatrix} \mat{F} \\ \mat{G}
  \end{smallmatrix}]}\), where \(\kerbas_1\) and \(\kerbas_2\) have
  \(\rdim\) and \(\cdim\) columns, respectively. Then, the
  \((\tuple{u},\tuple{v})\)-pivot entries of \(\kerbas\) are all located in
  \(\kerbas_2\); in particular, \(\kerbas_2\) is in \(\tuple{v}\)-weak
  Popov form.
\end{lemma}
\begin{proof}
  Since the \((\tuple{u},\tuple{v})\)-pivot entry of a row is the rightmost
  entry of that row which reaches the \((\tuple{u},\tuple{v})\)-row degree, it
  is enough to prove that \(\rdeg[\tuple{v}]{\kerbas_2} \ge
  \rdeg[\tuple{u}]{\kerbas_1}\). First, from \(\tuple{v} \ge
  \rdeg[\shiftt]{\mat{G}}\), we obtain \(\rdeg[\tuple{v}]{\kerbas_2} \ge
  \rdeg[{\rdeg[\shiftt]{\mat{G}}}]{\kerbas_2}\). Now, by definition,
  \(\rdeg[{\rdeg[\shiftt]{\mat{G}}}]{\kerbas_2} \ge
  \rdeg[\shiftt]{\kerbas_2\mat{G}}\). Since the rows of \(\kerbas\) are in
  \(\modKer{[\begin{smallmatrix} \mat{F} \\ \mat{G}
  \end{smallmatrix}]}\), we have \(\kerbas_2\mat{G} = -\kerbas_1\mat{F}\),
  hence \(\rdeg[\shiftt]{\kerbas_2\mat{G}} =
  \rdeg[\shiftt]{\kerbas_1\mat{F}}\).  Since \(\mat{F}\) is
  \(\shiftt\)-reduced, we can apply the predictable degree property:
  \(\rdeg[\shiftt]{\kerbas_1\mat{F}} = \rdeg[\tuple{u}]{\kerbas_1}\). This
  proves the sought inequality. For the last point, note that the
  \((\tuple{u},\tuple{v})\)-pivot entries of \([\kerbas_1 \;\; \kerbas_2]\)
  located in \(\kerbas_2\) correspond to \(\tuple{v}\)-pivot entries in
  \(\kerbas_2\). Thus, since \([\kerbas_1 \;\; \kerbas_2]\) is in
  \((\tuple{u},\tuple{v})\)-weak Popov form with all
  \((\tuple{u},\tuple{v})\)-pivot entries in \(\kerbas_2\), it follows that the
  \(\tuple{v}\)-pivot index of \(\kerbas_2\) is increasing.
\end{proof}

\subsection{Basic subroutines and their complexity}
\label{sec:prelim:subroutines}

To conclude these preliminaries, we recall known fast algorithms for three
polynomial matrix subroutines used in our determinant algorithm: multiplication
with unbalanced degrees, minimal approximant bases, and minimal kernel bases;
we give the corresponding complexity estimates adapted to our context and in
particular using our framework stated in \cref{sec:intro:cost_bounds}.

\paragraph{Unbalanced multiplication}

Polynomial matrix algorithms often involve multiplication with matrix operands
whose entries have degrees that may be unbalanced but still satisfy properties
that can be exploited to perform the multiplication efficiently. Here we will
encounter products of reduced matrices with degree properties similar to those
discussed in \cite[Sec.\,3.6]{ZhLaSt12}, where an efficient approach for
computing such products was given.

\begin{lemma}
  \label{lem:subroutine:mult}
  There is an algorithm \algoname{UnbalancedMultiplication} which takes as
  input a matrix \(\mat{B} \in \pmatRing[k][\rdim]\) with \(k\le\rdim\), a
  matrix \(\pmat \in \pmatRing[\rdim][\cdim]\) with \(\cdim\le\rdim\), and an
  integer \(\degdet\) greater than or equal to both the sum of the positive
  entries of \(\rdeg[\shiftz]{\pmat}\) and that of
  \(\rdeg[{\rdeg[\shiftz]{\pmat}}]{\mat{B}}\),
  and returns the product \(\mat{B}\pmat\) using \(\bigO{\rdim^\expmm
  \timepm{\degdet/\rdim}}\) operations in \(\field\), assuming \(\hypsm\) and
  \(\hypeps\).
\end{lemma}
\begin{proof}
  Zhou et al.~\cite[Sec.\,3.6]{ZhLaSt12} gave such an algorithm, yet with a
  cost analysis which hides logarithmic factors; because these factors are our
  main concern here we will rely on the version in \cite[Sec.\,4]{JeNeScVi17}.
  In this reference, Algorithm \algoname{UnbalancedMultiplication} was
  described for square matrices. One could adapt it to the case of rectangular
  \(\pmat\) and \(\mat{B}\) as in the statement above. However, for the sake of
  conciseness and with no impact on the asymptotic cost bound, we consider the
  more basic approach of forming the square \(\rdim\times\rdim\) matrices
  \(\mat{D} = [\begin{smallmatrix} \mat{B} \\ \matz
  \end{smallmatrix}]\)
  and \(\mat{C} = [\pmat \;\; \matz]\), computing \(\mat{D}\mat{C}\) using the
  above-cited algorithm, and retrieving \(\mat{B}\pmat\) from it. Now, by
  construction, both the sum of the positive entries of
  \(\rdeg[\shiftz]{\mat{C}}\) and that of
  \(\rdeg[{\rdeg[\shiftz]{\pmat}}]{\mat{D}}\) are at most \(\degdet\), hence
  \cite[Prop.\,4.1]{JeNeScVi17} applies: defining \(\bar{\rdim}\) and \(d\) as
  the smallest powers of \(2\) greater than or equal to \(\rdim\) and
  \(\degdet/\rdim\), it states that the computation of \(\mat{D}\mat{C}\) costs
  \(\bigO{\sum_{0 \le i \le \log_2(\bar{\rdim})} 2^i  (2^{-i}\bar\rdim)^\expmm
  \timepm{2^i d}}\) operations in \(\field\). Using \(\hypsm\) and \(\hypeps\),
  which ensure respectively that \(\timepm{2^i d} \le \timepm{2^i} \timepm{d}\)
  and \(\timepm{2^i} \in \bigO{2^{i(\expmm-1-\varepsilon)}}\) for some
  \(\varepsilon>0\), we obtain that this bound is in \(\bigO{\bar\rdim^\expmm
  \timepm{d} \sum_{0\le i \le \log_2(\bar{\rdim})}2^{-i\varepsilon}} \subseteq
  \bigO{\bar\rdim^\expmm \timepm{d}}\). This is in \(\bigO{\rdim^\expmm
  \timepm{\degdet/\rdim}}\), since \(\bar\rdim \in \Theta(\rdim)\) and \(d \in
  \Theta(1+\degdet/\rdim)\).
\end{proof}

\paragraph{Minimal approximant basis}

The second basic tool we will use is approximant bases computation; for this,
we will use the algorithm \algoname{PM-Basis}, originally described in
\cite{GiJeVi03}. Precisely, we rely on the slightly modified version presented
in \cite{JeaNeiVil20} which ensures that the computed basis is in shifted
weak Popov form.

\begin{lemma}
  \label{lem:subroutine:approx}
  There is an algorithm \algoname{PM-Basis} which takes as input a tuple
  \(\orders \in \ZZp^\cdim\), a matrix \(\sys \in \pmatRing[\rdim][\cdim]\)
  with \(\cdeg{\sys} < \orders\), and a shift \(\shifts \in \ZZ^\rdim\), and
  returns a basis of \(\modApp{\orders}{\sys}\) in \(\shifts\)-weak
  Popov form using \(\bigO{(\rdim+\cdim)\rdim^{\expmm-1}
  \timegcd{\max(\orders)}}\) operations in \(\field\).
\end{lemma}
\begin{proof}
  The algorithm is \cite[Algo.\,2]{JeaNeiVil20}; to accommodate non-uniform
  order \(\orders\), it is called with input order \(\Gamma = \max(\orders)\)
  and input matrix \(\sys\xDiag{(\Gamma,\ldots,\Gamma)-\orders}\) as explained
  in \cite[Rmk.\,3.3]{JeaNeiVil20}. According to
  \cite[Prop.\,3.2]{JeaNeiVil20}, this costs \(\bigO{(1 + \frac{\cdim}{\rdim})
  \sum_{0\le i \le \lceil \log_2(\Gamma) \rceil} 2^i m^\expmm \timepm{2^{-i}
\Gamma}}\) operations in \(\field\), which is precisely the claimed bound by
definition of \(\timegcd{\cdot}\).
\end{proof}

\paragraph{Minimal kernel basis}

We will make use of the algorithm of Zhou et al.~\cite{ZhLaSt12}, which itself
relies on unbalanced products and approximant bases, and returns a kernel basis
in shifted reduced form efficiently for input matrices with small average row
degree.

\begin{lemma}
  \label{lem:subroutine:kernel}
  There is an algorithm \algoname{KernelBasis} which takes as input a full
  column rank matrix \(\sys\in\pmatRing[\rdim][\cdim]\) with \(\rdim \ge\cdim\)
  and \(\rdim \in  \bigO{n}\), and a shift \(\shifts\in\NN^\rdim\) such that
  \(\shifts \ge \rdeg[\shiftz]{\sys}\), and returns a basis of
  \(\modKer{\sys}\) in \(\shifts\)-reduced form using \(\bigO{\rdim^\expmm
  \timegcd{\degdet/\rdim}}\) operations in \(\field\), assuming
  \(\hypsl,\hypsm,\hypeps\).  Here \(\degdet\) is the sum of the entries of
  \(\shifts\), and the sum of the \(\shifts\)-row degree of this basis is at
  most \(\degdet\).
\end{lemma}
\begin{proof}
  The algorithm of Zhou et al.~\cite{ZhLaSt12} computes an \(\shifts\)-reduced
  basis \(\kerbas \in \pmatRing[k][\rdim]\) of \(\modKer{\sys}\); precisely,
  this reference is about computing a basis of the right kernel in column
  reduced form, yet this naturally translates into left kernels and row reduced
  forms by taking suitable transposes. Furthermore, the last claim in the lemma
  follows from \cite[Thm.\,3.4]{ZhLaSt12}, which states that any such basis
  \(\kerbas\) is such that \(\sumTuple{\rdeg[\shifts]{\kerbas}} \le
  \sumTuple{\shifts} = \degdet\). For the complexity, we rely on the analysis
  in \cite[Prop.\,B.1]{JeNeScVi17} which shows that, defining \(\bar{\rdim}\)
  and \(d\) as the smallest powers of \(2\) greater than or equal to \(\rdim\)
  and \(\degdet/\rdim\), this computation costs
  \[
    \bigO{\sum_{j=0}^{\log_2(\bar\rdim)} 2^j \left( \sum_{i=0}^{\log_2(2^{-j}\bar\rdim)} 2^i (2^{-i-j}\bar\rdim)^\expmm \timepm{2^{i+j}d} + \sum_{i=0}^{\log_2(2^jd)} 2^i (2^{-j}\bar\rdim)^\expmm \timepm{2^{j-i}d}  \right)  }
  \]
  operations in \(\field\). Now the same analysis as in the proof of
  \cref{lem:subroutine:mult} shows that, assuming \(\hypsm\) and \(\hypeps\),
  the first inner sum is in \(\bigO{(2^{-j}\bar\rdim)^\expmm \timepm{2^jd}}\),
  and by definition of \(\timegcd{\cdot}\) the second inner sum is in
  \(\bigO{(2^{-j}\bar\rdim)^\expmm \timegcd{2^j d}}\). Thus the total cost is
  in \(\bigO{\sum_{0 \le j \le \log_2(\bar\rdim)} 2^j (2^{-j}\bar\rdim)^\expmm
    \timegcd{2^jd}}\), which is in \(\bigO{\bar\rdim^\expmm \timegcd{d} \sum_{0
  \le j \le \log_2(\bar\rdim)} 2^{j(1-\expmm)} \timegcd{2^j}}\) since
  \(\hypsl\) and \(\hypsm\) ensure that \(\timegcd{\cdot}\) is
  submultiplicative. Similarly to the proof of \cref{lem:subroutine:mult}, this
  bound is in \(\bigO{\rdim^\expmm \timegcd{\degdet/\rdim}}\) thanks to
  \(\hypeps\).
\end{proof}

%%%%%%%%%%%%%%%%%%%%%%%%%%%%%%%%%%%%%%%%%%%%%%%%%%%%%%%%%%%%%%%%%%%%%%%%%%%%%
\section{Determinant algorithm for reduced matrices}
\label{sec:determinant}

In this section, we present the main algorithm in this paper, which computes
the determinant of a matrix in reduced form using the subroutines listed in
\cref{sec:prelim:subroutines} as well as the algorithms
\algoname{ReducedToWeakPopov} and \algoname{WeakPopovToPopov} from
\cref{thm:reduced_to_weak_popov,thm:weak_popov_to_popov}. Taking for granted
the proof of these theorems in
\cref{sec:reduced_to_weak_popov,sec:weak_popov_to_popov}, here we prove the
correctness of our determinant algorithm in \cref{sec:determinant:algo} and
analyse its complexity in \cref{sec:determinant:complexity}, thus proving
\cref{thm:polmat_det}.

\subsection{Two properties of determinants of reduced matrices}
\label{sec:determinant:properties}

\paragraph{Leading coefficient of the determinant}

All bases of a given submodule of \(\pmatRing[1][\cdim]\) of rank \(\cdim\) have the
same determinant up to a constant factor, i.e.~up to multiplication by an
element of \(\field\setminus\{0\}\). Many algorithms operating on polynomial
matrices such as \algoname{PM-Basis} and \algoname{KernelBasis} compute such
bases, so that their use in a determinant algorithm typically leads to
obtaining the sought determinant up to a constant factor; then finding the
actual determinant requires to efficiently recover this constant \cite[see
e.g.][Sec.\,4]{LaNeZh17}. Since in this paper we seek determinants of matrices
in reduced form, this issue is easily handled using the next result.

\begin{lemma}
  \label{lem:leading_coeff_det}
  Let \(\shifts \in\ZZ^\cdim\) and \(\pmat\in\pmatRing[\cdim]\). If \(\pmat\)
  is in \(\shifts\)-reduced form, the leading coefficient of \(\det(\pmat)\) is
  \(\det(\lmat[\shifts]{\pmat})\). In particular, if \(\pmat\) is in
  \(\shifts\)-weak Popov form, the leading coefficient of
  \(\det(\pmat)\) is the product of the leading coefficients of the diagonal
  entries of \(\pmat\).
\end{lemma}
\begin{proof}
  The second claim is a direct consequence of the first, since for \(\pmat\) in
  \(\shifts\)-weak Popov form, \(\lmat[\shifts]{\pmat}\) is lower
  triangular with diagonal entries equal to the leading coefficients of the
  diagonal entries of \(\pmat\xDiag{\shifts}\), which are
  the leading coefficients of the diagonal entries of \(\pmat\). For the first claim in the case
  \(\shifts=\shiftz\), we refer to \cite[Sec.\,6.3.2]{Kailath80}, and in
  particular Eq.~(23) therein. Now, for an arbitrary \(\shifts\) and \(\pmat\)
  in \(\shifts\)-reduced form, we consider the nonnegative shift \(\shiftt =
  \shifts - (\min(\shifts),\ldots,\min(\shifts))\) and observe that
  \(\lmat[\shifts]{\pmat} = \lmat[\shiftz]{\pmat\xDiag{\shiftt}}\), hence
  \(\pmat\xDiag{\shiftt}\) is \(\shiftz\)-reduced, and thus the leading
  coefficient of \(\det(\pmat\xDiag{\shiftt}) = \det(\pmat)
  \det(\xDiag{\shiftt})\) (which is the same as that of \(\det(\pmat)\)) is
  equal to \(\det(\lmat[\shiftz]{\pmat\xDiag{\shiftt}}) =
  \det(\lmat[\shifts]{\pmat})\).
\end{proof}

\paragraph{From shifted to non-shifted}

In \cref{sec:intro:new_tools}, we explained that one step of our algorithm
consists in finding \(\det(\kerbas)\) for a matrix \(\kerbas\) in
\(\tuple{v}\)-weak Popov form, and that it achieves this by computing the
\(-\tuple{d}\)-Popov form of \(\trsp{\kerbas}\), which is already in
\(-\tuple{d}\)-weak Popov form. The next lemma substantiates this; note that in
\cref{sec:intro:new_tools} we had left out the reversal matrix
\(\revmat{\cdim}\) for the sake of exposition.

\begin{lemma}
  \label{lem:row_column_reduced}
  Let \(\tuple{v} \in\ZZ^\cdim\), let \(\mat{K}\in\pmatRing[\cdim]\), and let
  \(\tuple{d} = \rdeg[\tuple{v}]{\mat{K}}\).
  \begin{enumerate}[(a)]
    \item\label{lem:transposition:leadingmat} if \(\lmat[\tuple{v}]{\mat{K}}\) has no zero column, then
      \(\lmat[-\tuple{d}]{\trsp{\mat{K}}} = \trsp{\lmat[\tuple{v}]{\mat{K}}}\)
      and \(\rdeg[-\tuple{d}]{\trsp{\mat{K}}} = -\tuple{v}\);
    \item\label{lem:transposition:reduced} if \(\mat{K}\) is in \(\tuple{v}\)-reduced form, then
      \(\trsp{\mat{K}}\) is in \(-\tuple{d}\)-reduced form;
    \item\label{lem:transposition:owpopov} if \(\mat{K}\) is in \(\tuple{v}\)-weak Popov form, then
      \(\revmat{\cdim}\trsp{\mat{K}}\revmat{\cdim}\) is in
      \(-\tuple{d}\revmat{\cdim}\)-weak Popov form;
    \item\label{lem:transposition:popov} if furthermore \(\mat{P}\) is the \(-\tuple{d}\revmat{\cdim}\)-Popov
      form of \(\revmat{\cdim}\trsp{\mat{K}}\revmat{\cdim}\), then
      \(\trsp{\mat{P}}\) is in \(\shiftz\)-weak Popov form and
      \(\det(\mat{K})=\det(\lmat[\tuple{v}]{\mat{K}})\det(\trsp{\mat{P}})\).
  \end{enumerate} 
\end{lemma}
\begin{proof}
  By definition, \(\trsp{\lmat[\tuple{v}]{\mat{K}}}\) is the coefficient of
  degree \(0\) of \(\trsp{(\xDiag{-\tuple{d}} \mat{K} \xDiag{\tuple{v}})} =
  \xDiag{\tuple{v}} \trsp{\mat{K}} \xDiag{-\tuple{d}}\), which is a matrix over
  \(\field[x^{-1}]\). The assumption on \(\lmat[\tuple{v}]{\mat{K}}\) implies
  that this coefficient of degree \(0\) has no zero row. It follows that
  \(\rdeg[-\tuple{d}]{\trsp{\mat{K}}} = -\tuple{v}\) and that this coefficient
  of degree \(0\) is \(\lmat[-\tuple{d}]{\trsp{\mat{K}}}\).
  \cref{lem:transposition:reduced} follows from
  \cref{lem:transposition:leadingmat} by definition of shifted reduced forms.

  From now on, we assume that \(\mat{K}\) is in \(\tuple{v}\)-weak
  Popov form. Then \(\lmat[\tuple{v}]{\mat{K}}\) is invertible and lower
  triangular, and in particular \(\lmat[-\tuple{d}]{\trsp{\mat{K}}} =
  \trsp{\lmat[\tuple{v}]{\mat{K}}}\). Since \(\revmat{\cdim}\) is a permutation
  matrix, we obtain
  \(\lmat[-\tuple{d}\revmat{\cdim}]{\revmat{\cdim}\trsp{\mat{K}}\revmat{\cdim}}
  = \revmat{\cdim}\lmat[-\tuple{d}]{\trsp{\mat{K}}}\revmat{\cdim} =
  \revmat{\cdim}\trsp{\lmat[\tuple{v}]{\mat{K}}}\revmat{\cdim}\), which is
  invertible and lower triangular. Hence
  \(\revmat{\cdim}\trsp{\mat{K}}\revmat{\cdim}\) is in
  \(-\tuple{d}\revmat{\cdim}\)-weak Popov form.

  For \cref{lem:transposition:popov}, since \(\mat{P}\) is \(\cdim\times\cdim\) and in
  \(-\tuple{d}\revmat{\cdim}\)-Popov form, we have
  \(\lmat[\shiftz]{\trsp{\mat{P}}} = \idMat{\cdim}\), hence \(\trsp{\mat{P}}\)
  is in \(\shiftz\)-weak Popov form.  Furthermore, since
  \(\revmat{\cdim}\trsp{\mat{K}}\revmat{\cdim}\) is unimodularly equivalent to
  \(\mat{P}\), its determinant is \(\det(\mat{K}) =
  \det(\revmat{\cdim}\trsp{\mat{K}}\revmat{\cdim}) = \lambda\det(\mat{P})\) for
  some \(\lambda\in\field\setminus\{0\}\). Applying
  \cref{lem:leading_coeff_det} to \(\mat{P}\) shows that \(\det(\mat{P})\) is
  monic, hence \(\lambda\) is the leading coefficient of \(\det(\mat{K})\);
  applying the same lemma to \(\mat{K}\) yields \(\lambda =
  \det(\lmat[\tuple{v}]{\mat{K}})\).
\end{proof}

\subsection{Algorithm and correctness}
\label{sec:determinant:algo}

Our main determinant algorithm is \algoname{DeterminantOfWeakPopov}
(\cref{algo:determinant_weakpopov}), which takes as input a matrix in
\(\shiftz\)-weak Popov form and computes its determinant using
recursive calls on matrices of smaller dimension. Then, we compute the
determinant of a \(\shiftz\)-reduced matrix by first calling
\algoname{ReducedToWeakPopov} to find a \(\shiftz\)-weak Popov matrix which has
the same determinant up to a nonzero constant, and then calling the previous
algorithm on that matrix. This is detailed in \cref{algo:determinant_reduced}.

\begin{algorithm}
  \caption{\algoname{DeterminantOfReduced}$(\pmat)$}
  \label{algo:determinant_reduced}

  \begin{algorithmic}[1]
  \Require{a matrix \(\pmat\in\pmatRing[\rdim]\) in \(\shiftz\)-reduced form.}
  \Ensure{the determinant of \(\pmat\).}

  \State \(\mat{P} \in \pmatRing[\rdim][\rdim] \assign \algoname{ReducedToWeakPopov}(\pmat,\shiftz)\)
  \State \(\Delta \assign \Call{DeterminantOfWeakPopov}{\mat{P}}\);
          ~~\(\ell_\Delta \in \field\setminus\{0\} \assign\) leading coefficient of \(\Delta\)
  \State \Return \(\det(\lmat[\shiftz]{\pmat}) \,\Delta / \ell_\Delta\)
  \end{algorithmic}
\end{algorithm}

The correctness of \cref{algo:determinant_reduced} is obvious: according to
\cref{thm:reduced_to_weak_popov,pro:algo:determinant}, \(\mat{P}\) is a
\(\shiftz\)-weak Popov form of \(\pmat\), and \(\Delta\) is the determinant of
\(\mat{P}\) up to multiplication by some element of \(\field\setminus\{0\}\).
Thus \(\det(\pmat) = \ell \Delta\) for some \(\ell \in \field\setminus\{0\}\), and
\cref{lem:leading_coeff_det} yields \(\ell = \det(\lmat[\shiftz]{\pmat}) /
\ell_\Delta\).

Concerning the cost bound, \cref{thm:reduced_to_weak_popov} states that the
first step uses \(\bigO{\rdim^\expmm (1 + \degdet/\rdim)}\) operations in
\(\field\), where \(\degdet = \sumTuple{\rdeg[\shiftz]{\pmat}}\); since
\(\pmat\) is \(\shiftz\)-reduced, this is \(\degdet = \deg(\det(\pmat))\)
\cite[Sec.\,6.3.2]{Kailath80}. In the last step, the determinant computation
costs \(\bigO{\rdim^\expmm}\) operations, while scaling \(\Delta\) by a
constant costs \(\bigO{\degdet}\) operations. The second step uses
\(\bigO{\rdim^\expmm \timegcd{\degdet/\rdim}}\) operations according to
\cref{pro:algo:determinant}, under the assumptions \(\hypsl\), \(\hypsm\), and
\(\hypeps\).

We now describe the main algorithm of this paper
(\cref{algo:determinant_weakpopov}) and focus on its correctness. We also
mention cost bounds for all steps of the algorithm that are not recursive
calls, but we defer the core of the complexity analysis to
\cref{sec:determinant:complexity}.

\begin{algorithm}
  \caption{\algoname{DeterminantOfWeakPopov}$(\pmat)$}
  \label{algo:determinant_weakpopov}
  \begin{algorithmic}[1]

  \Require{a matrix \(\pmat\in\pmatRing[\rdim]\) in \(\shiftz\)-weak Popov form.}
  \Ensure{the determinant of \(\pmat\), up to multiplication by some element of \(\field\setminus\{0\}\).}

  \State \(\shifts = (\shift_1,\ldots,\shift_\rdim) \in \NN^\rdim \assign \rdeg[\shiftz]{\pmat}\);
            ~~\(\degdet \assign \shift_1 + \cdots + \shift_\rdim\)
          \Comment{\(\degdet = \) degree of \(\det(\pmat)\)}
          \label{algo:determinant:diagdeg_degdet} 

  \State\InlineIf{\(\rdim=1\)}{\Return the polynomial \(f\) such that \(\pmat = [f]\)}
      \Comment{base case: \(1\times1\) matrix}
      \label{algo:determinant:basestep_dim}

  \State\InlineIf{\(\degdet=0\)}{\Return the product of diagonal entries of \(\pmat\)}
      \Comment{base case: matrix over \(\field\)}
      \label{algo:determinant:basestep_deg}

  \If{\(\degdet<\rdim\)} \Comment{handle constant rows to reduce to dimension \(\le\degdet\)\label{algo:determinant:constant_rows_if}}
    \State \(\mat{B} \assign \pmat \: \lmat[0]{\pmat}^{-1}\) from which rows
    and columns with indices in \(\{i \mid \shift_i = 0\}\) are removed
          \label{algo:determinant:constant_rows}
    \State \Return \algoname{DeterminantOfWeakPopov}\((\mat{B})\)
  \EndIf

  \State\InlineIf{\(\shift_1+\cdots+\shift_{\lfloor \rdim/2 \rfloor} > \degdet/2\)}{%
        \Return \algoname{DeterminantOfWeakPopov}\((\revmat{\rdim} \pmat \: \lmat[\shiftz]{\pmat}^{-1}\revmat{\rdim})\)}
      \label{algo:determinant:reduce_to_smaller}
  \State Write \(\pmat = [\begin{smallmatrix} \pmat_1 & \pmat_2 \\ \pmat_3 & \pmat_4 \end{smallmatrix}]\),
  with \(\pmat_1\) of size \(\lfloor \rdim/2 \rfloor \times \lfloor \rdim/2 \rfloor\) and 
      \(\pmat_4\) of size \(\lceil \rdim/2 \rceil \times \lceil \rdim/2 \rceil\)
      \label{algo:determinant:blocks}

  \State \(\kerbas  \in \pmatRing[\lceil \rdim/2 \rceil][\rdim] \assign \algoname{KernelBasis}([\begin{smallmatrix} \pmat_1 \\ \pmat_3 \end{smallmatrix}],\shifts)\)
      \label{algo:determinant:kernel} 

  \State \([\kerbas_1 \;\; \kerbas_2] \in \pmatRing[\lceil \rdim/2 \rceil][\rdim] \assign \algoname{ReducedToWeakPopov}(\kerbas,\shifts)\), where \(\kerbas_2\) is \(\lceil \rdim/2 \rceil \times \lceil \rdim/2 \rceil\)
      \label{algo:determinant:to_weak_popov} 

  \State \(\mat{B} \assign \algoname{UnbalancedMultiplication}([\kerbas_1 \;\; \kerbas_2], [\begin{smallmatrix} \pmat_2 \\ \pmat_4 \end{smallmatrix}], \degdet)\)
      \Comment{\(\mat{B} = \kerbas_1 \pmat_2 + \kerbas_2 \pmat_4\)}
      \label{algo:determinant:unbalprod}

  \State \(\Delta_1 \assign \algoname{DeterminantOfWeakPopov}(\mat{B})\)
      \Comment{first recursive call}
      \label{algo:determinant:callone} 

  \State \(\Delta_2 \assign \Call{DeterminantOfWeakPopov}{\pmat_1}\)
      \Comment{second recursive call}
      \label{algo:determinant:calltwo}

  \State \(\mat{P} \assign \Call{WeakPopovToPopov}{\revmat{\lceil \rdim/2 \rceil}\trsp{\kerbas_2}\revmat{\lceil \rdim/2 \rceil},\rdeg[(\shift_{\lfloor \rdim/2 \rfloor + 1},\ldots,\shift_\rdim)]{\kerbas_2}\revmat{\lceil \rdim/2 \rceil}}\)
      \label{algo:determinant:weak_to_popov}

  \State \(\Delta_3 \assign \Call{DeterminantOfWeakPopov}{\trsp{\mat{P}}}\)
      \Comment{third recursive call}
      \label{algo:determinant:callthree}

  \State \Return \(\Delta_1 \Delta_2 / \Delta_3\)
      \label{algo:determinant:main_formula}
  \end{algorithmic}
\end{algorithm}

\begin{proposition}
  \label{pro:algo:determinant}
  \cref{algo:determinant_weakpopov} is correct, and assuming that \(\hypsl\),
  \(\hypsm\), and \(\hypeps\) hold (hence in particular \(\expmm>2\)), it uses
  \(\bigO{\rdim^\expmm \timegcd{\degdet/\rdim}}\) operations in \(\field\).
\end{proposition}

\begin{proof}[Proof of correctness]
  The fact that \(\pmat\) is in \(\shiftz\)-weak Popov form has two
  consequences on the tuple \(\shifts\) computed at
  \cref{algo:determinant:diagdeg_degdet}: first, it is the \(\shiftz\)-pivot
  degree of \(\pmat\) (i.e.~its diagonal degrees), and second, the sum
  \(\degdet = \sumTuple{\shifts}\) is equal to the degree of the determinant of
  \(\pmat\) \cite[Sec.\,6.3.2]{Kailath80}.

  The main base case of the recursion is when \(\rdim=1\) and is handled at
  \cref{algo:determinant:basestep_dim}; it uses no operation in \(\field\). We
  use a second base case at \cref{algo:determinant:basestep_deg}: if
  \(\degdet=0\), then \(\pmat\) is an $\rdim\times\rdim$ matrix over
  \(\field\). Since it is in \(\shiftz\)-weak Popov, it is invertible
  and lower triangular, hence \(\det(\pmat)\) is the product of its diagonal
  entries, which is computed in \(\bigO{\rdim}\) multiplications in \(\field\).
  This base case is not necessary for obtaining the correctness and the cost
  bound in \cref{pro:algo:determinant}; still, not using it would incur a cost
  of \(\bigO{\rdim^\expmm}\) operations in the case \(\degdet=0\).

  For the recursion, we proceed inductively: we assume that the algorithm
  correctly computes the determinant for all \(\shiftz\)-weak Popov
  matrices of dimension less than \(\rdim\), and based on this we show that it
  is also correct for any \(\shiftz\)-weak Popov matrix \(\pmat\) of
  dimension \(\rdim\).

  \emph{Case 1: \(\degdet < \rdim\)}. Then \(\pmat\) has at least one constant
  row; using linear algebra we reduce to the case of a matrix \(\mat{B}\) of
  dimension at most \(\degdet\) with all rows of degree at least \(1\).  Since
  \(\shifts = \rdeg[\shiftz]{\pmat}\), we can write \(\pmat = \xDiag{\shifts}
  \lmat[\shiftz]{\pmat} + \mat{R}\) for a matrix \(\mat{R} \in
  \pmatRing[\rdim][\rdim]\) such that \(\rdeg[\shiftz]{\mat{R}} < \shifts\).
  Since \(\pmat\) is \(\shiftz\)-reduced, \(\lmat[\shiftz]{\pmat}\) is
  invertible and \(\pmat \: \lmat[\shiftz]{\pmat}^{-1} = \xDiag{\shifts} + \mat{R}
  \: \lmat[\shiftz]{\pmat}^{-1}\) with \(\rdeg[\shiftz]{\mat{R}
  \: \lmat[\shiftz]{\pmat}^{-1}} < \shifts\), which implies
  \(\lmat[\shiftz]{\pmat\: \lmat[\shiftz]{\pmat}^{-1}} = \idMat{\rdim}\). In
  particular, \(\pmat \: \lmat[\shiftz]{\pmat}^{-1}\) is in \(\shiftz\)-weak Popov
  form, and for each \(i\) such that the row \(\matrow{\pmat}{i}\)
  is constant, i.e.~\(\shift_i=0\), the \(i\)th row of
  \(\pmat\: \lmat[\shiftz]{\pmat}^{-1}\) is the \(i\)th row of
  the identity matrix. Therefore the matrix \(\mat{B}\) at
  \cref{algo:determinant:constant_rows} is in \(\shiftz\)-weak Popov
  form, has the same determinant as \(\pmat\) up to a constant, and has
  dimension \(\card{\{i \mid \shift_i\neq0\}} \le \degdet\). Hence the
  correctness in this case. In terms of complexity, computing \(\mat{B}\)
  essentially amounts to computing the product \(\pmat
  \: \lmat[\shiftz]{\pmat}^{-1}\), which is done by row-wise expanding \(\pmat\)
  into a \((\rdim + \degdet) \times \rdim\) matrix over \(\field\),
  right-multiplying by \(\lmat[\shiftz]{\pmat}^{-1}\), and compressing the
  result back into a polynomial matrix: this costs \(\bigO{\rdim^\expmm (1 +
  \degdet/\rdim)} \subseteq \bigO{\rdim^\expmm}\) operations.

  \emph{Case 2: \(\shift_1+\cdots+\shift_{\lfloor \rdim/2 \rfloor} >
  \degdet/2\)}. Then we modify the input \(\pmat\) so as to reduce to
  \emph{Case 3}. As we have seen above,
  \(\lmat[\shiftz]{\pmat\: \lmat[\shiftz]{\pmat}^{-1}} = \idMat{\rdim}\). We now
  reverse the diagonal entries by reversing the order of rows and columns: let
  \(\mat{B} = \revmat{\rdim} \pmat \: \lmat[\shiftz]{\pmat}^{-1}\revmat{\rdim}\).
  Then \(\lmat[\shiftz]{\mat{B}} = \revmat{\rdim}\lmat[\shiftz]{\pmat
  \: \lmat[\shiftz]{\pmat}^{-1}} \revmat{\rdim}= \idMat{\rdim}\), hence
  \(\mat{B}\) is in \(\shiftz\)-weak Popov form:
  \cref{algo:determinant:reduce_to_smaller} calls the algorithm on this matrix
  to obtain \(\det(\mat{B})\) up to a constant, and this yields \(\det(\pmat)\)
  since it is equal to \(\det(\lmat[\shiftz]{\pmat})\det(\mat{B})\). To
  conclude the proof of correctness in that case (assuming correctness in
  \emph{Case 3}), it remains to observe that \(\mat{B}\) has the same matrix
  dimension \(\rdim\) as \(\pmat\), and that the matrix \(\mat{B}\) has degrees
  such that calling the algorithm with input \(\mat{B}\) does not enter
  \emph{Case 2} but \emph{Case 3}. Indeed, we have \(\rdeg[\shiftz]{\mat{B}} =
  \shifts\revmat{\rdim}\), hence the sum of the first \(\lfloor \rdim/2 \rfloor\)
  entries of the tuple \(\rdeg[\shiftz]{\mat{B}}\) is \(\shift_{\rdim} +
  \cdots + \shift_{\lceil \rdim/2 \rceil+1} = \degdet - (\shift_1 + \cdots +
  \shift_{\lceil \rdim/2 \rceil})\), which is at most \(\degdet/2\) by
  assumption. In terms of complexity, the main step is to compute the product
  \(\pmat \: \lmat[\shiftz]{\pmat}^{-1}\), which costs \(\bigO{\rdim^\expmm (1 +
  \degdet/\rdim)}\) operations as we have seen above; this is in
  \(\bigO{\rdim^\expmm \timegcd{\degdet/\rdim}}\).

  \emph{Case 3: \(\shift_1+\cdots+\shift_{\lfloor \rdim/2 \rfloor} \le
  \degdet/2\)}. Then, \cref{algo:determinant:reduce_to_smaller} performs no
  action, and \cref{algo:determinant:blocks} defines submatrices of \(\pmat\).
  By construction, \([\begin{smallmatrix} \pmat_1 \\ \pmat_3
  \end{smallmatrix}]\) has full column rank and \(\shifts \ge
  \rdeg[\shiftz]{[\begin{smallmatrix} \pmat_1 \\ \pmat_3 \end{smallmatrix}]}\)
  holds. Thus, according to \cref{lem:subroutine:kernel},
  \cref{algo:determinant:kernel} uses \(\bigO{\rdim^\expmm
  \timegcd{\degdet/\rdim}}\) operations to compute an \(\shifts\)-reduced
  basis \(\kerbas\) of \(\modKer{[\begin{smallmatrix} \pmat_1 \\ \pmat_3
  \end{smallmatrix}]}\), with \(\sumTuple{\rdeg[\shifts]{\kerbas}} \le
  \degdet\). Then, \cref{thm:reduced_to_weak_popov} states that
  \cref{algo:determinant:to_weak_popov} transforms \(\kerbas\) into an
  \(\shifts\)-weak Popov basis \([\kerbas_1 \;\; \kerbas_2]\) of this
  kernel at a cost of \(\bigO{\rdim^\expmm (1 + \degdet/\rdim)}\) operations,
  since \(\sumTuple{\rdeg[\shifts]{\kerbas}} \le \degdet\) and \(\min(\shifts)
  \ge 0\). Since all \(\shifts\)-reduced bases of \(\modKer{\sys}\) have the
  same \(\shifts\)-row degree up to permutation,
  \(\sumTuple{\rdeg[\shifts]{[\kerbas_1 \;\; \kerbas_2]}} \le \degdet\) holds,
  hence the assumptions of \cref{lem:subroutine:mult} are satisfied and
  \cref{algo:determinant:unbalprod} uses
  \(\bigO{\rdim^\expmm \timepm{\degdet/\rdim}}\) operations to compute
  \(\mat{B} = \kerbas_1 \pmat_2 + \kerbas_2 \pmat_4\).

  The important observation at this stage is the identity
  \begin{equation}
    \label{eqn:product_recursion}
    \begin{bmatrix}
      \idMat{\lceil \rdim/2 \rceil} & \matz \\
      \kerbas_1 & \kerbas_2
    \end{bmatrix}
    \begin{bmatrix}
      \pmat_1 & \pmat_2 \\
      \pmat_3 & \pmat_4
    \end{bmatrix}
    =
    \begin{bmatrix}
      \pmat_1 & \pmat_2 \\
      \matz & \mat{B}
    \end{bmatrix}
  \end{equation}
  which, provided that \(\kerbas_2\) is nonsingular, implies \(\det(\pmat) =
  \det(\mat{B}) \det(\pmat_1) / \det(\kerbas_2)\). We are going to show that
  this is the formula used in \cref{algo:determinant:main_formula} to compute
  \(\det(\pmat)\).

  First, \(\pmat_1\) has dimension less than \(\rdim\) and, being a principal
  submatrix of the \(\shiftz\)-weak Popov matrix \(\pmat\), it is also
  in \(\shiftz\)-weak Popov form. Hence the recursive call at
  \cref{algo:determinant:calltwo} is sound and \(\Delta_2\) is equal to
  \(\det(\pmat_1)\) up to a constant.

  Since \(\pmat\) is in \(\shiftz\)-weak Popov form and \([\begin{smallmatrix}
    \idMat{\lceil \rdim/2 \rceil} & \matz \\ \kerbas_1 & \kerbas_2
    \end{smallmatrix}]\) is in \(\rdeg[\shiftz]{\pmat}\)-weak Popov form, their
    product \([\begin{smallmatrix} \pmat_1 & \pmat_2 \\ \matz & \mat{B}
      \end{smallmatrix}]\) is in \(\shiftz\)-weak Popov form; see
      \cref{lem:product_of_reduced}, or note that
      \(\lmat[\shiftz]{[\begin{smallmatrix} \pmat_1 & \pmat_2 \\ \matz &
      \mat{B} \end{smallmatrix}]}\) is invertible and lower triangular
      according to \cref{lem:product_of_lmats}. It follows that \(\mat{B}\) is
      in \(\shiftz\)-weak Popov form and has dimension less than \(\rdim\):
      \cref{algo:determinant:callone} recursively computes \(\Delta_1\), equal
      to \(\det(\mat{B})\) up to a constant.

  It remains to prove that \(\Delta_3\) computed at
  \cref{algo:determinant:weak_to_popov,algo:determinant:callthree} is equal to
  \(\det(\kerbas_2)\) up to a constant. Let \(\tuple{v} =
  \rdeg[\shiftz]{\pmat_4} = (\shift_{\lfloor \rdim/2 \rfloor +
  1},\ldots,\shift_\rdim)\) be the shift used at
  \cref{algo:determinant:weak_to_popov}, and let \(\tuple{d} =
  \rdeg[\tuple{v}]{\kerbas_2} = \rdeg[\shifts]{[\kerbas_1 \;\; \kerbas_2]}\).
  Applying \cref{lem:kernel_pivots} (with \(\mat{F} = \pmat_1\), \(\mat{G} =
  \pmat_3\), \(\shiftt=\shiftz\), and \(\tuple{v}\) as above) shows that
  \([\kerbas_1 \;\; \kerbas_2]\) has all its \(\shifts\)-pivot entries in
  \(\kerbas_2\), and in particular \(\kerbas_2\) is in \(\tuple{v}\)-weak
  Popov form. Let \(\pivDegs \in \NN^{\cdim}\) be the \(\tuple{v}\)-pivot
  degree of \(\kerbas_2\), where \(\cdim = \lceil \rdim/2 \rceil\), and note
  that \(\tuple{d} = \pivDegs + \tuple{v} \ge \pivDegs\) since
  \(\tuple{v}\ge\shiftz\). Then, \cref{lem:row_column_reduced} states that
  \(\revmat{\cdim}\trsp{\kerbas_2}\revmat{\cdim}\) is in
  \(-\tuple{d}\revmat{\cdim}\)-weak Popov form; its
  \(-\tuple{d}\revmat{\cdim}\)-pivot degree is the list of degrees of its
  diagonal entries, that is, \(\pivDegs\revmat{\cdim}\). Since
  \(\tuple{d}\revmat{\cdim} \ge \pivDegs\revmat{\cdim}\), we can apply
  \cref{thm:weak_popov_to_popov}, which implies that
  \cref{algo:determinant:weak_to_popov} computes the
  \(-\tuple{d}\revmat{\cdim}\)-Popov form \(\mat{P}\) of
  \(\revmat{\cdim}\trsp{\kerbas_2}\revmat{\cdim}\) using \(\bigO{\rdim^\expmm
    \timegcd{\sumTuple{\tuple{d}}/\rdim}}\) operations; as we have seen
    above, \(\sumTuple{\tuple{d}} = \sumTuple{\rdeg[\shifts]{[\kerbas_1 \;\;
    \kerbas_2]}} \le \degdet\). Then, from the last item of
    \cref{lem:row_column_reduced}, \(\trsp{\mat{P}}\) is in \(\shiftz\)-weak
    Popov form and
    \(\det(\mat{K})=\det(\lmat[\tuple{v}]{\mat{K}})\det(\trsp{\mat{P}})\),
    hence \cref{algo:determinant:callthree} correctly computes
    \(\det(\mat{K})\) up to a constant.
\end{proof}

To conclude this presentation of our determinant algorithm, we note that it
would be beneficial, in a practical implementation, to add an early exit.
Precisely, just after computing \(\det(\mat{B})\) at
\cref{algo:determinant:callone}, one could perform the following action before
(possibly) proceeding to the next steps:

\begin{algorithm*}
  \begin{algorithmic}[1]
    \Statex {\footnotesize 12b:} \InlineIf{\(\deg(\Delta_1) = \degdet\)}{\Return \(\Delta_1\)} \Comment{early exit}
  \end{algorithmic}
\end{algorithm*}

Indeed, recall that \(\Delta_1\) is \(\det(\mat{B})\) up to a constant;
furthermore we claim that
\begin{itemize}
  \item for a generic \(\pmat\), we have \(\deg(\Delta_1) = \degdet\),
  \item if \(\deg(\Delta_1) = \degdet\) (i.e.~\(\deg(\det(\mat{B})) = \degdet\)),
    then \(\det(\mat{B})\) is \(\det(\pmat)\) up to a constant.
\end{itemize}
It follows that for a generic matrix \(\pmat\), then \(\Delta_1\) is
\(\det(\pmat)\) up to a constant, hence the correctness of this early exit (see
also \cite[Sec.\,4.2.2]{GiJeVi03} for similar considerations). To prove the
above claim, first note that since \([\kerbas_1 \;\; \kerbas_2]\) is a
kernel basis, it has unimodular column bases \cite[Lem.\,2.2]{GioNei18}, and
thus it can be completed into a unimodular matrix \(\mat{U} =
[\begin{smallmatrix} \mat{U}_1 & \mat{U}_2 \\ \kerbas_1 & \kerbas_2
\end{smallmatrix}] \in \pmatRing[\rdim]\)
\cite[Lem.\,2.10]{ZhoLab14}. Therefore
\[
  \mat{U} \pmat =
  \begin{bmatrix}
    \mat{U}_1 & \mat{U}_2 \\
    \kerbas_1 & \kerbas_2
  \end{bmatrix}
  \begin{bmatrix}
    \pmat_1 & \pmat_2 \\
    \pmat_3 & \pmat_4
  \end{bmatrix}
  =
  \begin{bmatrix}
    \mat{B}_1 & \mat{B}_2 \\
    \matz & \mat{B}
  \end{bmatrix}
\]
where \([\mat{B}_1 \;\; \mat{B}_2] = [\mat{U}_1 \;\; \mat{U}_2] \pmat\). Since
\(\det(\mat{U})\) is in \(\field\setminus\{0\}\), \(\det(\pmat)\) is
\(\det(\mat{B}_1) \det(\mat{B})\) up to a constant. For the second item of the
claim, \(\deg(\Delta_1) = \degdet\) implies \(\deg(\det(\mat{B})) = \degdet =
\deg(\det(\pmat))\), hence \(\det(\mat{B}_1)\) is in \(\field\setminus\{0\}\).
The first item follows from the fact that \(\mat{B}_1\) is a row basis of
\([\begin{smallmatrix} \pmat_1 \\ \pmat_3 \end{smallmatrix}]\)
\cite[Lem.\,3.1]{ZhoLab13}; since the latter matrix has more rows than columns,
if \(\pmat_1\) and \(\pmat_3\) have generic entries, then such a row basis
\(\mat{B}_1\) is unimodular which means \(\det(\mat{B}_1) \in
\field\setminus\{0\}\) and thus \(\deg(\Delta_1) = \degdet\).

\subsection{Complexity analysis}
\label{sec:determinant:complexity}

We have seen above that all computations in \cref{algo:determinant_weakpopov}
other than recursive calls have an arithmetic cost in \(\bigO{\rdim^\expmm
\timegcd{\degdet/\rdim}}\) operations in \(\field\); here, we complete the
proof of the cost bound in \cref{prop:algo:knowndegker}. In this section, we
use the assumptions \(\hypsl\), \(\hypsm\), and \(\hypeps\) as well as their
consequences stated in \cref{sec:intro:cost_bounds}.

Let \(\mathcal{C}(\rdim,\degdet)\) denote the arithmetic cost of
\cref{algo:determinant_weakpopov}; recall that \(\degdet\) is the degree of the
determinant of the input, which is also the sum of its row degrees. First
consider the case \(\rdim \le \degdet\). If \(\shift_1+\cdots+\shift_{\lfloor
\rdim/2 \rfloor} > \degdet/2\), the reduction to the case
\(\shift_1+\cdots+\shift_{\lfloor \rdim/2 \rfloor} \le \degdet/2\) with the
same \(\rdim\) and \(\degdet\) performed at
\cref{algo:determinant:reduce_to_smaller} costs \(\bigO{\rdim^\expmm
(1+\degdet/\rdim)}\). Once we are in the latter case, there are three recursive
calls with input matrices having the following dimensions and degrees:
\begin{itemize}
  \item At \cref{algo:determinant:callone}, the matrix \(\mat{B}\) is
    \(\lceil \rdim/2 \rceil \times \lceil \rdim/2 \rceil\), and applying the
    predictable degree property on \cref{eqn:product_recursion} gives in
    particular \(\rdeg[\shiftz]{\mat{B}} = \rdeg[\shifts]{[\kerbas_1 \;\;
    \kerbas_2]}\), hence \(\sumTuple{\rdeg[\shiftz]{\mat{B}}} \le \degdet\).
  \item At \cref{algo:determinant:calltwo}, the matrix \(\mat{\pmat}_1\)
    is \(\lfloor \rdim/2 \rfloor \times \lfloor \rdim/2 \rfloor\) and the sum
    of its row degrees is \(\shift_1+\cdots+\shift_{\lfloor \rdim/2 \rfloor}\),
    which is at most \(\degdet/2\) by assumption.
  \item At \cref{algo:determinant:callthree}, the matrix \(\trsp{\mat{P}}\) is
    \(\lceil \rdim/2 \rceil \times \lceil \rdim/2 \rceil\) and its
    \(\shiftz\)-pivot degree is \(\rdeg[\shiftz]{\trsp{\mat{P}}} =
    \pivDegs\revmat{\rdim}\). Recall indeed that this is the list of diagonal
    degrees of \(\trsp{\mat{P}}\), which is the same as that of \(\mat{P}\),
    and thus the same as that of
    \(\revmat{\cdim}\trsp{\kerbas_2}\revmat{\cdim}\) according to
    \cref{lem:popov_pivot}. Now, from \(\sumTuple{\pivDegs+\tuple{v}} =
    \sumTuple{\pivDegs}+\sumTuple{\tuple{v}} \le \degdet\) and the assumption
    \(\sumTuple{\tuple{v}} = \shift_{\lfloor \rdim/2
    \rfloor+1}+\ldots+\shift_\rdim > \degdet/2\), we obtain
    \(\sumTuple{\rdeg[\shiftz]{\trsp{\mat{P}}}} = \sumTuple{\pivDegs} \le
    \degdet/2\).
\end{itemize}
We assume without loss of generality that $\rdim$ is a power of 2. If it is
not, a given input matrix can be padded with zeros, and ones on the main
diagonal, so as to form a square matrix with dimension the next power of two
and the same determinant. According to the three items above, the cost bound
then satisfies:
\[
  \mathcal{C}(\rdim,\degdet) \le
      2\mathcal{C}(\rdim/2,\lfloor \degdet/2\rfloor) +
      \mathcal{C}(\rdim/2,\degdet) +
      \bigO{\rdim^\expmm\timegcd{\degdet/\rdim}}.
\]
Letting the \(\bigO{\cdot}\) term aside, we illustrate this recurrence relation
in \cref{fig:dag_calls}.
\begin{figure}[htbp]
  \centering
  \begin{tikzpicture}[->,>=stealth',shorten >=1pt,auto,node distance=2.05cm,semithick]

    %% Nodes with recursion dim / degdet
    \node[dimdegdet] (C00)                    {\(\rdim,\degdet\)};

    \node (C10) [dimdegdet,below left of=C00]{\(\rdim/2,\degdet\)};
    \node (C11) [dimdegdet,below right of=C00]{\(\rdim/2,\lfloor \degdet/2 \rfloor\)};

    \node (C20) [dimdegdet,below left of=C10]{\(\rdim/4,\degdet\)};
    \node (C21) [dimdegdet,below right of=C10]{\(\rdim/4,\lfloor \degdet/2 \rfloor\)};
    \node (C22) [dimdegdet,below right of=C11]{\(\rdim/4,\lfloor \degdet/4 \rfloor\)};

    \node (C30) [dimdegdet,below left of=C20]{\(\rdim/8,\degdet\)};
    \node (C31) [dimdegdet,below right of=C20]{\(\rdim/8,\lfloor \degdet/2 \rfloor\)};
    \node (C32) [dimdegdet,below right of=C21]{\(\rdim/8,\lfloor \degdet/4 \rfloor\)};
    \node (C33) [dimdegdet,below right of=C22]{\(\rdim/8,\lfloor \degdet/8 \rfloor\)};

    \node (Cmu0) [dimdegdet,below left of=C30,yshift=-1cm]{\(1,\degdet\)};
    \node (Cmu1) [below right of=C30,yshift=-1cm]{\(\ldots\)};
    \node (Cmu2) [dimdegdet,below right of=C31,yshift=-1cm]{\(1,\lfloor \degdet/2^{j} \rfloor\)};
    \node (Cmu3) [below right of=C32,yshift=-1cm]{\(\ldots\)};
    \node (Cmu4) [dimdegdet,below right of=C33,yshift=-1cm]{\(1,\lfloor \degdet/2^\mu \rfloor\)};

    %% Nodes with number of global recursive calls
    \node[above=-0.1cm of C00] (C00nb) {\scriptsize \(1\)};

    \node (C10nb) [above=-0.1cm of C10]{\scriptsize \(1\)};
    \node (C11nb) [above=-0.1cm of C11]{\scriptsize \(2\)};

    \node (C20nb) [above=-0.1cm of C20]{\scriptsize \(1\)};
    \node (C21nb) [above=-0.1cm of C21]{\scriptsize \(4\)};
    \node (C22nb) [above=-0.1cm of C22]{\scriptsize \(4\)};

    \node (C30nb) [above=-0.1cm of C30]{\scriptsize \(1\)};
    \node (C31nb) [above=-0.1cm of C31]{\scriptsize \(6\)};
    \node (C32nb) [above=-0.1cm of C32]{\scriptsize \(12\)};
    \node (C33nb) [above=-0.1cm of C33]{\scriptsize \(8\)};

    \node (Cmu0nb) [above=-0.1cm of Cmu0]{\scriptsize \(1\)};
    \node (Cmu2nb) [above=-0.1cm of Cmu2]{\scriptsize \(2^j\binom{\mu}{j}\)};
    \node (Cmu4nb) [above=-0.1cm of Cmu4]{\scriptsize \(2^\mu\)};

    %% Arrows, also indicating, number of local recursive calls

    \path (C00) edge            node[left,near start]  {\scriptsize \(1\)} (C10)
                edge            node[right,near start] {\scriptsize \(2\)} (C11);
    \path (C10) edge            node[left,near start]  {\scriptsize \(1\)} (C20)
                edge            node[right,near start] {\scriptsize \(2\)} (C21);
    \path (C11) edge            node[left,near start]  {\scriptsize \(2\)} (C21)
                edge            node[right,near start] {\scriptsize \(4\)} (C22);
    \path (C20) edge            node[left,near start]  {\scriptsize \(1\)} (C30)
                edge            node[right,near start] {\scriptsize \(2\)} (C31);
    \path (C21) edge            node[left,near start]  {\scriptsize \(4\)} (C31)
                edge            node[right,near start] {\scriptsize \(8\)} (C32);
    \path (C22) edge            node[left,near start]  {\scriptsize \(4\)} (C32)
                edge            node[right,near start] {\scriptsize \(8\)} (C33);
    \path (C30) edge[dashed]    (Cmu0)
                edge[dashed,-]  (Cmu1);
    \path (C31) edge[dashed,-]  (Cmu1)
                edge[dashed,]   ([xshift=-0.3cm]Cmu2.north);
    \path (C32) edge[dashed,]   ([xshift=0.3cm]Cmu2.north)
                edge[dashed,-]  (Cmu3);
    \path (C33) edge[dashed,-]  (Cmu3)
                edge[dashed]    (Cmu4);
  \end{tikzpicture}
  \caption{Directed acyclic graph of recursive calls, of depth \(\mu =
  \log_2(\rdim)\). Each boxed node shows the matrix dimensions and the
  determinantal degree of a recursive call. Beginning with one call in
  dimension and determinantal degree \((\rdim,\degdet)\), for a given node the
  number above it indicates the number of times a recursive call corresponding
  to this node is made, and the numbers of recursive sub-calls this node
  generates are indicated on both arrows starting from this node.}
  \label{fig:dag_calls}
\end{figure}
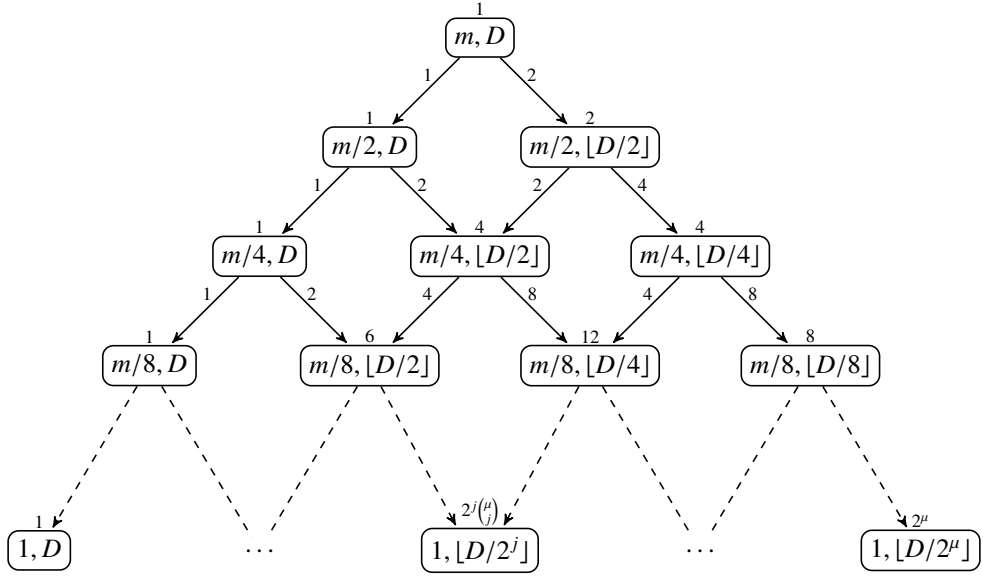

Let \(\mu= \log_2(\rdim)\) and let $K$ be the constant of the $\bigO{\cdot}$
term above. Recalling that \(\rdim\le \degdet\) and \(\lfloor\lfloor
\degdet/2^j \rfloor / 2\rfloor = \lfloor \degdet/2^{j+1} \rfloor\), unrolling
this recurrence to the \(i\)th recursion level for \(0\leq i\leq \mu\) yields

\begin{align*}
  \mathcal{C}(\rdim,\degdet) & \le
  \sum_{j=0}^{i} a_{i,j} \mathcal{C}\left(\frac{\rdim}{2^i},\left\lfloor\frac{\degdet}{2^j}\right\rfloor\right)
  +
  K\sum_{k=0}^{i-1}\sum_{j=0}^{k} a_{k,j} \left(\frac{\rdim}{2^k}\right)^\expmm \timegcd{\frac{\degdet/2^j}{m/2^k}}
  \\
  & \le
  \sum_{j=0}^{i} a_{i,j} \mathcal{C}\left(\frac{\rdim}{2^i},\left\lfloor\frac{\degdet}{2^j}\right\rfloor\right)
  +
  K\rdim^\expmm \timegcd{\frac{\degdet}{\rdim}}\left(
     \sum_{k=0}^{i-1} \sum_{j=0}^{k} a_{k,j}  2^{-k\expmm}\timegcd{2^{k-j}}\right),
\end{align*}
where the last inequality comes from the submultiplicativity of
\(\timegcd{\cdot}\) and the coefficients $a_{i,j}$ satisfy
\[
  \left\{
    \begin{array}{lcl}
      a_{i,0} &=& 1, \\
      a_{i,i} &=& 2^i,  \\
      a_{i,j} &=& a_{i-1,j} + 2a_{i-1,j-1} \text{ for } 0<j<i.
    \end{array}
  \right. 
\]
In \cref{fig:dag_calls}, one can observe the similarity between Pascal's
triangle and the number of calls with parameters \((\rdim/2^i,\degdet/2^j)\).
This translates as a connection between the \(a_{i,j}\)'s and the binomial
coefficients: one can prove by induction that \(a_{i,j} = 2^j \binom{i}{j}\).

Now, by assumption \(\timegcd{d} \in \bigO{d^{\expmm-1-\varepsilon}}\) for some
\(\varepsilon>0\); for such an \(\varepsilon>0\), let \(\tilde{K}\) be a
constant such that \(\timegcd{d} \le \tilde{K} d^{\expmm-1-\varepsilon}\) for
all \(d\ge 0\). Then
\[
  \sum_{k=0}^{i-1} \sum_{j=0}^{k} a_{k,j}  2^{-k\expmm}\timegcd{2^{k-j}}
  \le
  \tilde{K}\sum_{k=0}^{i-1} 2^{-(1+\varepsilon)k}\sum_{j=0}^{k} \frac{a_{k,j}}{2^{j(\expmm-1-\varepsilon)}}
  \le
  \tilde{K}\sum_{k=0}^{i-1} 2^{-(1+\varepsilon)k}\sum_{j=0}^k \binom{k}{j}
  =
  \tilde{K}\sum_{k=0}^{i-1} 2^{-\varepsilon k}
\]
and, defining the constant \(\hat{K} = K \tilde{K} \sum_{k=0}^{+\infty}
2^{-\varepsilon k}\), for \(i=\mu\) we obtain
\[
  \mathcal{C}(\rdim,\degdet) \leq
  \sum_{j=0}^{\mu} a_{\mu,j} \mathcal{C}\left(1, \left\lfloor\frac{\degdet}{2^j}\right\rfloor\right)
  + \hat{K} \rdim^\expmm \timegcd{\frac{\degdet}{\rdim}}
  .
\]
As we have seen above, parameters \((1,d)\) for any \(d \in \NN\) correspond to
base cases with a \(1\times 1\) matrix, and they incur no arithmetic cost (one
might want to consider them to use \(\bigO{1}\) operations each; then the total
cost of these base cases is bounded asymptotically by \(\sum_{j=0}^{\mu}
a_{\mu,j} \leq \rdim^{\log_2(3)}\)). Thus we obtain
\(\mathcal{C}(\rdim,\degdet) = \bigO{\rdim^{\expmm} \timegcd{\degdet /
\rdim}}\), under the assumption \(\rdim \le \degdet\).

For the case \(\degdet < \rdim\) handled at
\cref{algo:determinant:constant_rows_if}, \cref{sec:determinant:algo} showed
that \(\mathcal{C}(\rdim,\degdet) = \mathcal{C}(\hat{\rdim},\degdet) +
\bigO{\rdim^\expmm}\) where \(\hat{\rdim}\) is the number of non-constant rows
of the input matrix. Since \(\hat{\rdim} \le \degdet\), our proof above shows
that \(\mathcal{C}(\hat{\rdim},\degdet)\) is in \(\bigO{\hat{\rdim}^{\expmm}
\timegcd{\degdet/\hat{\rdim}}}\). Now our assumptions on \(\timepm{\cdot}\)
imply in particular that \(\timegcd{\cdot}\) is subquadratic, hence this bound
is in \(\bigO{\hat{\rdim}^{\expmm} (\degdet/\hat{\rdim})^{2}} =
\bigO{\hat{\rdim}^{\expmm-2} \degdet^2} \subseteq \bigO{\rdim^\expmm}\). This
concludes the complexity analysis.

%%%%%%%%%%%%%%%%%%%%%%%%%%%%%%%%%%%%%%%%%%%%%%%%%%%%%%%%%%%%%%%%%%%%%%%%%%%%%

\section{Shifted forms: from reduced to weak Popov}
\label{sec:reduced_to_weak_popov}

This section proves \cref{thm:reduced_to_weak_popov} by generalizing the
approach of \cite[Sec.\,2~and\,3]{SarSto11}, which focuses on the non-shifted
case \(\shifts=\shiftz\). It first uses Gaussian elimination on the
\(\shiftz\)-leading matrix of \(\pmat\) to find a unimodular matrix \(\mat{U}\)
such that \(\mat{U}\pmat\) is in \(\shiftz\)-weak Popov form, and then exploits
the specific form of \(\mat{U}\) to compute \(\mat{U}\pmat\) efficiently. Here
we extend this approach to arbitrary shifts and show how to take into account
the possible unbalancedness of the row degree of \(\pmat\).

First, we generalize \cite[Lem.\,8]{SarSto11} from \(\shifts=\shiftz\) to
an arbitrary \(\shifts\), by describing how \(\mat{U}\) can be obtained by
computing a \(\shiftz\)-weak Popov form of the \(\shifts\)-leading matrix of
\(\pmat\).

\begin{lemma}
  \label{lem:reduced_to_weak_popov}
  Let \(\shifts \in \ZZ^\cdim\) and let \(\pmat \in \pmatRing[\rdim][\cdim]\)
  be \(\shifts\)-reduced with \(\rdeg[\shifts]{\pmat}\) nondecreasing. There
  exists an invertible lower triangular \(\mat{T} \in \matRing[\rdim]\) which
  can be computed in \(\bigO{\rdim^{\expmm-1}\cdim}\) operations in \(\field\)
  and is such that \(\mat{T} \: \lmat[\shifts]{\pmat}\) is in
  \(\shiftz\)-unordered weak Popov form. For any such matrix \(\mat{T}\),
  \begin{itemize}
    \item \(\mat{U} = \xDiag{\shiftt} \mat{T} \xDiag{-\shiftt}\) has polynomial
      entries and is unimodular, where \(\shiftt = \rdeg[\shifts]{\pmat}\),
    \item \(\rdeg[\shifts]{\mat{U}\pmat}=\rdeg[\shifts]{\pmat}\)
      and \(\lmat[\shifts]{\mat{U}\pmat} = \mat{T} \:
      \lmat[\shifts]{\pmat}\),
    \item \(\mat{U} \pmat\) is in \(\shifts\)-unordered weak Popov form.
  \end{itemize}
\end{lemma}
\begin{proof}
  Consider the matrix \(\lmat[\shifts]{\pmat}\revmat{\cdim}\) and its
  generalized Bruhat decomposition \(\lmat[\shifts]{\pmat}\revmat{\cdim} =
  \mat{C}\mat{P}\mat{R}\) as defined in~\cite{MaHe07}: \(\mat{C}\in\matRing\)
  is in column echelon form, \(\mat{R}\in\matRing[\rdim][\cdim]\) is in row
  echelon form, and \(\mat{P}\in\matRing\) is a permutation matrix. Therefore
  \(\revmat{\rdim}\mat{R}\revmat{\cdim}\) is in \(\shiftz\)-weak Popov form
  (see the paragraph before \cref{lem:weak_popov_lmat} in
  \cref{sec:prelim:popov}), and \(\mat{P}\mat{R}\revmat{\cdim}\) is in
  \(\shiftz\)-unordered weak Popov form. Since \(\lmat[\shifts]{\pmat}\) has
  full row rank, \(\mat{C}\) is lower triangular and invertible, hence
  \(\mat{P}\mat{R}\revmat{\cdim} = \mat{C}^{-1}
  \lmat[\shifts]{\pmat}\) which proves the existence of \(\mat{T}
  = \mat{C}^{-1}\). Computing the decomposition costs
  \(\bigO{\rdim^{\expmm-1}\cdim}\) operations \cite[Cor.\,25]{DuPeSu17} while
  inverting \(\mat{C}\) costs \(\bigO{\rdim^{\expmm}}\) operations.
  Alternatively, \cite[Sec.\,3]{SarSto11} shows how to compute \(\mat{T}\)
  within the same cost bound using an LUP decomposition with  a modified
  pivoting strategy.

  For any such matrix \(\mat{T}\), write \(\mat{T} = (T_{ij})_{ij}\) and
  \(\shiftt = (t_i)_i\). Then the entry \((i,j)\) of \(\mat{U}\) is \(T_{ij}
  \var^{t_i-t_j}\). Thus, \(\mat{U}\) is lower triangular in
  \(\field(\var)^{\rdim \times \rdim}\) with diagonal entries in
  \(\field\setminus\{0\}\), and since \(\shiftt = \rdeg[\shifts]{\pmat}\) is
  nondecreasing, \(\mat{U}\) is a unimodular matrix in \(\pmatRing[\rdim]\).

  Now, consider the nonnegative shift \(\tuple{u} = \shiftt -
  (\min(\shiftt),\ldots,\min(\shiftt)) \in \NN^\rdim\), and note that \(\mat{U}
  = \xDiag{\tuple{u}} \mat{T} \xDiag{-\tuple{u}}\). (The introduction of
  \(\tuple{u}\) is to circumvent the fact that we have not defined the row
  degree of a matrix over the Laurent polynomials, a notion which would be
  needed if we used \(\xDiag{\shiftt}\) rather than \(\xDiag{\tuple{u}}\) in
  \cref{eqn:red_to_weak:defU}.) Since \(\pmat\) is in \(\shifts\)-reduced
  form, the predictable degree property yields
  \begin{equation}
    \label{eqn:red_to_weak:predictable}
    \rdeg[\shifts]{\mat{U}\pmat} = \rdeg[\shiftt]{\mat{U}} = \rdeg[\tuple{u}]{\mat{U}} + (\min(\shiftt),\ldots,\min(\shiftt)).
  \end{equation}
  On the other hand,
  \begin{equation}
    \label{eqn:red_to_weak:defU}
    \rdeg[\tuple{u}]{\mat{U}} = \rdeg[\shiftz]{\mat{U}\xDiag{\tuple{u}}}
    = \rdeg[\shiftz]{\xDiag{\tuple{u}}\mat{T}} = \tuple{u} = \rdeg[\shifts]{\pmat} - (\min(\shiftt),\ldots,\min(\shiftt)),
  \end{equation}
  where the third equality follows from the fact that \(\mat{T}\) is a constant
  matrix with no zero row. Then, from
  \cref{eqn:red_to_weak:predictable,eqn:red_to_weak:defU}, we obtain
  \(\rdeg[\shifts]{\mat{U}\pmat}=\rdeg[\shifts]{\pmat} = \shiftt\).

  Then, \(\lmat[\shifts]{\mat{U}\pmat}\) is formed by the coefficients of
  nonnegative degree of
  \[
    \xDiag{-\rdeg[\shifts]{\mat{U}\pmat}}\mat{U}\pmat\xDiag{\shifts} =
    \xDiag{-\shiftt}\mat{U}\pmat\xDiag{\shifts} = \mat{T}
    \xDiag{-\shiftt} \pmat \xDiag{\shifts}.
  \]
  Since \(\mat{T}\) is constant and \(\lmat[\shifts]{\pmat}\) is formed by the
  coefficients of nonnegative degree of \(\xDiag{-\shiftt} \pmat
  \xDiag{\shifts}\), we obtain that \(\lmat[\shifts]{\mat{U}\pmat} = \mat{T}
  \:\lmat[\shifts]{\pmat}\). The third item then directly follows from
  \cref{lem:weak_popov_lmat}.
\end{proof}

Knowing \(\mat{T}\), and therefore \(\mat{U}\), the remaining difficulty is to
efficiently compute \(\mat{U}\pmat\). For this, in the case
\(\shifts=\shiftz\), the approach in \cite{SarSto11} has a cost bound which
involves the maximum degree \(d = \deg(\pmat) = \max(\rdeg[\shiftz]{\pmat})\)
\cite[Thm.\,13]{SarSto11}, and uses the following steps:
\begin{itemize}
  \item first compute \(\var^d \xDiag{-\rdeg[\shiftz]{\pmat}} \mat{U} \pmat =
    \mat{T}(\xDiag{(d,\ldots,d)-\rdeg[\shiftz]{\pmat}} \pmat)\);
  \item then scale the rows via the left multiplication by
    \(\var^{-d}\xDiag{\rdeg[\shiftz]{\pmat}}\).
\end{itemize}
While the scaling does not use arithmetic operations, the first step asks to
multiply the constant matrix \(\mat{T}\) by an \(\rdim\times\cdim\) polynomial
matrix of degree \(d\): this costs \(\bigO{\rdim^{\expmm-1}\cdim d}\)
operations in \(\field\). Such a cost would not allow us to reach our target
bound for the computation of characteristic polynomials, since \(d\) may be too
large, namely when \(\pmat\) has unbalanced row degrees.

Below, we refine the approach to better conform with the row degrees of
\(\pmat\). This leads to an improvement of the above cost bound to
\(\bigO{\rdim^{\expmm-1}\cdim (1 + \degdet/\rdim)}\) where \(\degdet =
\sumTuple{\rdeg[\shiftz]{\pmat}}\), thus involving the average row degree of
\(\pmat\) instead of its maximum degree \(d\). For this, we follow the strategy
of splitting the rows of \(\pmat\) into subsets having degrees in prescribed
intervals; when these intervals get further away from the average row degree of
\(\pmat\), the corresponding subset contains a smaller number of rows. Earlier
works using a similar strategy such as \cite{ZhoLab09},
\cite[Sec.\,2.6]{ZhLaSt12}, and \cite[Sec.\,4]{GioNei18}, are not directly
applicable to the problem here. Furthermore, we exploit the fact that \(\pmat\)
has nondecreasing row degree and that \(\mat{T}\) has a lower triangular shape
to avoid logarithmic factors in the cost bound. This results in
\Cref{algo:reduced_to_weak_popov}.

\begin{algorithm}
  \caption{\algoname{ReducedToWeakPopov}$(\pmat,\shifts)$}
  \label{algo:reduced_to_weak_popov}

  \begin{algorithmic}[1]
    \Require{a matrix \(\pmat\in\pmatRing[\rdim][\cdim]\), a shift
    \(\shifts\in\ZZ^\cdim\) such that \(\pmat\) is in \(\shifts\)-reduced form.}
    \Ensure{an \(\shifts\)-weak Popov form of \(\pmat\).}

    \State \CommentLine{Step 1: Ensure nonnegative shift and nondecreasing \(\shifts\)-row degree}

    \State \(\tuple{\hat{s}} \in \NN^\cdim \assign \shifts - (\min(\shifts),\ldots,\min(\shifts))\)
    \label{step:reduced_to_weak_popov:prepare_beg} 

    \State \((\mat{B},\shiftt) \in \pmatRing[\rdim][\cdim] \times \NN^\rdim
    \assign\) matrix \(\mat{B}\) obtained from \(\pmat\) by row permutation
    such that the tuple \(\shiftt = \rdeg[\tuple{\hat{s}}]{\mat{B}}\) is
    nondecreasing
    \label{step:reduced_to_weak_popov:prepare_end} 

    \State \CommentLine{Step 2: Compute the factor \(\mat{T}\) in the unimodular transformation \(\mat{U} = \xDiag{\shiftt} \mat{T} \xDiag{-\shiftt}\)}

    \State \(\mat{L} \in \matRing[\rdim][\cdim] \assign
    \lmat[\tuple{\hat{s}}]{\mat{B}}\), that is, the entry \((i,j)\) of
    \(\mat{L}\) is the coefficient of degree \(t_i - s_j\) of the entry \((i,j)\)
    of \(\mat{B}\), where \(\tuple{\hat{s}} = (s_1,\ldots,s_\cdim)\) and
    \(\shiftt = (t_1,\ldots,t_\rdim)\)
    \label{step:reduced_to_weak_popov:unimodular_beg} 

    \State \(\mat{T} \in \matRing[\rdim] \assign\) invertible lower triangular
    matrix such that \(\mat{T} \mat{L}\) is in \(\shiftz\)-unordered weak Popov
    form \Comment{can be computed via a generalized Bruhat decomposition, see
    \cref{lem:reduced_to_weak_popov}}
    \label{step:reduced_to_weak_popov:unimodular_end} 

    \State \CommentLine{Step 3: Compute the product \(\mat{P} = \mat{U}\mat{B}\)}

    \State \(\degdet \assign t_1 + \cdots + t_\rdim\); \(K \assign \lfloor \log_2(\rdim t_\rdim / \degdet) \rfloor + 1\)
    \Comment{\(K = \min\{k \in \ZZp \mid t_\rdim < 2^k \degdet/\rdim\}\)}
    \label{step:reduced_to_weak_popov:product_beg}

    \State \(i_0 \assign 0\); \(i_k \assign \max \{i \mid t_i < 2^k \degdet/\rdim\}\) for \(1\le k \le K\)
    \Comment{\(0=i_0 < i_1 \le \cdots \le i_{K-1} < i_K = \rdim\)}
    \State \(\mat{P} \assign\) zero matrix in \(\pmatRing[\rdim][\cdim]\)
    \For{\(k\) from \(1\) to \(K\)}
    \State \(\mathcal{R} \assign \{i_{k-1}+1, \ldots, \rdim\}\); ~ \(\mathcal{C} \assign \{i_{k-1}+1, \ldots, i_k\}\); ~ \(\theta \assign t_{i_k} = \max(\subTuple{\shiftt}{\mathcal{C}})\)
    \State \(\matrows{\mat{P}}{\mathcal{R}} \assign \matrows{\mat{P}}{\mathcal{R}} \,+\, \xDiag{\subTuple{\shiftt}{\mathcal{R}}-(\theta,\ldots,\theta)} \, \submat{\mat{T}}{\mathcal{R}}{\mathcal{C}} \, \xDiag{(\theta,\ldots,\theta)-\subTuple{\shiftt}{\mathcal{C}}} \, \matrows{\mat{B}}{\mathcal{C}}\)
    \label{step:reduced_to_weak_popov:product_end}
    \EndFor

    \State \Return the row permutation of \(\mat{P}\) which has
    increasing \(\tuple{\hat{s}}\)-pivot index
  \end{algorithmic}
\end{algorithm}

\begin{proof}[Proof of \cref{thm:reduced_to_weak_popov}.]
  The first step builds a nonnegative shift \(\tuple{\hat{s}}\) which only
  differs from \(\shifts\) by an additive constant, and builds a matrix
  \(\mat{B}\) which is a row permutation of \(\pmat\); hence any
  \(\tuple{\hat{s}}\)-weak Popov form of \(\mat{B}\) is an
  \(\shifts\)-weak Popov form of \(\pmat\). Since the
  \(\tuple{\hat{s}}\)-row degree \(\shiftt\) of \(\mat{B}\) is nondecreasing,
  the construction of \(\mat{T}\) at
  Step 2 ensures that the matrix
  \(\mat{U} = \xDiag{\shiftt} \mat{T} \xDiag{-\shiftt} \in \pmatRing[\rdim]\)
  is unimodular and such that \(\mat{U}\mat{B}\) is in \(\tuple{\hat{s}}\)-unordered weak
  Popov form, according to \cref{lem:reduced_to_weak_popov}. Thus, for the
  correctness, it remains to prove that the matrix \(\mat{P}\) computed at
  Step 3 is \(\mat{P} = \mat{U}\mat{B}\).
  
  Using notation from the algorithm, define \(\mathcal{C}_k = \{i_{k-1}+1,
  \ldots, i_k\}\) for \(1\le k\le K\). The \(\mathcal{C}_k\)'s are pairwise
  disjoint and such that \(\{1,\ldots,\rdim\} = \mathcal{C}_1 \cup \cdots \cup
  \mathcal{C}_K\). Then, slicing the columns of \(\mat{T}\) according to these
  sets of column indices, we obtain
  \[
    \mat{U}\mat{B} = \xDiag{\shiftt} \mat{T} \xDiag{-\shiftt} \mat{B}
    = \sum_{1 \le k \le K} \xDiag{\shiftt} \matcols{\mat{T}}{\mathcal{C}_k} \xDiag{-\subTuple{\shiftt}{\mathcal{C}_k}} \matrows{\mat{B}}{\mathcal{C}_k}.
  \]
  Furthermore, since \(\mat{T}\) is lower triangular, all rows of
  \(\matcols{\mat{T}}{\mathcal{C}_k}\) with index not in \(\mathcal{R}_k =
  \{i_{k-1}+1, \ldots, \rdim\}\) are zero. Hence, more precisely,
  \[
    \mat{U}\mat{B}
    = \sum_{1 \le k \le K} \begin{bmatrix} \matz_{i_{k-1} \times \cdim} \\[0.1cm] \xDiag{\subTuple{\shiftt}{\mathcal{R}_k}} \, \submat{\mat{T}}{\mathcal{R}_k}{\mathcal{C}_k} \, \xDiag{-\subTuple{\shiftt}{\mathcal{C}_k}} \, \matrows{\mat{B}}{\mathcal{C}_k} \end{bmatrix}.
  \]
This formula corresponds to the slicing of the product
\(\mat{P}=\xDiag{\shiftt} \mat{T} \xDiag{-\shiftt} \mat{B}\) in blocks as follows: 
\[
\setlength{\arraycolsep}{.8\arraycolsep}
\left[
  \begin{array}{cc;{1pt/3pt}c}
      \multicolumn{3}{l}{*} \\
      * & \multicolumn{2}{l}{*} \\
      \cdashline{3-3}[1pt/3pt]
        & \phantom{hi} & \rule[-1.2em]{0em}{3em}~~~~\xDiag{\subTuple{\shiftt}{\mathcal{R}_k}}~~
  \end{array}
  \right]
  \left[
    \begin{array}{cc;{1pt/3pt}c;{1pt/3pt}cc}
      \multicolumn{5}{l}{*} \\
      * & \multicolumn{4}{l}{*} \\
      \cdashline{3-3}[1pt/3pt]
      * & * & & \\
      * & * & \submat{\mat{T}}{\mathcal{R}_k}{\mathcal{C}_k}\\
      * & * & & * \\
      * & * & & * & * \\
    \end{array}
  \right]
  \left[
    \begin{array}{ccccc}
      * &  & \\
        &* & \\
        & &
            \begin{array}{;{1pt/3pt}c;{1pt/3pt}}
              \hdashline[1pt/3pt]
              \rule[-0.8em]{0em}{2.2em}
              \xDiag{-\subTuple{\shiftt}{\mathcal{C}_k}}\!\! \\
              \hdashline[1pt/3pt]
            \end{array}
        \\
    &&& *\\
    &&&&*
  \end{array}\right]
  \left[
    \begin{array}{ccc}
      \hphantom{larg} & * & \hphantom{rger} \\
      & * & \\
    \hdashline[1pt/3pt]
      & \rule[-0.8em]{0em}{2.2em} \matrows{\mat{B}}{\mathcal{C}_{k}} \\
    \hdashline[1pt/3pt]
      & * & \\
      & * &
    \end{array}
    \right].
\]
  The \algorithmicfor{} loop at Step 3 computes
  \(\mat{P}\) by following this formula, hence \(\mat{P} = \mat{U}\mat{B}\).
  (Note indeed that the scaling by \((\theta,\ldots,\theta)\) in the algorithm
  can be simplified away and is just there to avoid computing with Laurent
  polynomials.)

  Concerning the complexity bound, we first note that the quantity \(\degdet\)
  defined in the algorithm is the same as that in
  \cref{thm:reduced_to_weak_popov}; indeed, \(\rdeg[\tuple{\hat{s}}]{\mat{B}} =
  \rdeg[\shifts]{\mat{B}} - (\min(\shifts),\ldots,\min(\shifts))\), hence
  \[
    \degdet = t_1 + \cdots + t_\rdim = \sumTuple{\rdeg[\tuple{\hat{s}}]{\mat{B}}}
    = \sumTuple{\rdeg[\shifts]{\mat{B}}} - \rdim \cdot \min(\shifts)
    = \sumTuple{\rdeg[\shifts]{\pmat}} - \rdim \cdot \min(\shifts).
  \]
  By \cref{lem:reduced_to_weak_popov},  Step 2 uses
  \(\bigO{\rdim^{\expmm-1}\cdim}\) operations. As for Step 3, its cost
  follows from bounds on the cardinalities of \(\mathcal{C}_k\) and
  \(\mathcal{R}_k\). Precisely, we have \(\mathcal{C}_k \subseteq
  \mathcal{R}_k\), and \(\min(\mathcal{R}_k) > i_{k-1}\) implies that each
  entry of the subtuple \(\subTuple{\shiftt}{\mathcal{R}_k}\) is at least
  \(2^{k-1}\degdet/\rdim\). Hence \(\card{\mathcal{R}_k}\cdot 2^{k-1}
  \degdet/\rdim \le \sum_{i \in \mathcal{R}_k} t_i \;\le t_1 + \cdots + t_\rdim
  = \degdet\), and \(\card{\mathcal{C}_k} \le \card{\mathcal{R}_k} \le
  \rdim/2^{k-1}\).

  Then, in the product
  \(\xDiag{\subTuple{\shiftt}{\mathcal{R}_k}-(\theta,\ldots,\theta)} \,
  \submat{\mat{T}}{\mathcal{R}_k}{\mathcal{C}_k} \,
  \xDiag{(\theta,\ldots,\theta)-\subTuple{\shiftt}{\mathcal{C}_k}} \,
  \matrows{\mat{B}}{\mathcal{C}_k}\), the left multiplication by
  \(\xDiag{\subTuple{\shiftt}{\mathcal{R}_k}-(\theta,\ldots,\theta)}\) does not
  use arithmetic operations; the matrix
  \(\submat{\mat{T}}{\mathcal{R}_k}{\mathcal{C}_k}\) is over \(\field\) and has
  at most \(\rdim/2^{k-1}\) rows and at most \(\rdim/2^{k-1}\) columns; and the
  matrix \(\xDiag{(\theta,\ldots,\theta)-\subTuple{\shiftt}{\mathcal{C}_k}} \,
  \matrows{\mat{B}}{\mathcal{C}_k}\) is over \(\polRing\) and has \(\cdim\)
  columns and at most \(\rdim/2^{k-1}\) rows. Furthermore, the latter matrix
  has degree at most \(\theta\): indeed, \(\tuple{\hat{s}} \ge \shiftz\)
  implies that \(\shiftt = \rdeg[\tuple{\hat{s}}]{\mat{B}} \ge
  \rdeg[\shiftz]{\mat{B}}\), hence in particular
  \(\subTuple{\shiftt}{\mathcal{C}_k} \ge
  \rdeg[\shiftz]{\matrows{\mat{B}}{\mathcal{C}_k}}\). Recall that \(\theta =
  t_{i_k} \le 2^k \degdet/\rdim\) holds, by definition of \(\theta\) and
  \(i_k\). From these bounds on the dimensions and the degrees of the involved
  matrices, and using the fact that \(\rdim/2^{k-1} \le \rdim \le \cdim\), it
  follows that computing the product
  \(\xDiag{\subTuple{\shiftt}{\mathcal{R}_k}-(\theta,\ldots,\theta)} \,
  (\submat{\mat{T}}{\mathcal{R}_k}{\mathcal{C}_k} \,
  (\xDiag{(\theta,\ldots,\theta)-\subTuple{\shiftt}{\mathcal{C}_k}} \,
  \matrows{\mat{B}}{\mathcal{C}_k}))\) uses \(\bigO{(\rdim/2^{k-1})^{\expmm-1}
  \cdim (\theta+1)} \subseteq \bigO{(\rdim/2^{k-1})^{\expmm-1} \cdim
(2^k\degdet/\rdim + 1)}\) operations in \(\field\). Thus, since \(\expmm>2\),
the cost of Step 3 is
  \begin{align*}
    \bigO{\sum_{1\le k\le K} \left(\frac{\rdim}{2^{k-1}}\right)^{\expmm-1}
    \cdim \left( \frac{2^k \degdet}{\rdim} + 1 \right)}
    & \subseteq
    \bigO{\rdim^{\expmm-2} \cdim \degdet \left(\sum_{1\le k\le K} 2^{k(2-\expmm)}\right) + \rdim^{\expmm-1} \cdim \left(\sum_{1\le k\le K} 2^{k(1-\expmm)}\right)}
    \\
    & \subseteq
    \bigO{\rdim^{\expmm-2} \cdim \degdet + \rdim^{\expmm-1} \cdim}. \qedhere
  \end{align*}
\end{proof}

\section{Shifted forms: from weak Popov to Popov}
\label{sec:weak_popov_to_popov}

This section culminates in \cref{sec:weak_popov_to_popov:proof} with the
description of Algorithm \algoname{WeakPopovToPopov} and a proof of
\cref{thm:weak_popov_to_popov}. Based on \cite[Lem.\,14]{SarSto11} (which
extends to the shifted case), the result in \cref{thm:weak_popov_to_popov} can
easily be used to solve the same problem in the rectangular case with an
\(\rdim\times\cdim\) matrix \(\pmat\); while this is carried out in
Algorithm~\algoname{WeakPopovToPopov}, it is out of the main scope of this
paper and thus for conciseness we only give a cost analysis in the case
\(\rdim=\cdim\).

Our approach is to obtain the \(-\shifts\)-Popov form of \(\pmat\) from a
shifted reduced kernel basis of some matrix \(\sys\) built from \(\pmat\). This
fact is substantiated in \cref{sec:weak_popov_to_popov:via_kernel}, which also
proves that we have precise a priori information on the degrees of this sought
kernel basis.

A folklore method for kernel basis computation is to find an approximant basis
at an order sufficiently large so that it contains a kernel basis as a
submatrix. More precisely, assuming we know a list of bounds \(\shifts\) such
that there exists a basis of \(\modKer{\sys}\) with column degree bounded by
\(\shifts\), the following algorithm computes such a kernel basis which is
furthermore \(-\shifts\)-reduced:
\begin{itemize}
    \item \(\orders \assign \cdeg[\shifts]{\sys}+1\);
    \item \(\appbas  \assign\) basis of \(\modApp{\orders}{\sys}\) in
      \(-\shifts\)-reduced form;
    \item \Return the submatrix of \(\appbas\) formed by its rows with
      nonpositive \(-\shifts\)-degree.
\end{itemize}
The idea is that any row \(\app\) of \(\appbas\) is such that \(\app \sys =
\matz \bmod \xDiag{\orders}\), and if it further satisfies \(\cdeg{\app} \le
\shifts\) then \(\cdeg{\app\sys} \le \cdeg[\shifts]{\sys} < \orders\), so that
\(\app\sys=\matz\) holds. Here the complexity mainly depends on
\(\sumTuple{\shifts}\) and \(\sumTuple{\cdeg[\shifts]{\sys}}\), quantities that
are often large in which case the algorithm of Zhou et al.~\cite{ZhLaSt12} is
more efficient. Nevertheless there are cases, such as the one arising in this
section, where both sums are controlled, and elaborating over this approach
leads to an efficient algorithm.

In order to propose Algorithm~\algoname{KnownDegreeKernelBasis} in
\cref{sec:weak_popov_to_popov:algo}, efficiently computing the kernel basis
using essentially a single call to \algoname{PM-Basis}, we transform the input
into one with a balanced shift and a balanced order. Here and in what follows,
for a nonnegative tuple \(\shiftt\in\NN^\cdim\), we say that \(\shiftt\) is
\emph{balanced} if \(\max(\shiftt) \in \bigO{\sumTuple{\shiftt}/\cdim}\),
meaning that the maximum entry in \(\shiftt\) is not much larger than the
average of all entries of \(\shiftt\).
\Cref{sec:weak_popov_to_popov:output_parlin} deals with the shifts by
describing a transformation of the input inspired from
\cite[Sec.\,3]{Storjohann06}, allowing us to reduce to the case of a shift
\(\shifts\) whose entries are balanced.
\cref{sec:weak_popov_to_popov:ovlp_parlin} deals with balancing the order
\(\orders\) by performing the overlapping partial linearization of
\cite[Sec.\,2]{Storjohann06}.

For the latter transformation, as noted above, we assume that there
exists a basis of the considered kernel whose column degree is bounded by
\(\shifts\), or equivalently that \(-\shifts\)-reduced kernel bases have
nonpositive \(-\shifts\)-row degree. On the other hand, for the first
transformation we must ensure that \(-\shifts\)-reduced kernel bases have
nonnegative \(-\shifts\)-row degree. Thus our algorithm works under the
requirement that \(-\shifts\)-reduced bases of \(\modKer{\sys}\) have
\(-\shifts\)-row degree exactly \(\shiftz\), hence its name. This is a
restriction compared to \cite[Algo.\,1]{ZhoLab13} which only assumes that
\(-\shifts\)-reduced bases of \(\modKer{\sys}\) have nonpositive
\(-\shifts\)-row degree and has to perform several approximant basis
computations to recover the whole kernel basis, as outlined in the
introduction. In short, we have managed to exploit the fact that we have better
a priori knowledge of degrees in kernel bases than in the context of
\cite{ZhoLab13}, leading to a faster kernel computation which, when used in our
determinant algorithm, brings no logarithmic factor in the complexity.

\subsection{Normalization via kernel basis computation}
\label{sec:weak_popov_to_popov:via_kernel}

We normalize the matrix \(\pmat\) into its \(-\shifts\)-Popov form \(\mat{P}\)
using a kernel basis computation, an approach already used in a context similar
to ours in the non-shifted case in \cite[Lem.\,19]{SarSto11}. Roughly, this
stems from the fact that the identity \(\mat{U} \pmat = \mat{P}\) with a
unimodular \(\mat{U}\) can be rewritten as
\[
  \begin{bmatrix}
    \mat{U} & \mat{P}
  \end{bmatrix}
  \begin{bmatrix}
    \pmat \\
    -\idMat{\rdim}
  \end{bmatrix}
  =
  \matz;
\]
and that, for a well-chosen shift, one retrieves \([\mat{U} \;\; \mat{P}]\) as a shifted
reduced kernel basis. The next statement gives a choice of shift suited to our
situation, and describes the degree profile of such kernel bases. The focus on
the shift \(-\pivDegs\) comes from the fact that any \(-\pivDegs\)-reduced form
\(\mat{R}\) of \(\pmat\) is only a constant transformation away from being the
\(-\shifts\)-Popov form \(\mat{P}\) of \(\pmat\)
\cite[see][Lem.\,4.1]{JeNeScVi16}.

\begin{lemma}
  \label{lem:normalize}
  Let \(\shifts\in\ZZ^\rdim\), let \(\pmat\in\pmatRing[\rdim]\) be in
  \(-\shifts\)-weak Popov form, let \(\pivDegs \in \NN^\rdim\) be the
  \(-\shifts\)-pivot degree of \(\pmat\), and assume that \(\shifts \ge
  \pivDegs\). Let \(\mat{R}\) be a \(-\pivDegs\)-weak Popov form of \(\pmat\)
  and let \(\mat{U}\) be the unimodular matrix such that \(\mat{U} \pmat =
  \mat{R}\). Let further \(\cdegs = (\shifts-\pivDegs,\pivDegs) \in
  \ZZ^{2\rdim}\) and \(\sys = \left[ \begin{smallmatrix} \pmat \\
  -\idMat{\rdim}\end{smallmatrix} \right] \in \pmatRing[2\rdim][\rdim]\).
  Then,
  \begin{itemize}
    \item the \(-\cdegs\)-pivot profile of \([\mat{U} \;\; \mat{R}]\)
      is \((\rdim+j, \pivDeg_j)_{1\le j\le\rdim}\),
    \item \([\mat{U} \;\; \mat{R}]\) is in \(-\cdegs\)-weak Popov form
      with \(\rdeg[-\cdegs]{[\mat{U} \;\; \mat{R}]} = \shiftz\) and
      \(\cdeg{[\mat{U} \;\; \mat{R}]} \le \cdegs\),
    \item \([\mat{U} \;\; \mat{R}]\) is a basis of \(\modKer{\sys}\).
  \end{itemize}
\end{lemma}
\begin{proof}
  First, we prove that \(\mat{R}\) has \(-\pivDegs\)-pivot degree \(\pivDegs\);
  note that this implies \(\rdeg[-\pivDegs]{\mat{R}}=\shiftz\) and
  \(\cdeg{\mat{R}}=\pivDegs\) since \(\mat{R}\) is in \(-\pivDegs\)-weak Popov
  form. \cref{lem:popov_pivot} shows that the \(-\shifts\)-Popov
  form \(\mat{P}\) of \(\pmat\) has the same \(-\shifts\)-pivot degree as
  \(\pmat\), that is, \(\pivDegs\). Hence, by definition of Popov forms,
  \(\cdeg{\mat{P}} = \pivDegs\). Then, \cite[Lem.\,4.1]{JeNeScVi16} states
  that \(\mat{P}\) is also in \(-\pivDegs\)-Popov form. Since \(\mat{R}\) is a
  \(-\pivDegs\)-weak Popov form of \(\mat{P}\), by
  \cref{lem:popov_pivot} it has the same \(-\pivDegs\)-pivot degree as
  \(\mat{P}\), that is, \(\pivDegs\).
  
  Now, by the predictable degree property and since \(\rdeg[-\shifts]{\pmat} =
  -\shifts+\pivDegs\),
  \[
    \rdeg[-\shifts+\pivDegs]{\mat{U}} =\rdeg[-\shifts]{\mat{U}\pmat} =
    \rdeg[-\shifts]{\mat{R}} \le \rdeg[-\pivDegs]{\mat{R}} = \shiftz,
  \]
  where the inequality holds because \(-\shifts\le-\pivDegs\).  Thus, by choice
  of \(\cdegs\), the \(-\cdegs\)-pivot entries of \([\mat{U} \;\; \mat{R}]\)
  are the \(-\pivDegs\)-pivot entries of its submatrix \(\mat{R}\); this proves
  the first item.

  Then, the second item follows: the matrix \([\mat{U} \;\; \mat{R}]\) is in
  \(-\cdegs\)-weak Popov form since its \(-\cdegs\)-pivot index is
  increasing; its \(-\cdegs\)-row degree is equal to the \(-\pivDegs\)-row
  degree of \(\mat{R}\) which is \(\shiftz\); and \(\rdeg[-\cdegs]{[\mat{U}
  \;\; \mat{R}]} = \shiftz\) implies \(\cdeg{[\mat{U} \;\; \mat{R}]} \le
  \cdegs\) by \cref{lem:rdegcdeg}.

  Let \([\kerbas_1 \;\; \kerbas_2]\in\pmatRing[\rdim][2\rdim]\) be a basis
  \(\modKer{\sys}\) (it has \(\rdim = 2\rdim-\rdim\) rows since \(\sys\) is
  \(2\rdim\times\rdim\) and has rank \(\rdim\)). Since \(\mat{U} \pmat =
  \mat{R}\), the rows of \([\mat{U} \;\; \mat{R}]\) are in this kernel. As a
  result, there exists a matrix \(\mat{V}\in\pmatRing[\rdim]\) such that
  \([\mat{U} \;\; \mat{R}] = \mat{V} [\kerbas_1 \;\; \kerbas_2]\). In
  particular, \(\mat{U} = \mat{V} \kerbas_1\), and since \(\mat{U}\) is
  unimodular, this implies that \(\mat{V}\) is unimodular as well. Thus, the
  basis \([\kerbas_1 \;\; \kerbas_2]\) is unimodularly equivalent to \([\mat{U}
  \;\; \mat{R}]\), and the latter matrix is also a basis of \(\modKer{\sys}\).
\end{proof}

One may note similarities with \cite[Lem.\,5.1]{JeNeScVi17}, which is about
changing the shift of reduced forms via kernel basis computations. Here, we
consider \(\pmat\) in \(-\shifts\)-reduced form and are interested in its
\(-\pivDegs\)-reduced forms \(\mat{R}\): we change \(-\shifts\) into
\(-\pivDegs = -\shifts + (\shifts-\pivDegs)\), with a nonnegative difference
\(\shifts-\pivDegs\ge \shiftz\). Still, the above lemma does not follow from
\cite{JeNeScVi17} since here the origin shift \(-\shifts\) is nonpositive and
thus we cannot directly incorporate it in the matrix \(\sys\) by considering
\(\pmat\xDiag{-\shifts}\).

The next corollary uses notation from \cref{lem:normalize} and shows how to
obtain the \(-\shifts\)-Popov form of \(\pmat\) via a \(-\cdegs\)-reduced basis
of the kernel of \(\sys\). Then
\cref{sec:weak_popov_to_popov:output_parlin,sec:weak_popov_to_popov:ovlp_parlin,sec:weak_popov_to_popov:algo}
focus on the efficient computation of such a kernel basis.

\begin{corollary}
  \label{cor:normalize}
  Let \([\matt{U} \;\; \matt{R}] \in \pmatRing[\rdim][2\rdim]\) be a
  \(-\cdegs\)-reduced basis of \(\modKer{\sys}\). Then, \(\matt{R}\) is a
  \(-\pivDegs\)-reduced form of \(\pmat\), and \(\mat{P} =
  (\lmat[-\pivDegs]{\matt{R}})^{-1}\matt{R}\) is the \(-\shifts\)-Popov form of
  \(\pmat\).
\end{corollary}
\begin{proof}
  It suffices to prove that \(\matt{R}\) is a \(-\pivDegs\)-reduced form of
  \(\pmat\); then, the conclusion follows from \cite[Lem.\,4.1]{JeNeScVi16}.
  Both matrices \([\matt{U} \;\; \matt{R}]\) and \([\mat{U} \;\; \mat{R}]\) are
  \(-\cdegs\)-reduced and thus have the same \(-\cdegs\)-row degree up to
  permutation. \cref{lem:normalize} yields \(\rdeg[-\cdegs]{[\mat{U} \;\;
  \mat{R}]} = \shiftz\), hence \(\rdeg[-\cdegs]{[\matt{U} \;\; \matt{R}]} =
  \shiftz\). In particular, since \(-\cdegs= (\pivDegs-\shifts,-\pivDegs)\), we
  have \(\rdeg[-\pivDegs]{\matt{R}} \le \shiftz\).

  We conclude by applying \cref{lem:minimality_implies_basis} to the row space
  of \(\pmat\), which has rank \(\rdim\), and which has a basis \(\mat{R}\) in
  \(-\pivDegs\)-reduced form with \(-\pivDegs\)-row degree \(\shiftz\). From
  \(\matt{U} \pmat = \matt{R}\), we obtain that the rows of \(\matt{R}\) are in
  this row space. This identity also implies that \(\matt{R}\) has rank
  \(\rdim\), otherwise there would exist a nonzero vector in the left kernel of
  \(\matt{R}\), which would also be in the left kernel of \(\matt{U}\) since
  \(\pmat\) is nonsingular, hence it would be in the left kernel of the full
  row rank matrix \([\matt{U} \;\; \matt{R}]\). The assumptions of the lemma
  are satisfied, and thus \(\matt{R}\) is a \(-\pivDegs\)-reduced form of
  \(\pmat\).
\end{proof}

\subsection{Reducing to the case of balanced pivot degree}
\label{sec:weak_popov_to_popov:output_parlin}

Here, we show that the output column partial linearization, used previously in
algorithms for approximant bases and generalizations of them
\cite[Sec.\,3]{Storjohann06}, \cite{ZhoLab12,JeNeScVi16}, can be applied to
kernel computations when the sought kernel basis has nonnegative shifted row
degree. The main effect of this transformation is to make the shift and output
degrees more balanced, while preserving most other properties of
\((\shifts,\sys)\). The transformation itself, defined below, is essentially a
column-wise \(\var^\degExp\)-adic expansion for a well-chosen integer parameter
\(\degExp\).

\begin{definition}
  \label{dfn:output_parlin}
  Let \(\shifts = (\shift_1,\ldots,\shift_\rdim)\in \NN^\rdim\), and let
  \(\degExp \in \ZZp\). For \(1\le j\le \rdim\), write \(\shift_j =
  (\quoExp_j-1)\degExp + \remExp_j\) with \(\quoExp_j = \lceil
  \shift_j/\degExp\rceil\) and \(1\le \remExp_j\le \degExp\) if \(\shift_j\neq
  0\), and \(\quoExp_j = 1\) and \(\remExp_j = 0\) if \(\shift_j=0\). Define
  \(\expand{\rdim} = \quoExp_1 + \cdots + \quoExp_\rdim\), and the
  expansion-compression matrix \(\expandMat \in
  \pmatRing[\expand{\rdim}][\rdim]\) as
  \[
    \expandMat =
    \left[\begin{smallmatrix}
        1 \\
        \var^\degExp \\
        \svdots \\
        \var^{(\quoExp_1-1)\degExp} \\
         & \sddots \\
         & & 1 \\
         & & \var^\degExp \\
         & & \svdots \\
         & & \var^{(\quoExp_\rdim-1)\degExp}
    \end{smallmatrix}\right].
  \]
  Define also
  \[
    \expand{\shifts} =
    (\,\underbrace{\degExp,\ldots,\degExp,\remExp_1}_{\quoExp_1},\,\,\ldots\,\,,\underbrace{\degExp,\ldots,\degExp,\remExp_\rdim}_{\quoExp_\rdim})
    \in \NN^{\expand{\rdim}}.
  \]
  In this context, for a matrix \(\kerbas \in \pmatRing[\cork][\rdim]\), we
  define the \emph{column partial linearization of \(\kerbas\)} as the unique
  matrix \(\expand{\kerbas} \in \pmatRing[\cork][\expand{\rdim}]\) such that
  \(\kerbas = \expand{\kerbas} \expandMat\) and all the columns of
  \(\expand{\kerbas}\) whose index is not in \(\{\quoExp_1+\cdots+\quoExp_j,
  1\le j\le \rdim\}\) have degree less than \(\degExp\).
\end{definition}

\noindent
We use the notation in this definition in all of
\cref{sec:weak_popov_to_popov:output_parlin}. More explicitly, \(\expand{\kerbas} =
[\expand{\kerbas}_1 \;\; \cdots \;\; \expand{\kerbas}_\rdim]\) where
\(\expand{\kerbas}_{j} \in \pmatRing[\cork][\quoExp_j]\) is the unique matrix
such that
\[
  \matcol{\kerbas}{j} = \expand{\kerbas}_j
  \begin{bmatrix} 1 \\ \var^\degExp \\ \vdots \\ \var^{(\quoExp_j-1)\degExp} \end{bmatrix}
\]
and the first \(\quoExp_j-1\) columns of \(\expand{\kerbas}_j\) have degree
less than \(\degExp\).

This construction originates from \cite[Sec.\,3]{Storjohann06}, where it was
designed for approximant basis computations, in order to make the shift more
balanced at the cost of a small increase of the dimension; this is stated in
\cref{lem:parlin_dims_degs}. As mentioned above, several slightly different
versions of this column partial linearization have been given in the
literature: each version requires some minor adaptations of the original
construction to match the context. Here, in order to benefit from properties
proved in \cite[Sec.\,5.1]{JeaNeiVil20}, we follow the construction in
\cite[Lem.\,5.2]{JeaNeiVil20}: one can check that
\cref{dfn:output_parlin} is a specialization of the construction in that
reference, for the shift \(-\shifts \in \ZZ_{\le 0}^\rdim\) and taking the
second parameter to be \(t=\max(-\shifts)\) so that \(\tuple{t} = -\shifts -
\max(-\shifts) + t = -\shifts\).

\begin{lemma}
  \label{lem:parlin_dims_degs}
  The entries of \(\expand{\shifts}\) are in \(\{0,1,\ldots,\degExp\}\).
  Furthermore, if \(\degExp \ge \sumTuple{\shifts}/\rdim\), then \(\rdim \le
  \expand{\rdim} \le 2\rdim\).
\end{lemma}
\begin{proof}
  The first remark is obvious. For the second one, note that \(1 \le \quoExp_j
  \le 1 + \shift_j / \degExp\) holds by construction, for \(1 \le j \le
  \rdim\).  Hence \(\rdim \le \expand{\rdim} \le \rdim +
  \sumTuple{\shifts}/\degExp \le 2\rdim\).
\end{proof}

Importantly, this column partial linearization behaves well with respect to
shifted row degrees and shifted reduced forms.

\begin{lemma}
  \label{lem:parlin_lmat}
  Let \(\kerbas \in \pmatRing[\cork][\rdim]\) with
  \(\rdeg[-\shifts]{\kerbas}\ge\shiftz\), and let \(\expand{\kerbas} \in
  \pmatRing[\cork][\expand{\rdim}]\) be its column partial linearization.
  Then, row degrees are preserved: \(\rdeg[-\expand{\shifts}]{\expand{\kerbas}}
  = \rdeg[-\shifts]{\kerbas}\). Furthermore, if \(\kerbas\) is in
  \(-\shifts\)-weak Popov form, then \(\expand{\kerbas}\) is in
  \(-\expand{\shifts}\)-weak Popov form.
\end{lemma}
\begin{proof}
  We show that this follows from the second item in
  \cite[Lem.\,5.2]{JeaNeiVil20}; as noted above, the shift \(\tuple{t} =
  (t_1,\ldots,t_\rdim)\) in that reference corresponds to \(-\shifts =
  (-\shift_1,\ldots,-\shift_\rdim)\) here. Let \((\pivInd_i,\pivDeg_i)_{1\le
  i\le\cork}\) denote the \(-\shifts\)-pivot profile of \(\kerbas\). By
  definition of the \(-\shifts\)-pivot index and degree,
  \(\rdeg[-\shifts]{\kerbas} = (\pivDeg_i - \shift_{\pivInd_i})_{1\le
  i\le\cork}\), hence our assumption \(\rdeg[-\shifts]{\kerbas}\ge\shiftz\)
  means that \(\pivDeg_i \ge \shift_{\pivInd_i}\) for \(1\le i\le\cork\). Thus,
  we can apply \cite[Lem.\,5.2]{JeaNeiVil20} to each row of \(\kerbas\),
  from which we conclude that \(\expand{\kerbas}\) has
  \(-\expand{\shifts}\)-pivot profile \((\quoExp_1+\cdots+\quoExp_{\pivInd_i},
  \pivDeg_i - \shift_{\pivInd_i} + \remExp_{\pivInd_i})_{1\le i\le \cork}\).

  First, since the entry of \(-\expand{\shifts}\) at index
  \(\quoExp_1+\cdots+\quoExp_{\pivInd_i}\) is \(-\remExp_{\pivInd_i}\), this
  implies that
  \[
    \rdeg[-\expand{\shifts}]{\expand{\kerbas}} = 
    (\pivDeg_i - \shift_{\pivInd_i} + \remExp_{\pivInd_i} - \remExp_{\pivInd_i})_{1\le i\le\cork}
    =
    (\pivDeg_i - \shift_{\pivInd_i})_{1\le i\le\cork}
    =
    \rdeg[-\shifts]{\kerbas}.
  \]
  Furthermore, if \(\kerbas\) is in \(-\shifts\)-weak Popov form, then
  \((\pivInd_i)_{1\le i\le\cork}\) is increasing, hence
  \((\quoExp_1+\cdots+\quoExp_{\pivInd_i})_{1\le i\le\cork}\) is increasing as
  well and \(\expand{\kerbas}\) is in \(-\expand{\shifts}\)-weak
  Popov form.
\end{proof}

\begin{remark}
  \label{rmk:parlin_lmat}
  One may further note the following properties:
  \begin{itemize}
    \setlength{\itemsep}{0em}
    \item Writing \(\lmat[-\shifts]{\kerbas} = [\matcol{\mat{L}}{1} \;\; \cdots
      \;\; \matcol{\mat{L}}{\rdim}] \in \matRing[k][\rdim]\), we have
      \[
        \lmat[-\expand{\shifts}]{\expand{\kerbas}} =
        [\,\underbrace{\matz \;\; \cdots \;\; \matz \;\; \matcol{\mat{L}}{1}}_{\quoExp_1}
        \;\; \cdots \;\;
        \underbrace{\matz \;\; \cdots \;\; \matz \;\; \matcol{\mat{L}}{\rdim}}_{\quoExp_\rdim}]
        \in \matRing[k][\expand{\rdim}].
      \]
    \item If \(\kerbas\) is in \(-\shifts\)-Popov form, then \(\expand{\kerbas}\)
      is in \(-\expand{\shifts}\)-Popov form.
  \end{itemize}
  These properties are not used here, but for reference we provide a proof in
  \ref{app:parlin_forms}.
\end{remark}

We will also need properties for the converse operation, going from some matrix
\(\mat{P} \in \pmatRing[\rdim][\expand{\rdim}]\) to its compression
\(\mat{P}\expandMat\).

\begin{lemma}
  \label{lem:from_parlin_owP}
  Let \(\mat{P} \in \pmatRing[\cork][\expand{\rdim}]\) have
  \(-\expand{\shifts}\)-pivot profile \((\pivInd_i,\pivDeg_i)_{1\le
  i\le\cork}\) and assume \(\pivInd_i = \quoExp_1+\cdots+\quoExp_{j_i}\) for
  some \(j_i\in\ZZp\), for \(1 \le i \le \cork\). Then,
  \(\rdeg[-\shifts]{\mat{P}\expandMat} = \rdeg[-\expand{\shifts}]{\mat{P}}\),
  and \(\mat{P}\expandMat\) has \(-\shifts\)-pivot profile
  \[
    (j_i, \pivDeg_i+(\quoExp_{j_i}-1)\degExp)_{1\le i\le \cork} = (j_i, \pivDeg_i + \shift_{j_i} - \remExp_{j_i})_{1\le i\le \cork}.
  \]
  If \(\mat{P}\) is in \(-\expand{\shifts}\)-weak Popov form, then
  \(\mat{P}\expandMat\) is in \(-\shifts\)-weak Popov form.
\end{lemma}
\begin{proof}
  The \(-\shifts\)-pivot profile of \(\mat{P}\expandMat\) is directly obtained
  by applying the first item in \cite[Lem.\,5.2]{JeaNeiVil20} to each row of
  \(\mat{P}\). From this \(-\shifts\)-pivot profile, we get
  \[
    \rdeg[-\shifts]{\mat{P}\expandMat} = (\pivDeg_i + \shift_{j_i} - \remExp_{j_i} - \shift_{j_i})_{1 \le i\le \cork}
    = (\pivDeg_i - \remExp_{j_i})_{1 \le i\le \cork}.
  \]
  On the other hand, since the entry of \(-\expand{\shifts}\) at index
  \(\quoExp_1+\cdots+\quoExp_{j_i}\) is \(-\remExp_{j_i}\), the
  \(-\expand{\shifts}\)-pivot profile of \(\mat{P}\) yields
  \(\rdeg[-\expand{\shifts}]{\mat{P}} = (\pivDeg_i - \remExp_{j_i})_{1 \le i\le
  \cork}\). Thus, \(\rdeg[-\shifts]{\mat{P}\expandMat} =
  \rdeg[-\expand{\shifts}]{\mat{P}}\). Now, if \(\mat{P}\) is in
  \(-\expand{\shifts}\)-weak Popov form, then
  \((\quoExp_1+\cdots+\quoExp_{j_i})_{1\le i\le \cork}\) is increasing, which
  implies that \((j_i)_{1\le i\le \cork}\) is increasing.  As a result,
  \(\mat{P}\expandMat\) is in \(-\shifts\)-weak Popov form, since \((j_i)_{1\le
  i\le \cork}\) is the \(-\shifts\)-pivot index of \(\mat{P}\expandMat\).
\end{proof}

Our approach for computing the kernel of \(\sys\) is based on the fact that if
\(\kerbas\) is a basis of \(\modKer{\sys}\) and \(\expand{\kerbas}\) is its
column partial linearization, then \(\kerbas \sys = \expand{\kerbas} \expandMat
\sys = \matz\). This identity shows that the kernel \(\modKer{\expandMat\sys}\)
contains the rows of \(\expand{\kerbas}\), so we may hope to recover
\(\expand{\kerbas}\), and thus \(\kerbas = \expand{\kerbas}\expandMat\), from a
basis of \(\modKer{\expandMat\sys}\); the main advantage is that the latter
basis is computed with the balanced shift \(-\expand{\shifts}\). Note
that a basis of \(\modKer{\expandMat\sys}\) does not straightforwardly yield
\(\expand{\kerbas}\), at least because this kernel also contains
\(\modKer{\expandMat}\). In \cref{lem:parlin_kerexpmat} we exhibit a basis
\(\kerExpMat\) for \(\modKer{\expandMat}\), and then in \cref{lem:to_parlin} we
show that \(\modKer{\expandMat\sys}\) is generated by \(\expand{\kerbas}\) and
\(\kerExpMat\). We also give properties which allow us, from a basis of
\(\modKer{\expandMat\sys}\), to easily recover a basis \(\kerbas\) of
\(\modKer{\sys}\) which has the sought form (see \cref{lem:from_parlin}).

\begin{lemma}
  \label{lem:parlin_kerexpmat}
  The matrix \(\kerExpMat = \diag{\kerExpMat_1,\ldots,\kerExpMat_\rdim} \in
  \pmatRing[(\expand{\rdim}-\rdim)][\expand{\rdim}]\), where
  \[
    \kerExpMat_j =
    \begin{bmatrix}
      \var^\degExp & -1  \\
                   & \ddots & \ddots \\
                   &        & \var^\degExp & -1
    \end{bmatrix}
    \in \pmatRing[(\quoExp_j-1)][\quoExp_j]
  \]
  for \(1\le j\le \rdim\), is the \(\shiftz\)-Popov basis of the kernel
  \(\modKer{\expandMat}\). Furthermore, \(\kerExpMat\) is also in
  \(-\expand{\shifts}\)-Popov form, it has \(-\expand{\shifts}\)-row degree
  \(\shiftz\), and its \(-\expand{\shifts}\)-pivot profile is
  \((\quoExp_1+\cdots+\quoExp_{j-1}+i,\degExp)_{1\le i < \quoExp_j, 1\le
  j\le\rdim}\).
\end{lemma}
\begin{proof}
  By construction, \(\kerExpMat\expandMat = \matz\) and \(\kerExpMat\) is in
  \(\shiftz\)-Popov form. Besides, \(\kerExpMat\) has rank
  \(\expand{\rdim}-\rdim\), which is the rank of \(\modKer{\expandMat}\) since
  \(\expandMat\) has rank \(\rdim\). Now, observe that there is no nonzero
  vector of degree less than \(\degExp\) in the left kernel of the vector
  \(\trsp{[1 \;\; x^\degExp \;\; \cdots \;\; x^{(\quoExp_j-1)\degExp}]}\),
  and thus there is no nonzero vector of degree less than \(\degExp\) in
  \(\modKer{\expandMat}\). It follows that \(\kerExpMat\) is a basis of
  \(\modKer{\expandMat}\). Indeed, if \(\kerbas \in
  \pmatRing[(\expand{\rdim}-\rdim)][\expand{\rdim}] \) is a basis of
  \(\modKer{\expandMat}\) in \(\shiftz\)-reduced form, then \(\rdeg{\kerbas}
  \ge (\degExp,\ldots,\degExp)\). Since \(\kerExpMat\expandMat = \matz\), we
  have \(\kerExpMat = \mat{U} \kerbas\) for some nonsingular \(\mat{U} \in
  \pmatRing[(\expand{\rdim}-\rdim)]\). By the predictable degree property,
  \[
    (\degExp,\ldots,\degExp) = \rdeg{\kerExpMat} = \rdeg{\mat{U}\kerbas} =
    \rdeg[\rdeg{\kerbas}]{\mat{U}} \ge \rdeg{\mat{U}} +
    (\degExp,\ldots,\degExp),
  \]
  hence \(\mat{U}\) is constant. Thus \(\mat{U}\) is unimodular, and
  \(\kerExpMat\) is a basis of \(\modKer{\expandMat}\).

  Now consider \(j\) such that \(\quoExp_j>1\), and write \(\shiftt =
  (-\degExp,\ldots,-\degExp,-\remExp_j) \in \ZZ^{\quoExp_j}\). Then, the
  definition of \(\quoExp_j\) and \(\remExp_j\), notably the fact that
  \(-\remExp_j < 0\), ensures that \(\kerExpMat_j\) is in \(\shiftt\)-Popov form
  with \(\shiftt\)-pivot degree \((\degExp,\ldots,\degExp)\) and
  \(\shiftt\)-pivot index \((1,\ldots,\quoExp_j-1)\). The conclusion about
  \(\kerExpMat\) follows.
\end{proof}

\begin{lemma}
  \label{lem:to_parlin}
  Let \(\sys \in \pmatRing[\rdim][\cdim]\). The following properties hold.
  \begin{enumerate}[(i)]
    \item \label{parlin:it:cdeg} Column degrees are preserved:
      \(\cdeg[\expand{\shifts}]{\expandMat\sys} = \cdeg[\shifts]{\sys}\).
    \item \label{parlin:it:ker} The kernels of \(\sys\) and \(\expandMat\sys\)
      are related by \(\modKer{\expandMat\sys} \expandMat = \modKer{\sys}\).
    \item \label{parlin:it:basis} Let \(\kerbas \in \pmatRing[\cork][\rdim]\) be a
      basis of \(\modKer{\sys}\) (hence \(\cork = \rdim - \rank{\sys}\)), and
      let \(\expand{\kerbas} \in \pmatRing[\cork][\expand{\rdim}]\) be any matrix
      such that \(\kerbas = \expand{\kerbas}\expandMat\). Then,
      \begin{equation}
        \label{eqn:basis_parlin}
        \mat{B} =
        \begin{bmatrix}
          \;\expand{\kerbas}\; \\
          \kerExpMat
        \end{bmatrix} \in
        \pmatRing[(\cork+\expand{\rdim}-\rdim)][\expand{\rdim}]
      \end{equation}
      is a basis of \(\modKer{\expandMat\sys}\), where \(\kerExpMat\) is the
      matrix defined in \cref{lem:parlin_kerexpmat}.
    \item \label{parlin:it:forms} Assume that \(\rdeg[-\shifts]{\kerbas} \ge
      \shiftz\) and that \(\expand{\kerbas}\) is the column partial
      linearization of \(\kerbas\). Then, \(\rdeg[-\expand{\shifts}]{\mat{B}} =
      (\rdeg[-\shifts]{\kerbas},\shiftz)\), and if \(\kerbas\) is in
      \(-\shifts\)-weak Popov form, then \(\mat{B}\) is in
      \(-\expand{\shifts}\)-unordered weak Popov form.
  \end{enumerate}
\end{lemma}
\begin{proof}
  \ref{parlin:it:cdeg} By definition, \(\cdeg[\expand{\shifts}]{\expandMat\sys} =
  \cdeg{\xDiag{\expand{\shifts}}\expandMat\sys}\) and \(\cdeg[\shifts]{\sys} =
  \cdeg{\xDiag{\shifts}\sys}\). Since the matrix
  \[
    \xDiag{\expand{\shifts}} \expandMat = 
    \left[\begin{smallmatrix}
        X^\degExp\\ \vdots \\ X^{(\quoExp_1-1)\degExp} \\ X^{\shift_1} \\ & \ddots \\
                                                                     && X^\degExp \\
                                                                     && \vdots\\
                                                                     && X^{(\quoExp_\rdim-1)\degExp}\\
                                                                     && X^{\shift_\rdim}
    \end{smallmatrix}\right]
  \]
  has column degree \(\shifts\) and contains \(\xDiag{\shifts}\) as a subset of
  its rows, we get \(\cdeg{\xDiag{\expand{\shifts}}\expandMat\sys}  =
  \cdeg{\xDiag{\shifts} \sys}\).

  \ref{parlin:it:ker} For any vector \(\row{p} \in \modKer{\expandMat\sys}\),
  we have \(\row{p} \expandMat \sys = \matz\), which means \(\row{p}\expandMat
  \in \modKer{\sys}\). Conversely, from any \(\row{p} \in \modKer{\sys}\) we
  can construct \(\row{q}\in\pmatRing[1][\expand{\rdim}]\) such that
  \(\row{q}\expandMat = \row{p}\) since \(\expandMat\) has the identity as a
  subset of its rows; then, \(\row{q} \expandMat\sys = \row{p} \sys = \matz\),
  which means \(\row{q} \in \modKer{\expandMat\sys}\).

  \ref{parlin:it:basis} We recall that, by \cref{lem:parlin_kerexpmat},
  \(\kerExpMat\) is a basis of \(\modKer{\expandMat}\). The rows of \(\mat{B}\)
  are in \(\modKer{\expandMat\sys}\), since \(\kerExpMat \expandMat\sys =
  \matz\) and \(\expand{\kerbas} \expandMat \sys = \kerbas \sys = \matz\).
  Now, we want to prove that any \(\row{u} \in \modKer{\expandMat\sys}\) is a
  \(\polRing\)-linear combination of the rows of \(\mat{B}\). By
  \cref{parlin:it:ker}, \(\row{u}\expandMat \in \modKer{\sys}\); thus
  \(\row{u}\expandMat = \row{v} \kerbas = \row{v} \expand{\kerbas} \expandMat\)
  for some \(\row{v} \in \pmatRing[1][\cork]\).  Therefore \(\row{u} - \row{v}
  \expand{\kerbas} \in \modKer{\expandMat}\), and it follows that \(\row{u} -
  \row{v}\expand{\kerbas} = \row{w} \kerExpMat\) for some \(\row{w} \in
  \pmatRing[1][\expand{\rdim}]\). This yields \(\row{u} = [\row{v} \;\;
  \row{w}] \mat{B}\).

  \ref{parlin:it:forms} Since \(\rdeg[-\shifts]{\kerbas} \ge \shiftz\), we can
  apply \cref{lem:parlin_lmat}, which yields \(\rdeg[-\shifts]{\kerbas} =
  \rdeg[-\expand{\shifts}]{\expand{\kerbas}}\), hence
  \(\rdeg[-\expand{\shifts}]{\mat{B}} = (\rdeg[-\shifts]{\kerbas},\shiftz)\).
  Now, if \(\kerbas\) is in \(-\shifts\)-weak Popov form, then
  \cref{lem:parlin_lmat} shows that \(\expand{\kerbas}\) is in
  \(-\expand{\shifts}\)-weak Popov form and that all entries of its
  \(-\expand{\shifts}\)-pivot index are in \(\{\quoExp_1+\cdots+\quoExp_i, 1\le
  i\le\rdim\}\). Besides, by \cref{lem:parlin_kerexpmat}, \(\kerExpMat\) is in
  \(-\expand{\shifts}\)-Popov form with a \(-\expand{\shifts}\)-pivot index
  which is disjoint from \(\{\quoExp_1+\cdots+\quoExp_i, 1\le
  i\le\rdim\}\).  Hence \(\mat{B}\) is in
  \(-\expand{\shifts}\)-unordered weak Popov form.
\end{proof}

\begin{remark}
  \label{rmk:parlin_forms}
  Similarly to \cref{parlin:it:forms}, one may observe that if \(\kerbas\) is
  in \(-\shifts\)-reduced form, then \(\mat{B}\) is in
  \(-\expand{\shifts}\)-reduced form; and that if \(\kerbas\) is in
  \(-\shifts\)-Popov form, then \(\mat{B}\) is in \(-\expand{\shifts}\)-Popov
  form up to row permutation. These points will not be used here; for reference
  a proof is given in \ref{app:parlin_forms}.
\end{remark}

Finally, we combine the above results to show that one can compute a basis of
\(\modKer{\sys}\) by computing a \(-\expand{\shifts}\)-weak Popov basis
of \(\modKer{\expandMat\sys}\) and taking a submatrix of it.

\begin{lemma}
  \label{lem:from_parlin}
  Let \(\sys \in \pmatRing[\rdim][\cdim]\) and let \(\rk =
  \expand{\rdim}-\rank{\sys}\) be the rank of \(\modKer{\expandMat\sys}\).
  Assume that \(-\shifts\)-reduced bases of \(\modKer{\sys}\) have nonnegative
  \(-\shifts\)-row degree. Let \(\mat{Q} \in \pmatRing[\rk][\expand{\rdim}]\) be
  a \(-\expand{\shifts}\)-weak Popov basis of
  \(\modKer{\expandMat\sys}\).  Let \(\mat{P} \in
  \pmatRing[k][\expand{\rdim}]\) be the submatrix of the rows of \(\mat{Q}\)
  whose \(-\expand{\shifts}\)-pivot index is in \(\{\quoExp_1+\cdots+\quoExp_j,
  1\le j\le \rdim\}\). Then, \(\mat{P}\expandMat\) is a
  \(-\shifts\)-weak Popov basis of \(\modKer{\sys}\).
\end{lemma}
\begin{proof}
  We first prove that the number of rows of \(\mat{P}\expandMat\) is the rank
  of \(\modKer{\sys}\), that is, \(k=\rdim-\rank{\sys}\). Indeed, by
  \cref{parlin:it:forms} of \cref{lem:to_parlin}, the
  \(-\expand{\shifts}\)-pivot index of the \(-\expand{\shifts}\)-Popov basis of
  \(\modKer{\expandMat\sys}\) contains the \(-\expand{\shifts}\)-pivot index of
  \(\kerExpMat\). By \cref{lem:parlin_kerexpmat}, the latter is the tuple
  formed by the integers in the set \(\{1,\ldots,\expand{\rdim}\} \setminus
  \{\quoExp_1+\cdots+\quoExp_j, 1\le j\le \rdim\}\) sorted in increasing order.
  Since \(\mat{P}\) is the submatrix of the rows of \(\mat{Q}\) whose
  \(-\expand{\shifts}\)-pivot index is not in this set, \(\mat{P}\) has \(k=\rk
  - (\expand{\rdim} - \rdim) = \rdim - \rank{\sys}\) rows.

  Now, by construction, \(\mat{P}\) is in \(-\expand{\shifts}\)-weak
  Popov form and its \(-\expand{\shifts}\)-pivot index has entries in
  \(\{\quoExp_1+\cdots+\quoExp_i, 1\le i\le\rdim\}\). Thus we can apply
  \cref{lem:from_parlin_owP}; it ensures that \(\mat{P}\expandMat\) is in
  \(-\shifts\)-weak Popov form and that
  \(\rdeg[-\shifts]{\mat{P}\expandMat} = \rdeg[-\expand{\shifts}]{\mat{P}}\).

  It remains to prove that \(\mat{P}\expandMat\) is a basis of
  \(\modKer{\sys}\). Let \(\kerbas \in \matRing[\cork][\rdim]\) be the
  \(-\shifts\)-Popov basis of \(\modKer{\sys}\), and let \(\expand{\kerbas} \in
  \matRing[\cork][\expand{\rdim}]\) be its column partial linearization. Let
  \(\tuple{d} = \rdeg[-\shifts]{\kerbas}\), which has nonnegative entries by
  assumption. Then, according to \cref{lem:parlin_lmat}, \(\expand{\kerbas}\)
  is in \(-\expand{\shifts}\)-weak Popov form, and
  \(\rdeg[-\expand{\shifts}]{\expand{\kerbas}} = \tuple{d}\). Then, we define
  \(\mat{B} \in \matRing[(\cork+\expand{\rdim}-\rdim)][\expand{\rdim}]\) as in
  \cref{eqn:basis_parlin}; by \cref{parlin:it:forms} of \cref{lem:to_parlin},
  the matrix \(\mat{B}\) is a \(-\expand{\shifts}\)-unordered weak Popov basis of
  \(\modKer{\expandMat\sys}\).

  Then, \cref{lem:popov_pivot} shows that \(\mat{Q}\) has the same
  \(-\expand{\shifts}\)-pivot profile as the row permutation of \(\mat{B}\)
  which is in \(-\expand{\shifts}\)-weak Popov form. Since \(\mat{P}\)
  (resp.~\(\expand{\kerbas}\)) is the submatrix of the rows of \(\mat{Q}\)
  (resp.~\(\mat{B}\)) whose \(-\expand{\shifts}\)-pivot index is in
  \(\{\quoExp_1+\cdots+\quoExp_j, 1\le j\le \rdim\}\), and since both
  \(\mat{P}\) and \(\expand{\kerbas}\) are in \(-\expand{\shifts}\)-weak Popov
  form, it follows that \(\mat{P}\) and \(\expand{\kerbas}\) have
  the same \(-\expand{\shifts}\)-pivot profile. In particular, we have
  \(\rdeg[-\expand{\shifts}]{\mat{P}} =
  \rdeg[-\expand{\shifts}]{\expand{\kerbas}} = \tuple{d}\), from which we get
  \(\rdeg[-\shifts]{\mat{P}\expandMat} = \tuple{d}\). According to
  \cref{lem:minimality_implies_basis}, since the rows of \(\mat{P} \expandMat\)
  are in \(\modKer{\sys}\), this implies that \(\mat{P}\expandMat\) is a basis
  of \(\modKer{\sys}\).
\end{proof}

\subsection{Reducing to the case of a balanced order}
\label{sec:weak_popov_to_popov:ovlp_parlin}

Now we apply the overlapping partial linearization of
\cite[Sec.\,2]{Storjohann06}, more precisely the version in
\cite[Sec.\,5.2]{JeaNeiVil20} which supports arbitrary \(\orders\) as showed in
the definition below that we recall for completeness. In the next lemma, we
will show that the sought kernel basis can be retrieved as a submatrix of an
approximant basis for the transformed problem.

\begin{definition}[\cite{Storjohann06,JeaNeiVil20}]
  \label{dfn:ovlplin}
  Let $\orders=(\order_1,\ldots,\order_\cdim) \in \ZZp^\cdim$, let $\sys \in
  \pmatRing[\rdim][\cdim]$ with $\cdeg{\sys} < \orders$, and let
  $\mu\in\ZZp$. Then, for $1\le i\le \cdim$, let $\order_i = \quoExp_i
  \mu + \remExp_i$ with $\quoExp_i = \left\lceil \frac{\order_i}{\mu} -
  1 \right\rceil$ and $1\le \remExp_i\le \mu$.
  Let also $\expand{\cdim} = \max(\quoExp_1-1,0) + \cdots +
  \max(\quoExp_\cdim-1,0)$, and define
  \[
    \mathcal{L}_\mu(\orders) =
    (\expand{\order}_1,\ldots,\expand{\order}_\cdim) \in \ZZp^{\cdim+\expand{\cdim}},
  \]
  where $\expand{\order}_i = (2\mu,\ldots,2\mu,\mu+\remExp_i) \in
  \ZZp^{\quoExp_i}$ if $\quoExp_i > 1$ and $\expand{\order}_i = \order_i$
  otherwise. Considering the $i$th column of $\sys$, we write its
  $\var^\mu$-adic representation as
  \begin{align*}
    \matcol{\sys}{i} & = \matcol{\sys}{i}^{(0)} + \matcol{\sys}{i}^{(1)}
    \var^{\mu} + \cdots + \matcol{\sys}{i}^{(\quoExp_i)} \var^{\quoExp_i\mu}
     \\
    & \quad\text{where }\, \cdeg{[\matcol{\sys}{i}^{(0)} \;\; \matcol{\sys}{i}^{(1)} \;\; \cdots \;\; \matcol{\sys}{i}^{(\quoExp_i)}]} < (\mu,\ldots,\mu,\remExp_i) .
  \end{align*}
  If $\quoExp_i> 1$, we define
  \[
    \matcol{\expand{\sys}}{i} = 
    \begin{bmatrix}
      \matcol{\sys}{i}^{(0)} + \matcol{\sys}{i}^{(1)}\var^{\mu} \,\;&\,\; \matcol{\sys}{i}^{(1)} + \matcol{\sys}{i}^{(2)}\var^{\mu} \,\;&\,\; \cdots \,\;&\,\; \matcol{\sys}{i}^{(\quoExp_i-1)} + \matcol{\sys}{i}^{(\quoExp_i)}\var^{\mu} \\
    \end{bmatrix}
    \in \pmatRing[\rdim][\quoExp_i]
  \]
  and $\mat{J}_i = [\matz \;\; \idMat{\quoExp_i-1}] \in
  \pmatRing[(\quoExp_i-1)][\quoExp_i]$, and otherwise we let
  $\matcol{\expand{\sys}}{i} = \matcol{\sys}{i}$ and $\mat{J}_i \in
  \pmatRing[0][1]$. Then,
  \[
    \mathcal{L}_{\orders,\mu}(\sys) =
    \begin{bmatrix}
      \matcol{\expand{\sys}}{1} & \matcol{\expand{\sys}}{2} & \cdots & \matcol{\expand{\sys}}{\cdim} \\
      \mat{J}_1 &  & &  \\
       & \mat{J}_2 & &  \\
       &  & \ddots &  \\
       &  &  & \mat{J}_\cdim
    \end{bmatrix}
    \in \pmatRing[(\rdim+\expand{\cdim})][(\cdim+\expand{\cdim})]
  \]
  is called the \emph{overlapping linearization} of $\sys$ with respect to
  $\orders$ and $\mu$.
\end{definition}

\begin{lemma}
  \label{lem:ovlp_lin}
  Let \(\sys \in \pmatRing[\rdim][\cdim]\), and let \(\shifts \in \NN^\rdim\)
  be such that there exists a basis of \(\modKer{\sys}\) with nonpositive
  \(-\shifts\)-row degree. Let \(\mu\in\ZZp\) with \(\mu > \max(\shifts)\),
  \(\shiftt = (\shifts,\mu-1,\ldots,\mu-1) \in \NN^{\rdim+\expand{\cdim}}\),
  and \(\orders \in \ZZp^\cdim\) with \(\orders \ge \cdeg[\shifts]{\sys}+1\).
  Let \(\appbas \in \pmatRing[(\rdim+\expand{\cdim})]\) be a
  \(-\shiftt\)-reduced basis of
  \(\modApp{\mathcal{L}_\mu(\orders)}{\mathcal{L}_{\orders,\mu}(\sys)}\).
  Then, exactly \(k = \rdim-\rank{\sys}\) rows of \(\appbas\) have nonpositive
  \(-\shiftt\)-degree, and the first \(\rdim\) columns of these rows form a
  matrix \(\kerbas \in \pmatRing[k][\rdim]\) which is a \(-\shifts\)-reduced
  basis of \(\modKer{\sys}\). Besides, if \(\appbas\) is in \(-\shiftt\)-weak
  Popov form, then \(\kerbas\) is in \(-\shifts\)-weak Popov form.
\end{lemma}
\begin{proof}
  Let \([\kerbas \;\; \mat{Q}]\) be the submatrix of \(\appbas\) formed by its
  rows of nonpositive \(-\shiftt\)-degree, where \(\kerbas \in
  \pmatRing[k][\rdim]\) and \(\mat{Q} \in \pmatRing[k][\expand{\cdim}]\); we
  have \(0 \le k \le \rdim+\expand{\cdim}\).  By choice of \(\shiftt\), from
  \(\rdeg[-\shiftt]{[\kerbas \;\; \mat{Q}]} \le \shiftz\) we get
  \(\deg(\mat{Q}) < \mu\) and \(\rdeg[-\shifts]{\kerbas} \le \shiftz\).

  In particular, \(\deg(\kerbas) \le \max(\shifts) < \mu\): applying the second
  item in \cite[Lem.\,5.6]{JeaNeiVil20} to each row of \([\kerbas \;\;
  \mat{Q}]\) shows that \(\rdeg{\mat{Q}} < \rdeg{\kerbas}\) and that the rows
  of \(\kerbas\) are in \(\modApp{\orders}{\sys}\), that is, \(\kerbas \sys =
  \matz \bmod \xDiag{\orders}\). On the other hand, \(\cdeg{\kerbas} \le
  \shifts\) implies that \(\cdeg{\kerbas \sys} \le \cdeg[\shifts]{\sys} <
  \orders\), hence \(\kerbas\sys=\matz\), i.e.~the rows of \(\kerbas\) are in
  \(\modKer{\sys}\). This implies that the rank of \(\kerbas\) is at most the
  rank of the module \(\modKer{\sys}\), i.e.~\(\rank{\kerbas} \le \rdim-r\)
  where \(r=\rank{\sys}\).

  Furthermore, from \(\rdeg{\mat{Q}} < \rdeg{\kerbas}\) and \(\max(\shifts) <
  \mu\) we obtain
  \[
    \rdeg[-\shifts]{\kerbas} \ge \rdeg{\kerbas} - \max(\shifts)
    > \rdeg{\mat{Q}} - \mu + 1 = \rdeg[(-\mu+1,\ldots,-\mu+1)]{\mat{Q}},
  \]
  hence by choice of \(\shiftt\) we have \(\lmat[-\shiftt]{[\kerbas \;\;
  \mat{Q}]} = [\lmat[-\shifts]{\kerbas} \;\; \shiftz]\). Since this matrix is a
  subset of the rows of the nonsingular matrix \(\lmat[-\shiftt]{\appbas}\), it
  has full row rank, and thus \(\lmat[-\shifts]{\kerbas}\) has full row rank.
  This shows that \(\kerbas\) is in \(-\shifts\)-reduced form, and that \(k =
  \rank{\kerbas}\). If \(\appbas\) is furthermore in \(-\shiftt\)-weak
  Popov form, then \([\kerbas \;\; \mat{Q}]\) is in \(-\shiftt\)-weak
  Popov form as well; then, the identity \(\lmat[-\shiftt]{[\kerbas \;\;
  \mat{Q}]} = [\lmat[-\shifts]{\kerbas} \;\; \shiftz]\) shows that the
  \(-\shiftt\)-pivot entries of \([\kerbas \;\; \mat{Q}]\) are all located in
  \(\kerbas\), hence \(\kerbas\) is in \(-\shifts\)-weak Popov form.

  It remains to prove that \(k=\rdim-r\) and that the rows of \(\kerbas\)
  generate \(\modKer{\sys}\).

  By assumption, there exists a basis \(\mat{P} \in
  \pmatRing[(\rdim-r)][\rdim]\) of \(\modKer{\sys}\) such that \(\cdeg{\mat{P}}
  \le \shifts\). In particular, the rows of \(\mat{P}\) are in
  \(\modApp{\orders}{\sys}\), and applying the first item of
  \cite[Lem.\,5.6]{JeaNeiVil20} to each of these rows shows that there
  exists a matrix \(\mat{R} \in \pmatRing[(\rdim-r)][\expand{\cdim}]\) with
  \(\rdeg{\mat{R}} < \rdeg{\mat{P}}\) and such that the rows of \([\mat{P} \;\;
  \mat{R}]\) are in
  \(\modApp{\mathcal{L}_\mu(\orders)}{\mathcal{L}_{\orders,\mu}(\sys)}\).
  Thus, \([\mat{P} \;\; \mat{R}]\) is a left multiple of \(\appbas\).

  A key remark now is that \(\rdeg[-\shiftt]{[\mat{P} \;\; \mat{R}]} \le
  \shiftz\). Indeed, by \cref{lem:rdegcdeg} \(\rdeg[-\shifts]{\mat{P}} \le
  \shiftz\) follows from \(\cdeg{\mat{P}} \le \shifts\), and we have
  \[
    \deg(\mat{R}) < \deg(\mat{P}) = \max(\cdeg{\mat{P}}) \le \max(\shifts) < \mu.
  \] 
  Thus, since \(\appbas\) is \(-\shiftt\)-reduced, the predictable degree
  property ensures that \([\mat{P} \;\; \mat{R}]\) is a left multiple of
  \(\appbas\) which does not involve the rows of \(\appbas\) of positive
  \(-\shiftt\)-degree, i.e.~a left multiple of \([\kerbas \;\; \mat{Q}]\). In
  particular, \(\mat{P}\) is a left multiple of \(\kerbas\): there exists a
  matrix \(\mat{V} \in \pmatRing[(\rdim-r)][k]\) such that \(\mat{V} \kerbas =
  \mat{P}\). Since \(\mat{P}\) has rank \(\rdim-r\), we obtain \(k \ge
  \rdim-r\), hence \(k=\rdim-r\).  On the other hand, since the rows of
  \(\kerbas\) are in \(\modKer{\sys}\), there exists a matrix \(\mat{W} \in
  \pmatRing[k][k]\) such that \(\mat{W}\mat{P} = \kerbas\). It follows that
  \(\mat{P} = \mat{V}\mat{W} \mat{P}\), and since \(\mat{P}\) has full row rank
  this implies \(\mat{V}\mat{W} = \idMat{k}\). This means that \(\mat{P}\) and
  \(\kerbas\) are left unimodularly equivalent, hence \(\kerbas\) is a basis of
  \(\modKer{\sys}\), which concludes the proof.
\end{proof}

\subsection{Computing kernel bases with known pivot degree}
\label{sec:weak_popov_to_popov:algo}

After applying the transformations presented in
\cref{sec:weak_popov_to_popov:output_parlin,sec:weak_popov_to_popov:ovlp_parlin},
we are left with the computation of an approximant basis for a balanced order
\(\mathcal{L}_\mu(\orders)\) and a balanced shift \(-\expand{\shifts}\): this
is done efficiently by PM-Basis, designed in \cite{GiJeVi03} as an improvement of
\cite[Algo.~SPHPS]{BecLab94}. Here, we use the version in
\cite[Algo.\,2]{JeaNeiVil20} which ensures that the output basis is in
\(-\expand{\shifts}\)-weak Popov form.

\begin{algorithm}
  \caption{\algoname{KnownDegreeKernelBasis}$(\sys,\shifts)$}
  \label{algo:knowndegker}

  \begin{algorithmic}[1]
  \Require{a matrix \(\sys \in \pmatRing[\rdim][\cdim]\), and a nonnegative shift \(\shifts \in \NN^\rdim\).}
  \Assume{\(-\shifts\)-reduced bases of \(\modKer{\sys}\) have \(-\shifts\)-row degree \(\shiftz\).}
  \Ensure{a \(-\shifts\)-weak Popov basis of \(\modKer{\sys}\).}

  \State \CommentLine{Step 1: Output column partial linearization}
      \label{bounddegker:pl} 

  \State \(\degExp \assign \lceil \degdet / \rdim \rceil \in \ZZp\), where \(\degdet = \max(\sumTuple{\shifts},\sumTuple{\cdeg[\shifts]{\sys}},1)\)

  \State Apply \cref{dfn:output_parlin} to \((\shifts,\degExp)\) to obtain the parameters and expansion-compression matrix:
    \((\quoExp_1,\ldots,\quoExp_\rdim) \in \ZZp^\rdim\), \(\expand{\rdim} \in \ZZp\),
    \(\expand{\shifts}\in\NN^{\expand{\rdim}}\),
    \(\expandMat \in \pmatRing[\expand{\rdim}][\rdim]\)

  \State \CommentLine{Step 2: Overlapping partial linearization}
      \label{bounddegker:ovlp}

  \State \(\orders \assign \cdeg[\shifts]{\sys}+1 \in \ZZp^\cdim\) \Comment{order for approximation, equal to \(\cdeg[\expand{\shifts}]{\expandMat\sys}+1\)}

  \State Apply \cref{dfn:ovlplin} to \((\orders, \expandMat\sys, \degExp+1)\)
  to obtain the order \(\mathcal{L}_{\degExp+1}(\orders) \in
  \ZZp^{\cdim+\expand{\cdim}}\) and the matrix
  \(\mathcal{L}_{\orders,\degExp+1}(\expandMat\sys) \in
  \pmatRing[(\expand{\rdim}+\expand{\cdim})][(\cdim+\expand{\cdim})]\)

  \State \CommentLine{Step 3: Compute \(-\shiftt\)-weak Popov basis of \(\modApp{\mathcal{L}_{\degExp+1}(\orders)}{\mathcal{L}_{\orders,\degExp+1}(\expandMat\sys)}\)}

  \State \(\shiftt \assign (\expand{\shifts},\degExp,\ldots,\degExp) \in \NN^{\expand{\rdim}+\expand{\cdim}}\)

  \State \(\Gamma \assign \max(\mathcal{L}_{\degExp+1}(\orders))\);
        \(\mat{G} = \mathcal{L}_{\orders,\degExp+1}(\expandMat\sys) \xDiag{(\Gamma,\ldots,\Gamma) - \mathcal{L}_{\degExp+1}(\orders)}\)
        \Comment{use uniform order \((\Gamma,\ldots,\Gamma)\)}

  \State
    \(\appbas \in \pmatRing[(\expand{\rdim}+\expand{\cdim})] \assign \Call{PM-Basis}{\Gamma,\mat{G},-\shiftt}\)
    \label{bounddegker:pmbasis}

  \State \CommentLine{Step 4: Deduce first \(-\expand{\shifts}\)-weak Popov basis of \(\modKer{\expandMat\sys}\), then \(-\shifts\)-weak Popov basis of \(\modKer{\sys}\)}

  \State \(\mat{Q} \in \pmatRing[\expand{k}][\expand{\rdim}] \assign\) first
    \(\expand{\rdim}\) columns of the rows of \(\appbas\)
    which have nonpositive \(-\shiftt\)-degree 
       \label{bounddegker:extract} 

  \State \(\mat{P} \in \pmatRing[k][\expand{\rdim}] \assign\) the rows of
    \(\mat{Q}\) whose \(-\expand{\shifts}\)-pivot index is in
    \(\{\quoExp_1+\cdots+\quoExp_j, 1\le j\le \rdim\}\)
       \label{bounddegker:reextract} 

  \State \Return \(\mat{P} \expandMat\)
      \label{bounddegker:return}
  \end{algorithmic}
\end{algorithm}

\begin{proposition}
  \label{prop:algo:knowndegker}
  \cref{algo:knowndegker} is correct. Let \(\degdet =
  \max(\sumTuple{\shifts},\sumTuple{\cdeg[\shifts]{\sys}},1)\). Then, assuming
  \(\rdim\ge\cdim\) and using notation from the algorithm, its cost is bounded
  by the sum of:
  \begin{itemize}
    \item the cost of performing \algoname{PM-Basis} at order at most \(2
      \lceil \degdet/\rdim \rceil + 2\) on an input matrix of row dimension
      \(\expand{\rdim}+\expand{\cdim} \le 3\rdim\) and column dimension
      \(\cdim+\expand{\cdim} \le 2\rdim\);
    \item \(\bigO{\rdim^2}\) extra operations in \(\field\).
  \end{itemize}
  Thus, \cref{algo:knowndegker} uses \(\bigO{\rdim^\expmm \timegcd{\degdet/\rdim}}\) operations in \(\field\).
\end{proposition}
\begin{proof}
  Using the assumption that \(-\shifts\)-reduced bases of \(\modKer{\sys}\)
  have \(-\shifts\)-row degree \(\shiftz\) along with \cref{parlin:it:forms} of
  \cref{lem:to_parlin} shows that the \(-\expand{\shifts}\)-reduced bases of
  \(\modKer{\expandMat\sys}\) have \(-\expand{\shifts}\)-row degree
  \(\shiftz\). Thus, we can apply \cref{lem:ovlp_lin} to
  \((\expandMat\sys,\expand{\shifts},\mu,\orders)\) with \(\mu=\degExp+1 >
  \max(\expand{\shifts})\) and \(\orders =
  \cdeg[\expand{\shifts}]{\expandMat\sys}+1\), which is
  \(\orders=\cdeg[\shifts]{\sys}+1\) according to \cref{parlin:it:cdeg} of
  \cref{lem:to_parlin}. Note that \(\appbas\) is a \(-\shiftt\)-weak
  Popov basis of \(\modApp{(\Gamma,\ldots,\Gamma)}{\mat{G}} =
  \modApp{\mathcal{L}_{\degExp+1}(\orders)}{\mathcal{L}_{\orders,\degExp+1}(\expandMat\sys)}\)
  (see e.g.~\cite[Rmk.\,3.3]{JeaNeiVil20} for this approach to make the order
  uniform). Then, \cref{lem:ovlp_lin} states that the matrix \(\mat{Q}\) at
  \cref{bounddegker:extract} has \(\expand{k} = \expand{\rdim} -
  \rank{\expandMat\sys} = \expand{\rdim}-\rank{\sys}\) rows and is a
  \(-\expand{\shifts}\)-weak Popov basis of
  \(\modKer{\expandMat\sys}\).  Then, by \cref{lem:from_parlin},
  \(\mat{P}\expandMat\) is a \(-\shifts\)-weak Popov basis of
  \(\modKer{\sys}\), hence the correctness.

  For the cost analysis, we assume \(\rdim\ge\cdim\), and we start by
  summarizing bounds on the dimensions and degrees at play.
  \Cref{lem:parlin_dims_degs} yields \(\rdim \le \expand{\rdim} \le 2\rdim\),
  while \cref{parlin:it:cdeg} of \cref{lem:to_parlin} ensures
  \(\cdeg[\expand{\shifts}]{\expandMat\sys} = \cdeg[\shifts]{\sys}\). Each entry
  of \(\mathcal{L}_{\degExp+1}(\orders)\) is at most \(2(\degExp+1) = 2\lceil
  \degdet/\rdim\rceil + 2\), by construction. Writing \(\orders =
  (\order_1,\ldots,\order_\cdim)\), by \cref{dfn:ovlplin} we have
  \begin{align*}
    \expand{\cdim} & = \max\left(\left\lceil \frac{\order_1}{\degExp+1} - 1 \right\rceil, 0\right) + \cdots + \max\left(\left\lceil \frac{\order_\cdim}{\degExp+1} - 1 \right\rceil, 0\right)
    \le \frac{\order_1}{\degExp+1} + \cdots + \frac{\order_\cdim}{\degExp+1} = \frac{\sumTuple{\orders}}{\degExp+1}.
  \end{align*}
  Since \(\sumTuple{\orders} = \sumTuple{\cdeg[\shifts]{\sys}} + \cdim \le
  \degdet+\rdim\), and since \(\degExp+1 \ge (\degdet+\rdim)/\rdim\), it follows that
  \(\expand{\cdim} \le \rdim\).
  Besides, \(\rdeg[-\expand{\shifts}]{\mat{P}} \le \shiftz\) by construction,
  so that \(\cdeg{\mat{P}} \le \expand{\shifts}\) by \cref{lem:rdegcdeg}.

  The only steps that involve operations in \(\field\) are the call to
  \algoname{PM-Basis} at \cref{bounddegker:pmbasis} and the multiplication
  \(\mat{P}\expandMat\) at \cref{bounddegker:return}. The construction of
  \(\expandMat\) and the inequality \(\cdeg{\mat{P}} \le \expand{\shifts}\)
  imply that the product \(\mat{P}\expandMat\) mainly involves
  concatenating vectors of coefficients; concerning operations in \(\field\),
  there are \(\expand{\rdim}-\rdim\) columns of \(\mat{P}\) for which
  we may add the constant term of that column to the term of degree \(\degExp\)
  of the previous column. Therefore \cref{bounddegker:return} has cost bound
  \(\bigO{\expand{\rdim} k}\); since \(\expand{\rdim} \le 2\rdim\) and \(k\le
  \rdim\) this is in \(\bigO{\rdim^2}\).
  At \cref{bounddegker:pmbasis}, we call the approximant basis subroutine
  \algoname{PM-Basis} discussed in \cref{sec:prelim:subroutines}; since \(\Gamma
  = \max(\mathcal{L}_{\degExp+1}(\orders))\) is at most \(2 \lceil \degdet/\rdim
  \rceil + 2 \in \bigO{1+\degdet/\rdim}\), \cref{lem:subroutine:approx} states that
  this call uses
  \[
    \bigO{(\expand{\rdim}+\expand{\cdim} + \cdim+\expand{\cdim})
    (\expand{\rdim}+\expand{\cdim})^{\expmm-1} \timegcd{\degdet/\rdim}}
  \]
  operations in \(\field\). Since \(\cdim+\expand{\cdim} \le \expand{\rdim} +
  \expand{\cdim} \le 3\rdim\), this yields the claimed cost bound.
\end{proof}

\subsection{Proof of Theorem~\ref{thm:weak_popov_to_popov}}
\label{sec:weak_popov_to_popov:proof}

\begin{algorithm}
  \caption{\algoname{WeakPopovToPopov}$(\pmat,\shifts)$}
  \label{algo:weak_to_popov}

  \begin{algorithmic}[1]
  \Require{a matrix \(\pmat\in\pmatRing[\rdim][\cdim]\), a shift \(\shifts\in\ZZ^\cdim\) such that \(\pmat\) is in \(-\shifts\)-weak Popov form.}
  \Assume{\(\subTuple{\shifts}{\pivInds} \ge \pivDegs\), where \((\pivInds,\pivDegs)\) is the \(-\shifts\)-pivot profile of \(\pmat\).}
  \Ensure{the \(-\shifts\)-Popov form of \(\pmat\).}

  \State \CommentLine{Step 1: Find unimodular transformation and \(-\pivDegs\)-reduced form of \(\matcols{\pmat}{\pivInds}\)}

  \State \((\pivInds,\pivDegs) \assign\) the \(-\shifts\)-pivot profile of \(\pmat\)

  \State \([\mat{U} \;\; \mat{R}] \in \pmatRing[\rdim][2\rdim] \assign \algoname{KnownDegreeKernelBasis}\left(\begin{bmatrix} \matcols{\pmat}{\pivInds} \\ -\idMat{\rdim} \end{bmatrix},(\subTuple{\shifts}{\pivInds}-\pivDegs,\pivDegs)\right)\)
     \label{step:normalize:kernel}

  \State \CommentLine{Step 2: Deduce \(-\shifts\)-Popov form of \(\pmat\)}

  \State \(\mat{P} \assign\) zero matrix in \(\pmatRing[\rdim][\cdim]\)

  \State \(\matcols{\mat{P}}{\pivInds} \assign \lmat[-\pivDegs]{\mat{R}}^{-1} \mat{R}\)
      \label{step:normalize:popov} 

  \State \(\matcols{\mat{P}}{\{1,\ldots,\cdim\}\setminus\pivInds} \assign
  \lmat[-\pivDegs]{\mat{R}}^{-1}\mat{U} \,
  \matcols{\pmat}{\{1,\ldots,\cdim\}\setminus\pivInds}\) \label{step:nonsquare:update}

  \State \Return \(\mat{P}\)
\end{algorithmic}
\end{algorithm}

For proving \cref{thm:weak_popov_to_popov}, we describe
Algorithm~\algoname{WeakPopovToPopov} (\cref{algo:weak_to_popov}) and we focus
on the square case, \(\rdim=\cdim\). Then, by definition of the
\(-\shifts\)-weak Popov form, \(\pivDegs\) is the tuple of degrees of the
diagonal entries of \(\pmat\), and \(\pivInds = (1,\ldots,\rdim)\).
Furthermore, in this case, \(\subTuple{\shifts}{\pivInds} = \shifts\),
\(\matcols{\pmat}{\pivInds} = \pmat\), and \(\matcols{\mat{P}}{\pivInds} =
\mat{P}\); in particular, we can discard the step
at~\cref{step:nonsquare:update} since the submatrices it involves are empty.

First note that the shift \(\tuple{d} =
(\subTuple{\shifts}{\pivInds}-\pivDegs,\pivDegs) =
(\shifts-\pivDegs,\pivDegs)\) used at \cref{step:normalize:kernel} is
nonnegative. Besides, \cref{lem:normalize} shows that \(-\cdegs\)-reduced bases
of \(\modKer{\sys}\) have \(-\cdegs\)-row degree \(\shiftz\). Thus the
requirements of Algorithm~\algoname{KnownDegreeKernelBasis} are met, and
\cref{prop:algo:knowndegker} shows that the matrix \([\mat{U} \;\; \mat{R}]\)
computed at \cref{step:normalize:kernel} is a \(-\cdegs\)-weak Popov
basis of \(\modKer{\sys}\). Then, \cref{cor:normalize} shows that the matrix
\(\mat{P} = \matcols{\mat{P}}{\pivInds} = \lmat[-\pivDegs]{\mat{R}}^{-1}
\mat{R}\) computed at \cref{step:normalize:popov} is the \(-\shifts\)-Popov
form of \(\pmat\). This proves that \cref{algo:weak_to_popov} is correct.

Concerning the cost bound, we focus on the case \(\sumTuple{\shifts}>0\).
Indeed, since \(\shifts\ge\pivDegs\ge\shiftz\), if \(\sumTuple{\shifts}=0\),
then \(\shifts=\pivDegs=\shiftz\). In this case, no computation needs to be
done: the \(-\shifts\)-Popov form of \(\pmat\) is \(\idMat{\rdim}\), the unique
matrix in Popov form whose pivot degree is \(\shiftz\).

The cost of \cref{step:normalize:popov} is that of multiplying
\(\lmat[-\pivDegs]{\mat{R}}^{-1} \in \matRing[\rdim]\) by \(\mat{R} \in
\pmatRing[\rdim]\), which has column degree \(\pivDegs\) as explained in
\cref{sec:weak_popov_to_popov:via_kernel}; this computation corresponds to the second
item in \cref{thm:weak_popov_to_popov}. This multiplication can be done by first
performing a column linearization of \(\mat{R}\) into a \(\rdim \times
(\rdim+\sumTuple{\pivDegs})\) matrix \(\expand{\mat{R}}\) over \(\field\),
computing \(\lmat[-\pivDegs]{\mat{R}}^{-1} \expand{\mat{R}}\), and finally
compressing the result back into a polynomial matrix. This uses
\(\bigO{\rdim^\expmm (1 + \sumTuple{\pivDegs} / \rdim})\) operations in
\(\field\).

Concerning \cref{step:normalize:kernel}, we rely on
\cref{prop:algo:knowndegker}. Here, the matrix we give as input to
Algorithm~\algoname{KnownDegreeKernelBasis} has dimensions
\(2\rdim\times\rdim\), hence the dimensions in \cref{prop:algo:knowndegker}
satisfy \(\expand{\rdim} \le 4\rdim\) and \(\expand{\cdim} \le 2\rdim\) (the
latter bound comes from the proof of that proposition). Then,
\cref{prop:algo:knowndegker} states that \cref{step:normalize:kernel} costs:
\begin{itemize}
  \item \(\bigO{\rdim^2}\) operations in \(\field\), which is the third item in
    \cref{thm:weak_popov_to_popov},
  \item one call to \algoname{PM-Basis} on a matrix of row dimension
    \(\expand{\rdim}+\expand{\cdim} \le 6\rdim\), column dimension
    \(\rdim+\expand{\cdim} \le 3\rdim\), and at order at most \(2 \lceil
    \degdet/(2\rdim) \rceil + 2\), where \(\degdet =
    \max(\sumTuple{\cdegs},\sumTuple{\cdeg[\cdegs]{\sys}},1)\) and \(\sys\) is
    the input matrix \(\trsp{[\trsp{\pmat} \;\; -\idMat{\rdim}]}\).
\end{itemize}
Besides, \cref{prop:algo:knowndegker} also implies that
\cref{step:normalize:kernel} uses \(\bigO{\rdim^\expmm
\timegcd{\degdet/\rdim}}\) operations in \(\field\).

We are going to prove that \(\degdet = \sumTuple{\shifts}\), which concludes
the proof. Indeed, the previous paragraph then directly gives the overall cost
bound \(\bigO{\rdim^\expmm \timegcd{\sumTuple{\shifts}/\rdim}}\) in
\cref{thm:weak_popov_to_popov}, and using
\[
  2 \left\lceil \frac{\degdet}{2\rdim} \right\rceil + 2
  < 2 \left(\frac{\degdet}{2\rdim} + 1\right) + 2
  = \sumTuple{\shifts}/\rdim + 4,
\]
the previous paragraph also gives the first item in that theorem.

To observe that \(\degdet = \sumTuple{\shifts}\), we first use the definition
of \(\tuple{d}\) to obtain \(\sumTuple{\tuple{d}} = \sumTuple{\shifts-\pivDegs}
+ \sumTuple{\pivDegs} = \sumTuple{\shifts}\). Since \(\sumTuple{\shifts}\ge
1\), this gives \(\degdet = \max(\sumTuple{\shifts},
\sumTuple{\cdeg[\cdegs]{\sys}})\).  Now, \(\cdeg[\cdegs]{\sys}\) is the
entry-wise maximum of \(\cdeg[\shifts-\pivDegs]{\pmat}\) and
\(\cdeg[\pivDegs]{-\idMat{\rdim}} = \pivDegs\).  Since \(\pmat\) is in
\(-\shifts\)-row reduced form with \(\rdeg[-\shifts]{\pmat} =
-\shifts+\pivDegs\), \cref{lem:row_column_reduced} ensures that
\(\trsp{\pmat}\) is in \(\shifts-\pivDegs\)-reduced form with
\(\cdeg[\shifts-\pivDegs]{\pmat} = \rdeg[\shifts-\pivDegs]{\trsp{\pmat}} =
\shifts\). Then, since \(\shifts\ge\pivDegs\) we obtain \(\cdeg[\cdegs]{\sys} =
\shifts\), and thus \(\degdet = \sumTuple{\shifts}\). This concludes the proof
of \cref{thm:weak_popov_to_popov}.

As for the rectangular case \(\rdim<\cdim\), the correctness of
\cref{algo:weak_to_popov} follows from the above proof in the square case,
which shows that the algorithm correctly computes the
\(-\subTuple{\shifts}{\pivInds}\)-Popov form \(\matcols{\mat{P}}{\pivInds}\) of
\(\matcols{\pmat}{\pivInds}\), and from \cite[Lem.\,5.1]{NeiRosSol18}
concerning the computation of
\(\matcols{\mat{P}}{\{1,\ldots,\cdim\}\setminus\pivInds}\) and the fact that
\(\lmat[-\pivDegs]{\mat{R}}^{-1}\mat{U}\) is the unimodular matrix which
transforms \(\pmat\) into \(\mat{P}\). The cost bound can be derived using the
degree bounds on rectangular shifted Popov forms given in
\cite{BeLaVi06,NeiRosSol18} (we do not detail this here since this would add
technical material unrelated to the main results of this article).

\appendix
\section{}
\renewcommand\thesection{A} % to avoid that the lemmas are called "Lemma Appendix A.1", rather just "Lemma A.1"
\label{app:parlin_forms}

In this appendix, we give proofs for \cref{rmk:parlin_lmat,rmk:parlin_forms}.
We use notation from \cref{sec:weak_popov_to_popov:output_parlin}: a shift
$\shifts \in \NN^\rdim$, an integer $\degExp\in\ZZp$, and a matrix \(\sys \in
\pmatRing[\rdim][\cdim]\) are given; \(\kerbas \in \pmatRing[\cork][\rdim]\) is
a basis of \(\modKer{\sys}\); \((\quoExp_j,\remExp_j)_{1\le j\le\rdim}\) and
\(\expand{\shifts} \in \NN^{\expand{\rdim}}\) and \(\expand{\kerbas} \in
\pmatRing[\cork][\expand{\rdim}]\) are as described in
\cref{dfn:output_parlin}. Following the context of
\cref{rmk:parlin_lmat,rmk:parlin_forms}, we assume \(\rdeg[-\shifts]{\kerbas}
\ge \shiftz\), hence in particular \(\rdeg[-\expand{\shifts}]{\expand{\kerbas}}
= \rdeg[-\shifts]{\kerbas}\) (see \cref{lem:parlin_lmat}). We start with the
claims in \cref{rmk:parlin_lmat}.

\begin{lemma}
  \label{lem:rmk:parlin_lmat}
  Writing
  \(\lmat[-\shifts]{\kerbas} = [\matcol{\mat{L}}{1} \;\; \cdots \;\;
  \matcol{\mat{L}}{\rdim}] \in \matRing[k][\rdim]\), we have
  \[
    \lmat[-\expand{\shifts}]{\expand{\kerbas}} =
    [\,\underbrace{\matz \;\; \cdots \;\; \matz \;\; \matcol{\mat{L}}{1}}_{\quoExp_1}
    \;\; \cdots \;\;
    \underbrace{\matz \;\; \cdots \;\; \matz \;\; \matcol{\mat{L}}{\rdim}}_{\quoExp_\rdim}]
    \in \matRing[k][\expand{\rdim}].
  \]
  If \(\kerbas\) is in \(-\shifts\)-Popov form, then
  \(\expand{\kerbas}\) is in \(-\expand{\shifts}\)-Popov form.
\end{lemma}
\begin{proof}
  For given \(i \in \{1,\ldots,\cork\}\) and \(j \in \{1,\ldots,\rdim\}\), we
  rely on the definition of a leading matrix (see \cref{sec:prelim:reduced}):
  the entry \((i,j)\) of \(\lmat[-\expand{\shifts}]{\expand{\kerbas}}\) is the
  coefficient of degree \(d_i + \expand{\shifts}_j\) of
  \(\expand{\kerbas}_{i,j}\), where $d_i =
  \rdeg[-\expand{\shifts}]{\matrow{\expand{\kerbas}}{i}} =
  \rdeg[-\shifts]{\matrow{\kerbas}{i}} \ge 0$. First consider the case \(j
  \not\in \{\quoExp_1+\cdots+\quoExp_\pivInd, 1\le \pivInd \le\rdim\}\).  Then
  the \(j\)th entry of \(\expand{\shifts}\) is \(\expand{\shifts}_j =
  \degExp\), and by definition of the output column partial linearization
  \(\expand{\kerbas}_{i,j}\) has degree less than \(\degExp\), hence its
  coefficient of degree \(d_i + \degExp \ge \degExp\) must be zero. This proves
  that all columns of \(\lmat[-\expand{\shifts}]{\expand{\kerbas}}\) with index
  not in \(\{\quoExp_1+\cdots+\quoExp_\pivInd, 1\le \pivInd \le\rdim\}\) are
  indeed zero. It remains to prove that in the case \(j = \quoExp_1 + \cdots +
  \quoExp_\pivInd\) for some \(1\le \pivInd \le\rdim\), then the entry
  \((i,j)\) of \(\lmat[-\expand{\shifts}]{\expand{\kerbas}}\) is equal to
  \(\mat{L}_{i,\pivInd}\). This holds, since in this case we have
  \(\expand{\shifts}_j = \remExp_\pivInd\), and by construction of
  \(\expand{\kerbas}\) the coefficient of degree \(d_i + \remExp_\pivInd\) of
  \(\expand{\kerbas}_{i,j}\) is equal to the coefficient of degree \(d_i +
  (\quoExp_\pivInd-1)\degExp + \remExp_\pivInd = d_i + \shift_{\pivInd}\) of
  \(\kerbas_{i,\pivInd}\), which itself is equal to \(\mat{L}_{i,\pivInd}\) by
  definition of a leading matrix.

  Now assume \(\kerbas\) is in \(-\shifts\)-Popov form. We have showed in
  \cref{lem:parlin_lmat} that \(\expand{\kerbas}\) is in
  \(-\expand{\shifts}\)-weak Popov form and that the row
  \(\matrow{\expand{\kerbas}}{i}\) has its \(-\expand{\shifts}\)-pivot at index
  \(\quoExp_1+\cdots+\quoExp_{\pivInd_i}\) where \(\pivInd_i\) is the
  \(-\shifts\)-pivot index of \(\matrow{\kerbas}{i}\). By construction, the
  column \(\matcol{\expand{\kerbas}}{\quoExp_1+\cdots+\quoExp_{\pivInd_i}}\) is
  the part of nonnegative degree of the column
  \(\var^{-(\quoExp_{\pivInd_i}-1)\degExp} \matcol{\kerbas}{\pivInd_i} =
  \var^{-\shift_{\pivInd_i} + \remExp_i} \matcol{\kerbas}{\pivInd_i}\).  It
  follows first that the \(-\expand{\shifts}\)-pivot entry
  \(\expand{\kerbas}_{i,\quoExp_1+\cdots+\quoExp_{\pivInd_i}}\) is monic since
  it is a high degree part of the (monic) \(-\shifts\)-pivot entry
  \(\kerbas_{i,\pivInd_i}\); and second that
  \(\expand{\kerbas}_{i,\quoExp_1+\cdots+\quoExp_{\pivInd_i}}\) has degree
  strictly larger than all other entries in the column
  \(\matcol{\expand{\kerbas}}{\quoExp_1+\cdots+\quoExp_{\pivInd_i}}\) since
  \(\kerbas_{i,\pivInd_i}\) has degree strictly larger than all other entries
  in the column \(\matcol{\kerbas}{\pivInd_i}\).  Hence \(\expand{\kerbas}\) is
  in \(-\expand{\shifts}\)-Popov form.
\end{proof}

Now we prove the claims in \cref{rmk:parlin_forms} concerning variants of
\cref{parlin:it:forms} of \cref{lem:to_parlin}, using further notation from
\cref{sec:weak_popov_to_popov:output_parlin}: \(\kerExpMat\) is the basis of
\(\modKer{\expandMat}\) described in \cref{lem:parlin_kerexpmat} and \(\mat{B}
= [\begin{smallmatrix} \;\expand{\kerbas}\; \\ \kerExpMat \end{smallmatrix}]
\in \pmatRing[(\cork+\expand{\rdim}-\rdim)][\expand{\rdim}]\) is the basis of
\(\modKer{\expandMat\sys}\) given in \cref{eqn:basis_parlin}. As above we
assume \(\rdeg[-\shifts]{\kerbas} \ge \shiftz\), and therefore
\(\rdeg[-\expand{\shifts}]{\mat{B}} = (\rdeg[-\shifts]{\kerbas},\shiftz)\) as
stated in \cref{parlin:it:forms} of \cref{lem:to_parlin}.

\begin{lemma}
  \label{lem:rmk:parlin_forms}
  With the above notation and assumptions,
  \begin{itemize}
    \item If \(\kerbas\) is in \(-\shifts\)-reduced form, then
      \(\mat{B}\) is in \(-\expand{\shifts}\)-reduced form.
    \item If \(\kerbas\) is in \(-\shifts\)-Popov form, then
      \(\mat{B}\) is in \(-\expand{\shifts}\)-Popov form up to row permutation.
  \end{itemize}
\end{lemma}

\begin{proof}
  Suppose first that \(\kerbas\) is in \(-\shifts\)-reduced form. We apply
  \cref{lem:rmk:parlin_lmat}:
  \[
   \lmat[-\expand{\shifts}]{\expand{\kerbas}} =
   [\,\matz \;\; \cdots \;\; \matz \;\; \matcol{\mat{L}}{1}
     \;\; \cdots \;\;
   \matz \;\; \cdots \;\; \matz \;\; \matcol{\mat{L}}{\rdim}]
  \]
  where \(\lmat[-\shifts]{\kerbas} = [\matcol{\mat{L}}{1} \;\; \cdots \;\;
  \matcol{\mat{L}}{\rdim}] \in \matRing[k][\rdim]\). The latter matrix has full
  rank, since \(\kerbas\) is in \(-\shifts\)-reduced form. 
  On the other hand, for matrices such as the diagonal blocks of
  \(\kerExpMat\) we have
  \[
    \lmatBig[(-\degExp,\ldots,-\degExp,-\remExp)]{\begin{bmatrix}
      \var^\degExp & -1  \\
                   & \ddots & \ddots \\
                   &        & \var^\degExp & -1
    \end{bmatrix}}
    =
    \begin{bmatrix}
      1 & 0  \\
                   & \ddots & \ddots \\
                   &        & 1 & 0
    \end{bmatrix}
  \]
  for any integer \(\remExp \ge 1\). As a result,
  \begin{equation*}
   \lmat[-\expand{\shifts}]{\mat{B}} =
   \begin{bmatrix}
     \lmat[-\expand{\shifts}]{\expand{\kerbas}} \\
     \lmat[-\expand{\shifts}]{\kerExpMat}
   \end{bmatrix}
   =
   \left[\begin{array}{cccc;{1pt/3pt}c;{1pt/3pt}cccc}
           & & & \vert               &        & & & & \vert \\
           & & & \matcol{\mat{L}}{1} & \cdots & & & & \matcol{\mat{L}}{\rdim} \\
           & & & \vert               &        & & & & \vert \\ \hdashline[1pt/3pt]
         1 & & & & &     \\
           & \ddots & & & & &  \\
           &        & 1 & & & & &  \\ \hdashline[1pt/3pt]
           &        & & & \ddots \\ \hdashline[1pt/3pt]
           & & & & & 1     \\
           & & & & & & \ddots \\
           & & & & & &        & 1 \\
   \end{array}\right]
  \end{equation*}
  has full rank as well, which means that \(\mat{B}\) is in
  \(-\expand{\shifts}\)-reduced form.

  Now, assume further that \(\kerbas\) is in \(-\shifts\)-Popov form. Then
  \cref{lem:rmk:parlin_lmat} states that \(\expand{\kerbas}\) is in
  \(-\expand{\shifts}\)-Popov form, and the above description of
  \(\lmat[-\expand{\shifts}]{\mat{B}}\) shows that \(\expand{\kerbas}\) and
  \(\kerExpMat\) have disjoint \(-\expand{\shifts}\)-pivot indices, and that
  all their \(-\expand{\shifts}\)-pivots are monic. Hence \(\mat{B}\) is in
  \(-\expand{\shifts}\)-unordered weak Popov form with monic
  \(-\expand{\shifts}\)-pivots. It remains to show that each of these
  \(-\expand{\shifts}\)-pivots has degree strictly larger than the other
  entries in the same column. Each row of the submatrix \(\kerExpMat\) of
  \(\mat{B}\) has \(-\expand{\shifts}\)-pivot entry \(\var^\degExp\) and the
  other entries in the same column of \(\mat{B}\) are either \(-1\), which has
  degree \(0<\degExp\), or are in a column of \(\expand{\kerbas}\) whose index
  is not in \(\{\quoExp_1+\cdots+\quoExp_j, 1\le j\le \rdim\}\), which has
  degree less than \(\degExp\) by definition of the column partial
  linearization. The \(j\)th row of the submatrix \(\expand{\kerbas}\) of
  \(\mat{B}\) has \(-\expand{\shifts}\)-pivot index
  \(\quoExp_1+\cdots+\quoExp_j\) and strictly positive
  \(-\expand{\shifts}\)-pivot degree since
  \(\rdeg[-\expand{\shifts}]{\expand{\kerbas}} \ge \shiftz\) and the entry of
  \(-\expand{\shifts}\) at index \(\quoExp_1+\cdots+\quoExp_j\) is \(-\remExp_j
  < 0\). This concludes the proof since the column
  \(\matcol{\kerExpMat}{\quoExp_1+\cdots+\quoExp_j}\) has degree \(0\), and
  since \(\expand{\kerbas}\) is in \(-\expand{\shifts}\)-Popov form.
\end{proof}

% Fakesection   %% Vincent: used for folding, keep this line thanks
\bibliographystyle{elsarticle-harv}

\begin{thebibliography}{58}
\expandafter\ifx\csname natexlab\endcsname\relax\def\natexlab#1{#1}\fi
\providecommand{\url}[1]{\texttt{#1}}
\providecommand{\href}[2]{#2}
\providecommand{\path}[1]{#1}
\providecommand{\DOIprefix}{doi:}
\providecommand{\ArXivprefix}{arXiv:}
\providecommand{\URLprefix}{URL: }
\providecommand{\Pubmedprefix}{pmid:}
\providecommand{\doi}[1]{\href{http://dx.doi.org/#1}{\path{#1}}}
\providecommand{\Pubmed}[1]{\href{pmid:#1}{\path{#1}}}
\providecommand{\bibinfo}[2]{#2}
\ifx\xfnm\relax \def\xfnm[#1]{\unskip,\space#1}\fi
%Type = Article
\bibitem[{Abdeljaoued and Malaschonok(2001)}]{AbMa01}
\bibinfo{author}{Abdeljaoued, J.}, \bibinfo{author}{Malaschonok, G.I.},
  \bibinfo{year}{2001}.
\newblock \bibinfo{title}{Efficient algorithms for computing the characteristic
  polynomial in a domain}.
\newblock \bibinfo{journal}{Journal of Pure and Applied Algebra}
  \bibinfo{volume}{156}, \bibinfo{pages}{127--145}.
\newblock \DOIprefix\doi{10.1016/S0022-4049(99)00158-9}.
%Type = Article
\bibitem[{Baur and Strassen(1983)}]{BaurStrassen83}
\bibinfo{author}{Baur, W.}, \bibinfo{author}{Strassen, V.},
  \bibinfo{year}{1983}.
\newblock \bibinfo{title}{The complexity of partial derivatives}.
\newblock \bibinfo{journal}{Theoretical computer science} \bibinfo{volume}{22},
  \bibinfo{pages}{317--330}.
\newblock \DOIprefix\doi{10.1016/0304-3975(83)90110-X}.
%Type = Article
\bibitem[{Beckermann and Labahn(1994)}]{BecLab94}
\bibinfo{author}{Beckermann, B.}, \bibinfo{author}{Labahn, G.},
  \bibinfo{year}{1994}.
\newblock \bibinfo{title}{A uniform approach for the fast computation of
  matrix-type {P}ad\'e approximants}.
\newblock \bibinfo{journal}{SIAM J. Matrix Anal. Appl.} \bibinfo{volume}{15},
  \bibinfo{pages}{804--823}.
\newblock \DOIprefix\doi{10.1137/S0895479892230031}.
%Type = Inproceedings
\bibitem[{Beckermann et~al.(1999)Beckermann, Labahn and Villard}]{BeLaVi99}
\bibinfo{author}{Beckermann, B.}, \bibinfo{author}{Labahn, G.},
  \bibinfo{author}{Villard, G.}, \bibinfo{year}{1999}.
\newblock \bibinfo{title}{Shifted normal forms of polynomial matrices}, in:
  \bibinfo{booktitle}{ISSAC'99}, \bibinfo{publisher}{ACM}. pp.
  \bibinfo{pages}{189--196}.
\newblock \DOIprefix\doi{10.1145/309831.309929}.
%Type = Article
\bibitem[{Beckermann et~al.(2006)Beckermann, Labahn and Villard}]{BeLaVi06}
\bibinfo{author}{Beckermann, B.}, \bibinfo{author}{Labahn, G.},
  \bibinfo{author}{Villard, G.}, \bibinfo{year}{2006}.
\newblock \bibinfo{title}{Normal forms for general polynomial matrices}.
\newblock \bibinfo{journal}{J. Symbolic Comput.} \bibinfo{volume}{41},
  \bibinfo{pages}{708--737}.
\newblock \DOIprefix\doi{10.1016/j.jsc.2006.02.001}.
%Type = Article
\bibitem[{Berkowitz(1984)}]{Berkowitz84}
\bibinfo{author}{Berkowitz, S.J.}, \bibinfo{year}{1984}.
\newblock \bibinfo{title}{On computing the determinant in small parallel time
  using a small number of processors}.
\newblock \bibinfo{journal}{Information Processing Letters}
  \bibinfo{volume}{18}, \bibinfo{pages}{147--150}.
\newblock \DOIprefix\doi{10.1016/0020-0190(84)90018-8}.
%Type = Article
\bibitem[{Bunch and Hopcroft(1974)}]{BunHop74}
\bibinfo{author}{Bunch, J.R.}, \bibinfo{author}{Hopcroft, J.E.},
  \bibinfo{year}{1974}.
\newblock \bibinfo{title}{Triangular factorization and inversion by fast matrix
  multiplication}.
\newblock \bibinfo{journal}{Mathematics of Computation} \bibinfo{volume}{28},
  \bibinfo{pages}{231--236}.
\newblock \DOIprefix\doi{10.2307/2005828}.
%Type = Book
\bibitem[{B{\"u}rgisser et~al.(1997)B{\"u}rgisser, Clausen and
  Shokrollahi}]{BurgisserClausenShokrollahi2010}
\bibinfo{author}{B{\"u}rgisser, P.}, \bibinfo{author}{Clausen, M.},
  \bibinfo{author}{Shokrollahi, A.}, \bibinfo{year}{1997}.
\newblock \bibinfo{title}{Algebraic Complexity Theory}.
\newblock \bibinfo{edition}{1st} ed., \bibinfo{publisher}{Springer-Verlag
  Berlin Heidelberg}.
\newblock \DOIprefix\doi{10.1007/978-3-662-03338-8}.
%Type = Article
\bibitem[{Cantor and Kaltofen(1991)}]{CanKal91}
\bibinfo{author}{Cantor, D.G.}, \bibinfo{author}{Kaltofen, E.},
  \bibinfo{year}{1991}.
\newblock \bibinfo{title}{On fast multiplication of polynomials over arbitrary
  algebras}.
\newblock \bibinfo{journal}{Acta Inform.} \bibinfo{volume}{28},
  \bibinfo{pages}{693--701}.
\newblock \DOIprefix\doi{10.1007/BF01178683}.
%Type = Phdthesis
\bibitem[{Cook(1966)}]{Cook1966}
\bibinfo{author}{Cook, S.A.}, \bibinfo{year}{1966}.
\newblock \bibinfo{title}{{On the minimum computation time of functions}}.
\newblock Ph.D. thesis.
%Type = Inproceedings
\bibitem[{{Csanky}(1975)}]{Csanky75}
\bibinfo{author}{{Csanky}, L.}, \bibinfo{year}{1975}.
\newblock \bibinfo{title}{Fast parallel matrix inversion algorithms}, in:
  \bibinfo{booktitle}{16th Annual Symposium on Foundations of Computer Science
  (sfcs 1975)}, pp. \bibinfo{pages}{11--12}.
\newblock \DOIprefix\doi{10.1109/SFCS.1975.14}.
%Type = Article
\bibitem[{Danilevskij(1937)}]{Danilevskij1937}
\bibinfo{author}{Danilevskij, A.M.}, \bibinfo{year}{1937}.
\newblock \bibinfo{title}{The numerical solution of the secular equation}.
\newblock \bibinfo{journal}{Matem. Sbornik} \bibinfo{volume}{44},
  \bibinfo{pages}{169--171}.
\newblock \bibinfo{note}{In Russian}.
%Type = Article
\bibitem[{Dumas et~al.(2017)Dumas, Pernet and Sultan}]{DuPeSu17}
\bibinfo{author}{Dumas, J.G.}, \bibinfo{author}{Pernet, C.},
  \bibinfo{author}{Sultan, Z.}, \bibinfo{year}{2017}.
\newblock \bibinfo{title}{Fast computation of the rank profile matrix and the
  generalized {B}ruhat decomposition}.
\newblock \bibinfo{journal}{J. Symbolic Comput.} \bibinfo{volume}{83},
  \bibinfo{pages}{187--210}.
\newblock \DOIprefix\doi{10.1016/j.jsc.2016.11.011}.
%Type = Inproceedings
\bibitem[{Dumas et~al.(2005)Dumas, Pernet and Wan}]{DPW05}
\bibinfo{author}{Dumas, J.G.}, \bibinfo{author}{Pernet, C.},
  \bibinfo{author}{Wan, Z.}, \bibinfo{year}{2005}.
\newblock \bibinfo{title}{Efficient computation of the characteristic
  polynomial}, in: \bibinfo{booktitle}{ISSAC'05}, \bibinfo{publisher}{ACM}. pp.
  \bibinfo{pages}{140--147}.
\newblock \DOIprefix\doi{10.1145/1073884.1073905}.
%Type = Book
\bibitem[{Dummit and Foote(2004)}]{DumFoo04}
\bibinfo{author}{Dummit, D.S.}, \bibinfo{author}{Foote, R.M.},
  \bibinfo{year}{2004}.
\newblock \bibinfo{title}{Abstract Algebra}.
\newblock \bibinfo{publisher}{John Wiley \& Sons}.
%Type = Book
\bibitem[{Faddeev and Sominskii(1949)}]{FaSo49}
\bibinfo{author}{Faddeev, D.}, \bibinfo{author}{Sominskii, I.},
  \bibinfo{year}{1949}.
\newblock \bibinfo{title}{Collected Problems in Higher Algebra, Problem
  n°979}.
%Type = Article
\bibitem[{{Forney, Jr.}(1975)}]{Forney75}
\bibinfo{author}{{Forney, Jr.}, G.D.}, \bibinfo{year}{1975}.
\newblock \bibinfo{title}{{Minimal Bases of Rational Vector Spaces, with
  Applications to Multivariable Linear Systems}}.
\newblock \bibinfo{journal}{SIAM Journal on Control} \bibinfo{volume}{13},
  \bibinfo{pages}{493--520}.
\newblock \DOIprefix\doi{10.1137/0313029}.
%Type = Article
\bibitem[{Frame(1949)}]{Frame49}
\bibinfo{author}{Frame, J.}, \bibinfo{year}{1949}.
\newblock \bibinfo{title}{A simple recurrent formula for inverting a matrix
  (abstract)}.
\newblock \bibinfo{journal}{Bull. of Amer. Math. Soc.} \bibinfo{volume}{55},
  \bibinfo{pages}{1045}.
%Type = Book
\bibitem[{Gathen and Gerhard(2013)}]{vzGathen13}
\bibinfo{author}{Gathen, J.v.z.}, \bibinfo{author}{Gerhard, J.},
  \bibinfo{year}{2013}.
\newblock \bibinfo{title}{Modern Computer Algebra (third edition)}.
\newblock \bibinfo{publisher}{Cambridge University Press}.
\newblock \DOIprefix\doi{10.1017/CBO9781139856065}.
%Type = Article
\bibitem[{Giesbrecht(1995)}]{Giesbrecht95}
\bibinfo{author}{Giesbrecht, M.}, \bibinfo{year}{1995}.
\newblock \bibinfo{title}{Nearly optimal algorithms for canonical matrix
  forms}.
\newblock \bibinfo{journal}{SIAM Journal on Computing} \bibinfo{volume}{24},
  \bibinfo{pages}{948--969}.
\newblock \DOIprefix\doi{10.1137/S0097539793252687}.
%Type = Inproceedings
\bibitem[{Giorgi et~al.(2003)Giorgi, Jeannerod and Villard}]{GiJeVi03}
\bibinfo{author}{Giorgi, P.}, \bibinfo{author}{Jeannerod, C.P.},
  \bibinfo{author}{Villard, G.}, \bibinfo{year}{2003}.
\newblock \bibinfo{title}{On the complexity of polynomial matrix computations},
  in: \bibinfo{booktitle}{ISSAC'03}, \bibinfo{publisher}{ACM}. pp.
  \bibinfo{pages}{135--142}.
\newblock \DOIprefix\doi{10.1145/860854.860889}.
%Type = Inproceedings
\bibitem[{Giorgi and Neiger(2018)}]{GioNei18}
\bibinfo{author}{Giorgi, P.}, \bibinfo{author}{Neiger, V.},
  \bibinfo{year}{2018}.
\newblock \bibinfo{title}{Certification of minimal approximant bases}, in:
  \bibinfo{booktitle}{ISSAC'18}, \bibinfo{publisher}{ACM}. pp.
  \bibinfo{pages}{167--174}.
\newblock \DOIprefix\doi{10.1145/3208976.3208991}.
%Type = Article
\bibitem[{Gupta et~al.(2012)Gupta, Sarkar, Storjohann and
  Valeriote}]{GuSaStVa12}
\bibinfo{author}{Gupta, S.}, \bibinfo{author}{Sarkar, S.},
  \bibinfo{author}{Storjohann, A.}, \bibinfo{author}{Valeriote, J.},
  \bibinfo{year}{2012}.
\newblock \bibinfo{title}{Triangular $x$-basis decompositions and
  derandomization of linear algebra algorithms over ${K}[x]$}.
\newblock \bibinfo{journal}{J. Symbolic Comput.} \bibinfo{volume}{47},
  \bibinfo{pages}{422--453}.
\newblock \DOIprefix\doi{10.1016/j.jsc.2011.09.006}.
%Type = Article
\bibitem[{Harvey et~al.(2017)Harvey, Van Der~Hoeven and
  Lecerf}]{HarveyVDHoevenLecerf2017}
\bibinfo{author}{Harvey, D.}, \bibinfo{author}{Van Der~Hoeven, J.},
  \bibinfo{author}{Lecerf, G.}, \bibinfo{year}{2017}.
\newblock \bibinfo{title}{Faster polynomial multiplication over finite fields}.
\newblock \bibinfo{journal}{J. ACM} \bibinfo{volume}{63}.
\newblock \DOIprefix\doi{10.1145/3005344}.
%Type = Article
\bibitem[{Ibarra et~al.(1982)Ibarra, Moran and Hui}]{IbaMorHui82}
\bibinfo{author}{Ibarra, O.H.}, \bibinfo{author}{Moran, S.},
  \bibinfo{author}{Hui, R.}, \bibinfo{year}{1982}.
\newblock \bibinfo{title}{A generalization of the fast {LUP} matrix
  decomposition algorithm and applications}.
\newblock \bibinfo{journal}{Journal of Algorithms} \bibinfo{volume}{3},
  \bibinfo{pages}{45--56}.
\newblock \DOIprefix\doi{10.1016/0196-6774(82)90007-4}.
%Type = Inproceedings
\bibitem[{Jeannerod et~al.(2016)Jeannerod, Neiger, Schost and
  Villard}]{JeNeScVi16}
\bibinfo{author}{Jeannerod, C.P.}, \bibinfo{author}{Neiger, V.},
  \bibinfo{author}{Schost, E.}, \bibinfo{author}{Villard, G.},
  \bibinfo{year}{2016}.
\newblock \bibinfo{title}{Fast computation of minimal interpolation bases in
  {Popov} form for arbitrary shifts}, in: \bibinfo{booktitle}{ISSAC'16},
  \bibinfo{publisher}{ACM}. pp. \bibinfo{pages}{295--302}.
\newblock \DOIprefix\doi{10.1145/2930889.2930928}.
%Type = Article
\bibitem[{Jeannerod et~al.(2017)Jeannerod, Neiger, Schost and
  Villard}]{JeNeScVi17}
\bibinfo{author}{Jeannerod, C.P.}, \bibinfo{author}{Neiger, V.},
  \bibinfo{author}{Schost, E.}, \bibinfo{author}{Villard, G.},
  \bibinfo{year}{2017}.
\newblock \bibinfo{title}{Computing minimal interpolation bases}.
\newblock \bibinfo{journal}{J. Symbolic Comput.} \bibinfo{volume}{83},
  \bibinfo{pages}{272--314}.
\newblock \DOIprefix\doi{10.1016/j.jsc.2016.11.015}.
%Type = Article
\bibitem[{Jeannerod et~al.(2020)Jeannerod, Neiger and Villard}]{JeaNeiVil20}
\bibinfo{author}{Jeannerod, C.P.}, \bibinfo{author}{Neiger, V.},
  \bibinfo{author}{Villard, G.}, \bibinfo{year}{2020}.
\newblock \bibinfo{title}{Fast computation of approximant bases in canonical
  form}.
\newblock \bibinfo{journal}{J. Symbolic Comput.} \bibinfo{volume}{98},
  \bibinfo{pages}{192--224}.
\newblock \DOIprefix\doi{10.1016/j.jsc.2019.07.011}.
%Type = Book
\bibitem[{Kailath(1980)}]{Kailath80}
\bibinfo{author}{Kailath, T.}, \bibinfo{year}{1980}.
\newblock \bibinfo{title}{{Linear Systems}}.
\newblock \bibinfo{publisher}{Prentice-Hall}.
%Type = Article
\bibitem[{Kaltofen and Villard(2005)}]{KaVi05}
\bibinfo{author}{Kaltofen, E.}, \bibinfo{author}{Villard, G.},
  \bibinfo{year}{2005}.
\newblock \bibinfo{title}{On the complexity of computing determinants}.
\newblock \bibinfo{journal}{Computational Complexity} \bibinfo{volume}{13},
  \bibinfo{pages}{91--130}.
\newblock \DOIprefix\doi{10.1007/s00037-004-0185-3}.
%Type = Article
\bibitem[{Keller-Gehrig(1985)}]{KelGeh85}
\bibinfo{author}{Keller-Gehrig, W.}, \bibinfo{year}{1985}.
\newblock \bibinfo{title}{Fast algorithms for the characteristic polynomial}.
\newblock \bibinfo{journal}{Theoretical Computer Science} \bibinfo{volume}{36},
  \bibinfo{pages}{309--317}.
\newblock \DOIprefix\doi{10.1016/0304-3975(85)90049-0}.
%Type = Article
\bibitem[{Labahn et~al.(2017)Labahn, Neiger and Zhou}]{LaNeZh17}
\bibinfo{author}{Labahn, G.}, \bibinfo{author}{Neiger, V.},
  \bibinfo{author}{Zhou, W.}, \bibinfo{year}{2017}.
\newblock \bibinfo{title}{Fast, deterministic computation of the {H}ermite
  normal form and determinant of a polynomial matrix}.
\newblock \bibinfo{journal}{J. Complexity} \bibinfo{volume}{42},
  \bibinfo{pages}{44--71}.
\newblock \DOIprefix\doi{10.1016/j.jco.2017.03.003}.
%Type = Inproceedings
\bibitem[{Le~Gall(2014)}]{LeGall14}
\bibinfo{author}{Le~Gall, F.}, \bibinfo{year}{2014}.
\newblock \bibinfo{title}{Powers of tensors and fast matrix multiplication},
  in: \bibinfo{booktitle}{ISSAC'14}, \bibinfo{publisher}{ACM}. pp.
  \bibinfo{pages}{296--303}.
\newblock \DOIprefix\doi{10.1145/2608628.2608664}.
%Type = Article
\bibitem[{{Le Verrier}(1840)}]{Leverrier1840}
\bibinfo{author}{{Le Verrier}, U.}, \bibinfo{year}{1840}.
\newblock \bibinfo{title}{Sur les variations s\'eculaires des \'el\'ements
  elliptiques des sept plant\`etes principales}.
\newblock \bibinfo{journal}{Journal des Math{\'e}matiques Pures et
  Appliqu{\'e}es} \bibinfo{volume}{5}, \bibinfo{pages}{220--254}.
%Type = Article
\bibitem[{Manthey and Helmke(2007)}]{MaHe07}
\bibinfo{author}{Manthey, W.}, \bibinfo{author}{Helmke, U.},
  \bibinfo{year}{2007}.
\newblock \bibinfo{title}{{Bruhat} canonical form for linear systems}.
\newblock \bibinfo{journal}{Linear Algebra Appl.} \bibinfo{volume}{425},
  \bibinfo{pages}{261--282}.
\newblock \DOIprefix\doi{10.1016/j.laa.2007.01.022}.
%Type = Article
\bibitem[{Mulders and Storjohann(2003)}]{MulSto03}
\bibinfo{author}{Mulders, T.}, \bibinfo{author}{Storjohann, A.},
  \bibinfo{year}{2003}.
\newblock \bibinfo{title}{On lattice reduction for polynomial matrices}.
\newblock \bibinfo{journal}{J. Symbolic Comput.} \bibinfo{volume}{35},
  \bibinfo{pages}{377--401}.
\newblock \DOIprefix\doi{10.1016/S0747-7171(02)00139-6}.
%Type = Inproceedings
\bibitem[{Neiger et~al.(2018)Neiger, Rosenkilde and Solomatov}]{NeiRosSol18}
\bibinfo{author}{Neiger, V.}, \bibinfo{author}{Rosenkilde, J.},
  \bibinfo{author}{Solomatov, G.}, \bibinfo{year}{2018}.
\newblock \bibinfo{title}{Computing {P}opov and {H}ermite {F}orms of
  {R}ectangular {P}olynomial {M}atrices}, in: \bibinfo{booktitle}{ISSAC'18},
  \bibinfo{publisher}{ACM}. pp. \bibinfo{pages}{295--302}.
\newblock \DOIprefix\doi{10.1145/3208976.3208988}.
%Type = Inproceedings
\bibitem[{Neiger and Vu(2017)}]{NeiVu17}
\bibinfo{author}{Neiger, V.}, \bibinfo{author}{Vu, T.X.}, \bibinfo{year}{2017}.
\newblock \bibinfo{title}{Computing canonical bases of modules of univariate
  relations}, in: \bibinfo{booktitle}{ISSAC'17}, \bibinfo{publisher}{ACM}. pp.
  \bibinfo{pages}{357--364}.
\newblock \DOIprefix\doi{10.1145/3087604.3087656}.
%Type = Inproceedings
\bibitem[{Pernet and Storjohann(2007)}]{PerSto07}
\bibinfo{author}{Pernet, C.}, \bibinfo{author}{Storjohann, A.},
  \bibinfo{year}{2007}.
\newblock \bibinfo{title}{Faster {A}lgorithms for the {C}haracteristic
  {P}olynomial}, in: \bibinfo{booktitle}{ISSAC'07}, \bibinfo{publisher}{ACM}.
  pp. \bibinfo{pages}{307--314}.
\newblock \DOIprefix\doi{10.1145/1277548.1277590}.
%Type = Article
\bibitem[{Popov(1972)}]{Popov72}
\bibinfo{author}{Popov, V.M.}, \bibinfo{year}{1972}.
\newblock \bibinfo{title}{Invariant description of linear, time-invariant
  controllable systems}.
\newblock \bibinfo{journal}{SIAM Journal on Control} \bibinfo{volume}{10},
  \bibinfo{pages}{252--264}.
\newblock \DOIprefix\doi{10.1137/0310020}.
%Type = Article
\bibitem[{Samuelson(1942)}]{Samuelson42}
\bibinfo{author}{Samuelson, P.A.}, \bibinfo{year}{1942}.
\newblock \bibinfo{title}{A method of determining explicitly the coefficients
  of the characteristic equation}.
\newblock \bibinfo{journal}{Annals of Mathematical Statistics}
  \bibinfo{volume}{13}, \bibinfo{pages}{424--429}.
%Type = Inproceedings
\bibitem[{Sarkar and Storjohann(2011)}]{SarSto11}
\bibinfo{author}{Sarkar, S.}, \bibinfo{author}{Storjohann, A.},
  \bibinfo{year}{2011}.
\newblock \bibinfo{title}{Normalization of row reduced matrices}, in:
  \bibinfo{booktitle}{ISSAC'11}, \bibinfo{publisher}{ACM}. pp.
  \bibinfo{pages}{297--304}.
\newblock \DOIprefix\doi{10.1145/1993886.1993931}.
%Type = Article
\bibitem[{Souriau(1948)}]{Souriau48}
\bibinfo{author}{Souriau, J.M.}, \bibinfo{year}{1948}.
\newblock \bibinfo{title}{Une m{\'e}thode pour la d{\'e}composition spectrale
  et l'inversion des matrices}.
\newblock \bibinfo{journal}{Comptes-Rendus de l'Acad{\'e}mie des Sciences}
  \bibinfo{volume}{227}, \bibinfo{pages}{1010--1011}.
%Type = Inproceedings
\bibitem[{{Storjohann}(2001)}]{Storjohann01}
\bibinfo{author}{{Storjohann}, A.}, \bibinfo{year}{2001}.
\newblock \bibinfo{title}{Deterministic computation of the frobenius form}, in:
  \bibinfo{booktitle}{Proceedings 42nd IEEE Symposium on Foundations of
  Computer Science}, pp. \bibinfo{pages}{368--377}.
\newblock \DOIprefix\doi{10.1109/SFCS.2001.959911}.
%Type = Article
\bibitem[{Storjohann(2003)}]{Storjohann03}
\bibinfo{author}{Storjohann, A.}, \bibinfo{year}{2003}.
\newblock \bibinfo{title}{High-order lifting and integrality certification}.
\newblock \bibinfo{journal}{J. Symbolic Comput.} \bibinfo{volume}{36},
  \bibinfo{pages}{613--648}.
\newblock \DOIprefix\doi{10.1016/S0747-7171(03)00097-X}.
%Type = Inproceedings
\bibitem[{Storjohann(2006)}]{Storjohann06}
\bibinfo{author}{Storjohann, A.}, \bibinfo{year}{2006}.
\newblock \bibinfo{title}{Notes on computing minimal approximant bases}, in:
  \bibinfo{booktitle}{Challenges in Symbolic Computation Software}.
\newblock \URLprefix \url{http://drops.dagstuhl.de/opus/volltexte/2006/776}.
%Type = Article
\bibitem[{Strassen(1969)}]{Strassen1969}
\bibinfo{author}{Strassen, V.}, \bibinfo{year}{1969}.
\newblock \bibinfo{title}{Gaussian elimination is not optimal}.
\newblock \bibinfo{journal}{Numer. Math.} \bibinfo{volume}{13},
  \bibinfo{pages}{354--356}.
\newblock \DOIprefix\doi{10.1007/BF02165411}.
%Type = Misc
\bibitem[{{The~FFLAS-FFPACK~Group}(2019)}]{fflas-ffpack}
\bibinfo{author}{{The~FFLAS-FFPACK~Group}}, \bibinfo{year}{2019}.
\newblock \bibinfo{title}{{FFLAS-FFPACK}: {F}inite {F}ield {L}inear {A}lgebra
  {S}ubroutines / {P}ackage, version 2.4.3}.
\newblock
  \bibinfo{howpublished}{\url{http://github.com/linbox-team/fflas-ffpack}}.
%Type = Misc
\bibitem[{{The LinBox Group}(2019)}]{Linbox}
\bibinfo{author}{{The LinBox Group}}, \bibinfo{year}{2019}.
\newblock \bibinfo{title}{Linbox: Linear algebra over black-box matrices,
  version 1.6.3}.
\newblock \bibinfo{howpublished}{\url{https://github.com/linbox-team/linbox/}}.
%Type = Article
\bibitem[{Toom(1963)}]{Toom1963}
\bibinfo{author}{Toom, A.L.}, \bibinfo{year}{1963}.
\newblock \bibinfo{title}{{The complexity of a scheme of functional elements
  realizing the multiplication of integers}}.
\newblock \bibinfo{journal}{Soviet Mathematics Doklady} \bibinfo{volume}{3},
  \bibinfo{pages}{714--716}.
%Type = Article
\bibitem[{Van~Barel and Bultheel(1992)}]{BarBul92}
\bibinfo{author}{Van~Barel, M.}, \bibinfo{author}{Bultheel, A.},
  \bibinfo{year}{1992}.
\newblock \bibinfo{title}{A general module theoretic framework for vector
  {M-Pad\'e} and matrix rational interpolation}.
\newblock \bibinfo{journal}{Numer. Algorithms} \bibinfo{volume}{3},
  \bibinfo{pages}{451--462}.
\newblock \DOIprefix\doi{10.1007/BF02141952}.
%Type = Book
\bibitem[{Wolovich(1974)}]{Wolovich74}
\bibinfo{author}{Wolovich, W.A.}, \bibinfo{year}{1974}.
\newblock \bibinfo{title}{Linear Multivariable Systems}.
  volume~\bibinfo{volume}{11} of \textit{\bibinfo{series}{Applied Mathematical
  Sciences}}.
\newblock \bibinfo{publisher}{Springer-Verlag New-York}.
\newblock \DOIprefix\doi{10.1007/978-1-4612-6392-0}.
%Type = Phdthesis
\bibitem[{Zhou(2012)}]{Zhou12}
\bibinfo{author}{Zhou, W.}, \bibinfo{year}{2012}.
\newblock \bibinfo{title}{Fast Order Basis and Kernel Basis Computation and
  Related Problems}.
\newblock Ph.D. thesis. University of Waterloo.
\newblock \URLprefix \url{http://hdl.handle.net/10012/7326}.
%Type = Inproceedings
\bibitem[{Zhou and Labahn(2009)}]{ZhoLab09}
\bibinfo{author}{Zhou, W.}, \bibinfo{author}{Labahn, G.}, \bibinfo{year}{2009}.
\newblock \bibinfo{title}{Efficient computation of order bases}, in:
  \bibinfo{booktitle}{ISSAC'09}, \bibinfo{publisher}{ACM}. pp.
  \bibinfo{pages}{375--382}.
\newblock \DOIprefix\doi{10.1145/1576702.1576753}.
%Type = Article
\bibitem[{Zhou and Labahn(2012)}]{ZhoLab12}
\bibinfo{author}{Zhou, W.}, \bibinfo{author}{Labahn, G.}, \bibinfo{year}{2012}.
\newblock \bibinfo{title}{Efficient algorithms for order basis computation}.
\newblock \bibinfo{journal}{J. Symbolic Comput.} \bibinfo{volume}{47},
  \bibinfo{pages}{793--819}.
\newblock \DOIprefix\doi{10.1016/j.jsc.2011.12.009}.
%Type = Inproceedings
\bibitem[{Zhou and Labahn(2013)}]{ZhoLab13}
\bibinfo{author}{Zhou, W.}, \bibinfo{author}{Labahn, G.}, \bibinfo{year}{2013}.
\newblock \bibinfo{title}{Computing column bases of polynomial matrices}, in:
  \bibinfo{booktitle}{ISSAC'13}, \bibinfo{publisher}{ACM}. pp.
  \bibinfo{pages}{379--386}.
\newblock \DOIprefix\doi{10.1145/2465506.2465947}.
%Type = Inproceedings
\bibitem[{Zhou and Labahn(2014)}]{ZhoLab14}
\bibinfo{author}{Zhou, W.}, \bibinfo{author}{Labahn, G.}, \bibinfo{year}{2014}.
\newblock \bibinfo{title}{Unimodular completion of polynomial matrices}, in:
  \bibinfo{booktitle}{ISSAC'14}, \bibinfo{publisher}{ACM}. pp.
  \bibinfo{pages}{413--420}.
\newblock \DOIprefix\doi{10.1145/2608628.2608640}.
%Type = Inproceedings
\bibitem[{Zhou et~al.(2012)Zhou, Labahn and Storjohann}]{ZhLaSt12}
\bibinfo{author}{Zhou, W.}, \bibinfo{author}{Labahn, G.},
  \bibinfo{author}{Storjohann, A.}, \bibinfo{year}{2012}.
\newblock \bibinfo{title}{Computing minimal nullspace bases}, in:
  \bibinfo{booktitle}{ISSAC'12}, \bibinfo{publisher}{ACM}. pp.
  \bibinfo{pages}{366--373}.
\newblock \DOIprefix\doi{10.1145/2442829.2442881}.

\end{thebibliography}

\end{document}